\documentclass[a4paper,USenglish,cleveref, autoref, thm-restate, numberwithinsect]{lipics-v2021}

\pdfoutput=1 %uncomment to ensure pdflatex processing (mandatatory e.g. to submit to arXiv)

\hideLIPIcs  %uncomment to remove references to LIPIcs series (logo, DOI, ...), e.g. when preparing a pre-final version to be uploaded to arXiv or another public repository
\nolinenumbers

\usepackage{microtype}%if unwanted, comment out or use option "draft"
\usepackage[T5,T1]{fontenc}
\usepackage[utf8]{inputenc}
\usepackage{xspace}
\newcommand\Nguyen{Nguy{\~{\^e}}n\xspace}

\usepackage{amsmath,amsthm}
\usepackage{stmaryrd}
\usepackage{mathtools}
\usepackage{graphicx}
\usepackage{csquotes}
\usepackage{cmll}
\usepackage{tikz-cd}
\usetikzlibrary{decorations.pathreplacing, calc}

\usepackage{quiver}

\title{Tree transducers of linear size-to-height increase}
\subtitle{(and the additive conjunction of linear logic)}

\titlerunning{Tree transducers of linear size-to-height increase} %TODO optional, please use if title is longer than one line

\author{Luc Dartois}{Université Marie et Louis Pasteur, CNRS, institut FEMTO-ST, F-25000 Besançon, France}{luc.dartois@femto-st.fr}{}{}
\author{Lê Thành Dũng (Tito) \Nguyen}{LIS, CNRS \& Aix-Marseille University, France \and \url{https://nguyentito.eu/}}{nltd@nguyentito.eu}{https://orcid.org/0000-0002-6900-5577}{}
\author{Charles Peyrat}{Université Paris-Est Créteil, France}{}{}{}
\authorrunning{L.~Dartois, L.~T.~D.~\Nguyen \& C.~Peyrat} %TODO mandatory. First: Use abbreviated first/middle names. Second (only in severe cases): Use first author plus 'et al.'

\Copyright{Jane Open Access and Joan R. Public} %TODO mandatory, please use full first names. LIPIcs license is "CC-BY";  http://creativecommons.org/licenses/by/3.0/


 %TODO mandatory; please add comma-separated list of keywords

\category{} %optional, e.g. invited paper

\relatedversion{} %optional, e.g. full version hosted on arXiv, HAL, or other respository/website
\EventEditors{John Q. Open and Joan R. Access}
\EventNoEds{2}
\EventLongTitle{42nd Conference on Very Important Topics (CVIT 2016)}
\EventShortTitle{CVIT 2016}
\EventAcronym{CVIT}
\EventYear{2016}
\EventDate{December 24--27, 2016}
\EventLocation{Little Whinging, United Kingdom}
\EventLogo{}
\SeriesVolume{42}
\ArticleNo{23}
\newcommand\IM{\mathrm{IM}}
\newcommand\OM{\mathrm{OM}}
\newcommand\ttL{\mathtt{L}}
\newcommand\ttR{\mathtt{R}}
\newcommand\exa{\mathtt{exa}}

\newcommand{\subtree}[2]{#1{\restriction_{#2}}}
\newcommand{\dirs}[1]{\mathsf{D}_{#1}}
\newcommand{\db}[1]{[#1]}
\newcommand{\qq}[1]{\langle #1 \rangle}
\newcommand{\pp}[1]{\langle #1 \rangle}
\newcommand{\bet}[1]{\lceil #1 \rceil}
\newcommand{\bec}[1]{\lfloor #1 \rfloor}
\newcommand{\bek}[1]{\llceil #1 \rrfloor}
\newcommand{\N}{\mathbb{N}}

\DeclareMathOperator{\Stitch}{Stitch}
\DeclareMathOperator{\DecodeH}{DecodeH}
\DeclareMathOperator{\DecodeN}{DecodeN}
\DeclareMathOperator{\Point}{Point}
\DeclareMathOperator{\memToStitch}{MemToStitch}
\DeclareMathOperator{\rank}{rk}
\DeclareMathOperator{\Dom}{Dom}

\DeclareMathOperator{\qzero}{qzero}

\DeclareMathOperator{\Img}{Im}
\DeclareMathOperator{\Link}{Link}
\DeclareMathOperator{\downLink}{DownLink}
\DeclareMathOperator{\upLink}{UpLink}
\DeclareMathOperator{\downNewLink}{DownNewLink}
\DeclareMathOperator{\orgn}{o}

\DeclareMathOperator{\child}{child}

\DeclareMathOperator{\height}{height}

\newcommand\cA{\mathcal{A}}
\newcommand\cB{\mathcal{B}}
\newcommand\cG{\mathcal{G}}
\newcommand\cH{\mathcal{H}}
\newcommand{\cI}{\mathcal{I}}
\newcommand\cM{\mathcal{M}}
\newcommand\cN{\mathcal{N}}
\newcommand\cR{\mathcal{R}}
\newcommand\cT{\mathcal{T}}

\newcommand{\parent}[1]{#1{\uparrow}}
\DeclareMathOperator\lab{lab}

  \usepackage{todonotes}
  \newcommand{\tito}[1]{}
  \newcounter{todocounter}

\newcommand{\LA}{\mathrm{LA}}
\newcommand{\twhm}{\mathsf{THM}}

\newcommand{\THMr}{\mathsf{THM}^{\mathsf{R}}}
\newcommand{\THM}{\mathsf{THM}}
\newcommand{\uTHMr}{\mathsf{uTHM}^{\mathsf{R}}}
\newcommand{\uTHM}{\mathsf{uTHM}}
\newcommand{\mtt}{\mathsf{MTT}}
\newcommand{\twt}{\mathsf{TWT}}
\newcommand\lsi{\mathsf{LSI}}
\newcommand\lshi{\mathsf{LSHI}}
\newcommand\lhi{\mathsf{LHI}}
\newcommand{\lshii}{\lshi\text{-}}
\newcommand{\lhii}{\lhi\text{-}}
\newcommand{\lshimtt}{\lshii\mtt}
\newcommand\sem[1]{\llbracket #1 \rrbracket}

\newcommand{\msosi}{\mathsf{MSO}\text{-}\mathsf{SI}}
\newcommand{\mso}{\mathsf{MSO}}
\newcommand{\maxrank}[1]{\mathsf{R}_{#1}}
\newcommand{\rankK}[2]{#1^{(#2)}}

\newcommand\RHS{\mathsf{RHS}}

\newcommand\Conf{\mathsf{Conf}}
\newcommand\Redex{\mathsf{Redex}}
\DeclareMathOperator\brawo{brawo}
\newcommand{\EncPeb}{\mathrm{EncPeb}}

\newcommand\Strat{\mathrm{Strat}}

\NewDocumentCommand\Tree{m o}{%
  \IfNoValueTF{#2}{%true
    \mathsf{T}_{#1}
  }%
  {%false
    \mathsf{T}_{#1}(#2)
  }
}
\newcommand{\tto}[1][]{\mathrel{
  \vphantom{\xrightarrow{#1}}
  \smash{\xrightarrow{#1}}
  \vphantom{\to}^*}
}

\newcommand\simulambda{\leftrightsquigarrow}

\begin{document}

% https://github.com/borisveytsman/acmart/issues/399
\AddToHook{env/proposition/begin}{\crefalias{theorem}{proposition}}
\AddToHook{env/lemma/begin}{\crefalias{theorem}{lemma}}
\AddToHook{env/definition/begin}{\crefalias{theorem}{definition}}
\AddToHook{env/corollary/begin}{\crefalias{theorem}{corollary}}
\AddToHook{env/conjecture/begin}{\crefalias{theorem}{conjecture}}
\AddToHook{env/example/begin}{\crefalias{theorem}{example}}

\maketitle

\begin{abstract}
  We investigate a natural generalization to trees of Hennie machines, a known automaton model for regular string functions. Tree-to-tree Hennie machines are tree-walking tree transducers with the ability to rewrite the node labels of their input tree, subject to a bounded visit restriction. Interestingly, they do not merely compute regular tree functions (i.e.\ $\mso$ transductions), but a larger class of functions with linear size-to-height increase (LSHI).

  We prove that this class sits between LSHI macro tree transducers (MTTs) and $\mso$ set interpretations. To argue for its robustness, we show that it is closed under precomposition (resp.\ postcomposition) by MTTs of linear size (resp.\ height) increase. As a consequence, it contains the entire composition hierarchy of MTTs of linear height increase; we also prove that this composition hierarchy is strict.

  Finally, we give an alternative characterization of this function class based on a $\lambda$-calculus with linear types. The key difference with similar characterizations of $\mso$ transductions is the use of additive tuples in the encoding of output trees. Our equivalence proof, using game semantics / geometry of interaction, is heavily inspired by an analogous result on higher-order recursion schemes.
\end{abstract}

\section{Introduction}%
\label{sec:introduction}

\emph{Macro tree transducers} (MTTs) are an automata model for tree transformations first introduced by Engelfriet and Vogler in the 1980s~\cite{Macro,EngelfrietPushdownMacro}. There have been many developments concerning MTTs since then;
one starting point of this paper is the recent:
\begin{theorem}[{Gallot, Maneth, Nakano \& Peyrat~\cite{GMNP24}}]%
  \label{thm:mtt-lhi}
  Given a \emph{macro tree transducer} computing a function $f$ on ranked trees, one can decide whether it has:
  \begin{description}
    \item[linear height increase (LHI):] $\height(f(t)) = O(\height(t))$
    \item[linear size-to-height increase (LSHI):] $\height(f(t)) = O(|t|)$
  \end{description}
\end{theorem}
Compare this to an important result from a quarter century ago:
\begin{theorem}[{Engelfriet \& Maneth~\cite{MacroMSOLinear}}]\label{thm:mtt-lsi}
  Given an $\mtt$ computing a function $f$ on ranked trees, one can decide whether it is of \emph{linear size increase (LSI)}.
  Furthermore, if it is, then one can translate it to an \emph{Monadic Second-Order transduction} that defines $f$. (Conversely, every $\mso$ transduction can be translated into an $\mtt$.)
\end{theorem}

This seems more informative than the statement of the analogous \Cref{thm:mtt-lhi}: in addition to decidability, it gives us an equivalence between $\lsi$-$\mtt$s and a completely different formalism for describing functions from trees to trees, based on $\mso$ logic. This equivalence is a consequence of the strategy for decidability: it comes from boundedness properties on the runs of (suitably normalized) $\mtt$s that characterize the $\lsi$ property.\footnote{Reducing various quantitative decision problems to boundedness questions is a general idea with wide-ranging applicability in automata theory, see e.g.~\cite{CostFunctionsLICS,AnandSSZ24}. Note also that using a quite different approach, Gallot, Lhote and \Nguyen~\cite{StructPoly} gave alternative proofs of \Cref{thm:mtt-lsi} and of a generalization of the $\lshi$ half of \Cref{thm:mtt-lhi}.}

Similar boundedness properties arise in the decision procedure for \Cref{thm:mtt-lhi}. Hence the question: what consequences can we draw concerning MTTs of linear (size-to-)height increase? We show that they can be translated to \emph{$\mso$ set interpretations} ($\msosi$), an extension of $\mso$ transductions from the 2000s~\cite{ColcombetL07} recently studied\footnote{Let us also mention the work of Bojańczyk et al.~\cite{msoInterpretations,Bojanczyk23} on string-to-string \emph{$\mso$ interpretations}, an intermediate case between transductions and set interpretations. They are one of the equivalent characterizations of \emph{polyregular functions}.} by Lhote et al.~\cite{Lex,StructPoly}. More than that: to prove this comparison in expressive power, we find it enlightening to consider an intermediate model of independent interest: 
\[ \lshimtt \subset \text{\textbf{tree-to-tree Hennie machines}} \subset \msosi \]
In fact, this new machine model is our actual main topic here. Recall that a \emph{Hennie machine}~\cite[Section~7]{EngelfrietHoogeboom} is a Turing machine:
\begin{itemize}
  \item with a single read/write input tape and a write-only left-to-right output tape;
  \item whose tape head must stay within the original bounds of the input string (that is, the only usable memory is the space initially occupied by the input);\footnote{Actually, this condition does not appear in Hennie's original work~\cite[Section~II.C]{Hennie65} on language recognition by bounded-visit Turing machines, nor in its subsequent extension by Rajlich~\cite{Rajlich75} to string outputs. However the name \enquote{Hennie machine} seems be consistently used to refer to a machine model that satisfies this restriction on the tape head, starting with its introduction by Engelfriet and Hoogeboom~\cite[\S7]{EngelfrietHoogeboom} (for later examples, cf.~\cite{GuillonPPP23,GuillonPPP22,GuillonP19}).}
  \item visiting each cell of the tape \emph{a bounded number of times}.\footnote{Guillon, Pighizzini, Prigioniero and Průša~\cite{GuillonPPP23,GuillonPPP22,GuillonP19} define Hennie machines (as language acceptors) by a linear-time restriction instead of a bounded-visit restriction, but these are effectively equivalent, see~\cite[Lemma~1]{GuillonPPP23} or~\cite[proof of Theorem~3]{Hennie65}.}
\end{itemize}
In other words, it is a bounded-visit two-way transducer extended with the ability to use the positions of input letters as mutable memory cells. In this informal description, replacing the two-way string transducer by its usual generalization to trees, namely the \emph{tree-walking tree transducer} (see e.g.~\cite[Chapter~8]{courcellebook}), results in the notion of \emph{tree-to-tree Hennie machine} ($\THM$).
\begin{claim}\label{clm:lshi}
  $\THM$s have linear size-to-height increase.
\end{claim}
\begin{proof}[Proof idea]
  The bounded visit property forces $\THM$s to terminate in linear time. In the tree-to-tree case, the branches are produced in parallel by forking processes, each taking linear time. We give a more rigorous proof in \Cref{sec:bounded-visit}.
\end{proof}

\begin{figure}
\begin{footnotesize}
\begin{tikzpicture}
\draw[fill,color=red!20,draw=none] (-9,1.5) -- (7,1.5) -- (7,0.5) -- (-9,.5);
\draw[fill,color=blue!20,draw=none] (-9,0.5) -- (7,0.5) -- (7,-0.4) -- (-9,-.4);
\draw[fill,color=green!20,draw=none] (-9,-0.4) -- (7,-.4) -- (7,-1.3) -- (-9,-1.3);

\node (thm) at (0,0) {\textbf{tree-to-tree Hennie Machines ($\THM$)}};
\node (lshimtt) at (-6,0) {$\lshimtt$};
\node (msosi) at (6,0) {$\msosi$};
\node (actors) at (3,-1) {actor-based tree transducers};
\node (lambda) at (-3,-1) {affine $\lambda$-transducers with additive branching};
\node (thm71) at (0,-.6) {\Cref{thm:lambda-actor-hennie}};

\node (lhicomp) at (-2,1.3) {$\lhii\mtt\circ\THM$};
\node (compmso) at (2,1.3) {$\THM\circ\mso\text{ transductions} $};

%\node (compprecomp) at (0,0.6) {\Cref{thm:closure-props}};
\node (thmpostcom) at (-2.33,0.8) {\Cref{cor:lhimtt-comp-thm}};
\node (thmprecomp) at (2.4,0.8) {\Cref{thm:precomp-msot}};

\path[<->,>=stealth] (thm) edge node{} (lhicomp);
\path[<->,>=stealth] (thm) edge node {} (compmso);

\path[->,>=stealth] (lshimtt) edge node[above] {\Cref{cor:lshimtt-hennie}} (thm);
\path[->,>=stealth] (thm) edge node[above] {\Cref{thm:THMtoMSOSI}} (msosi);
\path[->,>=stealth] (thm) edge node {} (lambda);
\path[->,>=stealth] (lambda) edge node {} (actors);
\path[->,>=stealth] (actors) edge node {} (thm);

\node[align=right] at (-8,0) { Closure prop. \\[1.8em]
								Inclusions \\[1.8em]
								Equivalences};

\end{tikzpicture}
\end{footnotesize}
\caption{Summary of the main results. An arrow indicates inclusion of function classes.}\label{fig:results}
\end{figure}
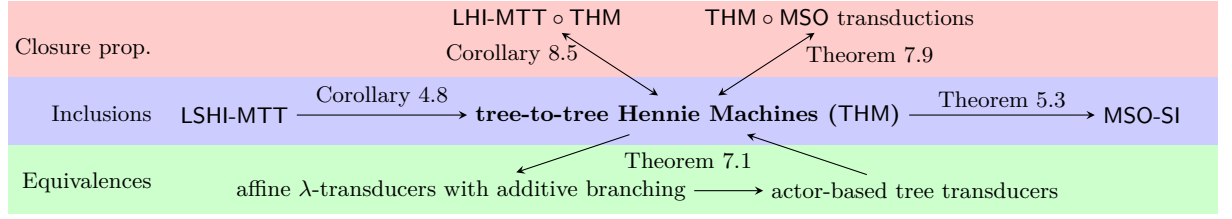

Just as recent works on polyregular functions~\cite{PolyregSurvey,Kiefer24} seem to answer \enquote{what string-to-string functions of polynomial size increase can be computed by \enquote{reasonable} transducers}, we envision our study of tree-to-tree Hennie machines as a contribution towards
understanding \enquote{what $\lshi$ functions can be computed by \enquote{reasonable} transducers}.
The class of functions computed by $\THM$s does not seem to coincide with any previously defined formalism.\footnote{Though in the tree-to-string case, it collapses to $\mso$ transductions (\Cref{cor:string-msot}). In the string-to-string case this was one of Engelfriet and Hoogeboom's main results on Hennie machines~\cite[Theorem~26]{EngelfrietHoogeboom}.}
Yet we argue that it is robust, by exhibiting closure properties and alternative characterizations. These results are summarized in \Cref{fig:results}; the rest of this introduction discusses them further.

\subparagraph{Closure properties.}

We show that the functions computed by~$\THM$
\begin{enumerate}
  \item are \emph{regularity reflecting}, i.e.\ the inverse image of a regular tree language is effectively regular (\Cref{cor:reg-preserving});
  \item are closed under precomposition by $\mso$ tree transductions (i.e.\ macro tree transducers of linear size increase);
  \item are closed under postcomposition by $\lhii\mtt$s.
\end{enumerate}
As a consequence, $\THM$s subsume the composition hierarchy of $\lhii\mtt$s, whose strictness we also show (\Cref{cor:lhimtt-strict}).

The proof of the third item relies on characteristics of $\lhii\mtt$s uncovered in Gallot et al.'s proof of \Cref{thm:mtt-lhi}, as discussed at the beginning. We relate these characteristics to a \enquote{narrow-visit} condition on $\THM$s. Narrow-visit implies $\mathsf{LHI}$, and a large part of the proof consists in showing:
$[\text{narrow-visit}\ \THM] \circ \THM = \THM$ (\Cref{thm:narrow-postcomp}).

For the first two items above, it is convenient to use our alternative characterization, due to its good compositional properties.

\subparagraph{$\lambda$-calculus characterization.}

We prove that $\THM$s are expressively equivalent to \emph{affine\footnote{In the $\lambda$-calculus, a linear (resp.\ affine) function uses its argument exactly (resp.\ at most) once. Affineness is closer to the \enquote{single use restrictions} in automata theory.} $\lambda$-transducers with additive branching}.
%a
This is a formalism in the style of linear higher-order tree transducers~\cite{LambdaTransducer,nguyen_et_al:LIPIcs.STACS.2025.68} or \enquote{implicit automata}~\cite{iatlc1,titoPhD,entics:14804}. Compared to these works, we extend the type system of these affine $\lambda$-transducers with the \emph{additive conjunction}~`$\with$' of linear logic, and use it\footnote{An idea originating in the work of Clairambault, Grellois and Murawski~\cite{LHORS}.} in the encoding of output trees to group the children of a node. Intuitively, this means that the affineness restrictions apply separately for each output branch (just like the bounded-visit restriction of our Hennie machines) rather than globally.

This idea is natural in the context of the previous work of Clairambault and Murawski~\cite{MAHORS}. They prove that two formalisms generating infinite trees are expressively equivalent: one based on affine $\lambda$-calculus and using additive branching, the other based on automata whose infinite memory is subject to a bounded-visit condition. Our own proof is a direct adaptation of theirs.
\begin{itemize}
  \item To compile $\lambda$-transducers into $\THM$s, we reuse a form of finite-memory geometry of interaction\footnote{\label{ftn:goi}The name \enquote{Geometry of Interaction} refers to a family of semantics (both operational and denotational) of linear logic and $\lambda$-calculi, closely related to game semantics.}  developed in~\cite{MAHORS}.
This involves a detour through yet another model of computation, based on communicating finite-state processes that we call \enquote{actors}.\footnote{They were originally called \enquote{transducers}~\cite[\S3.2]{MAHORS}, but we avoid this name here to avoid confusion. The name \enquote{actor} alludes both to the geometry of inter\emph{action} and to the more general \enquote{Actor model} of message-passing computation~\cite{Actors,ActorsLL,KosterCM16}.}
  \item The converse translation is more technically elementary, but of broader significance in the grand scheme of things: it gives us regularity reflection. 
        In a paper closely related to our work~\cite{ExpReg}, Colcombet, Lhote and Ohlmann leverage this result to prove that string-to-string $\msosi$ are regularity-reflecting --- with the side effect of settling an old open problem on automatic structures, cf.~\cite[\S1]{Lex}. One initial impetus for our work, actually, was to provide this piece of the proof of~\cite{ExpReg}. The use of \enquote{yield-Hennie    machines} in~\cite{ExpReg} led us to introduce and study tree-to-tree Hennie machines more generally, with an independent motivation coming from the line of work on $\mtt$s described previously.
\end{itemize}

\begin{remark}\label{rem:lambda-reg}
  Without the affineness constraint, simply typed \mbox{$\lambda$-transducers} subsume virtually all transducer models (for both strings and trees) from the literature, see e.g.~\cite[\S1.4.1]{titoPhD}.
  
  Regularity reflection for $\lambda$-transducers (cf.~\cite[\S5.3.1]{GallotPhD} or~\cite[Corollary~1.4.1]{titoPhD}\footnote{Corollary of Hillebrand and Kanellakis's~\cite[Theorem~3.4]{HillebrandKanellakis}, which is the origin of Nguy\~{\^e}n and Pradic's aforementioned \enquote{implicit automata} research programme.}) is a straightforward application of the naive set-theoretic semantics of the simply typed $\lambda$-calculus: if we interpret the base type as a finite set of states, then all higher-order types are also interpreted as finite sets, and $\lambda$-terms are interpreted as elements of these sets. This technically simple yet conceptually deep idea underlies several connections between automata theory and $\lambda$-calculus, from the use of finitary denotational semantics in higher-order model checking~\cite{Walukiewicz16,GrelloisPhD} to the theory of regular languages of $\lambda$-terms~\cite{Salvati15,HOParity,moreau:tel-05428993}. Let us note that Clairambault and Murawski's aforementioned~\cite{MAHORS} comes from a line of work motivated by higher-order model checking~\cite{LHORS,LiMO22}.
\end{remark}

\subparagraph{Plan of the paper.}

In \Cref{sec:notations} we fix some notations and recall some generalities on tree automata/transducers.
In \Cref{sec:hennie} we introduce unrestricted tree-to-tree Hennie machines, the bounded-visit restriction (defining $\THM$s) and the narrow-visit restriction, as well as some features extending $\THM$s. The three sections that follow are independent; we translate:
\begin{itemize}
  \item $\lshii\mtt$s to $\THM$s in \S\ref{sec:mtt};
  \item $\THM$s to $\mso$ set interpretations in \S\ref{sec:msosi};
  \item actor-based tree transducers to $\THM$s in \S\ref{sec:actors}.
\end{itemize}
In \S\ref{sec:lambda} we use actors to relate $\THM$s and $\lambda$-transducers, and deduce the closure by precomposition by $\mso$ transductions. In \S\ref{sec:postcomp} we look at postcomposition properties: we prove regularity reflection (via \mbox{$\lambda$-transducers}) and postcomposition by $\lhii\mtt$. Finally, in \S\ref{sec:lhimtt-comp} we study the $\lhii\mtt$ composition hierarchy.

\section{Notations \& preliminaries}\label{sec:notations}

A \emph{ranked set} is a set $\Sigma$ endowed with a function $\rank \colon \Sigma \rightarrow \N$. A finite ranked set is called a \emph{ranked alphabet}. We use the notations $\rankK{\Sigma}{k} = \{\sigma\in\Sigma \mid \rank(\sigma)=k\}$ and $\maxrank{\Sigma} = \max(\rank(\Sigma))$.
We use $\db{k}$ as a shorthand for the set $\{1,\ldots,k\}$.

A \emph{tree} $t$ over a ranked set $\Sigma$ is defined by:
\begin{itemize}
  \item a nonempty prefix-closed domain $\Dom(t) \subset \db{\maxrank{\Sigma}}^*$, whose elements are called the \emph{nodes} of $t$;
  \item a map $\lab_t \colon \Dom(t)\to \Sigma$, the \emph{labeling} of $t$, such that
        \[ \forall u\in\Dom(t),\quad \{i \in \N \mid ui \in \Dom(t) \} = \db{\rank(\lab_t(u))} \]
\end{itemize}
Thus, we also refer to $\rank(\lab_t(u))$ as the \emph{rank of the node $u$}.

We write $\Tree{\Sigma}$ for the set of trees over $\Sigma$, and when $\Sigma \cap Y = \varnothing$,
\[ \Tree{\Sigma}[Y] = \Tree{\Sigma \cup Y}\quad \text{with}\ \rank(y) = 0\; \forall y \in Y \]
In other words, the elements of $\Tree{\Sigma}[Y]$ are like trees over $\Sigma$ except that some leaves may have labels from $Y$.

The axioms entail that $\Dom(t)$ contains the empty sequence~$\varepsilon$; it is called the \emph{root} of $t$. The \emph{leaves} are the rank 0 nodes, or equivalently, the maximal nodes for the prefix ordering. This ordering is also called the \emph{ancestor} relation on nodes, and its converse is the \emph{descendant} relation.
For $u \in \Dom(t) \setminus \{\varepsilon\}$, we write $\parent{u}$ for the \emph{parent} of $u$, i.e.\ the unique node for which $\exists i \in \db{\maxrank{\Sigma}} : u = (\parent{u})i$.

We write $\dirs{k} = \{ \uparrow \} \cup\db{k}$ for the set of directions we can move in from a node of arity $k$ in a tree, where $ \uparrow $ represents the upward direction towards the parent. According to the previous notations, if $u \neq \varepsilon$ is a rank $k$ node of a tree $t$ and $d \in \dirs{k}$, then $ud$ defines a node in $t$ (one of the neighbors in $u$).

A \emph{partial branch} $b$ of $t$ is the downward closure of a node $s$ --- its \emph{endpoint} --- in $\Dom(t)$ with respect to the prefix order.
 A branch, or \textit{total branch} for disambiguation purposes, is a partial branch whose endpoint is a leaf.

The \emph{size} $|t|$ of a tree $t$ is the cardinality of its domain.
Its \emph{height} $\height(t)$ is the maximum size of its branches minus one.

For a node $s$ in a tree $t$, the \emph{subtree of $t$ rooted at $s$} is defined by:
\begin{align*}
  \Dom(\subtree{t}{s}) &= \{u \in \db{\maxrank{\Sigma}}^{*} \mid su \in \Dom(t)\}\\
  \lab_{\subtree{t}{s}}(u) &= \lab_{t}(su)
\end{align*}
The \emph{substitution of $s$ by $t'$ in $t$} is obtained from $t$ by replacing the subtree rooted at $s$ by $t'$ --- formally:
\begin{align*}
  \Dom(t[s \leftarrow t']) &= (\Dom(t)\setminus \{s\}\cdot [\maxrank{\Sigma}]^{*})\cup \{s\}\cdot \Dom(r) \\
  \lab_{t[s \leftarrow t']}(u) &= \begin{cases}
    \lab_{t'}(v) &\text{if}\ u=sv\ \text{for some}\ v\in\Dom(r) \\
    \lab_t(u) &\text{otherwise}
  \end{cases}
\end{align*}
In this paper, it is often (but not always) the case that in such a substitution, $s$ is a leaf; we then speak of a \emph{first-order substitution}.

For two families $\vec\sigma = (\sigma_{i})$ of distinct letters and $\vec{t}' = (t'_{i})$ of trees with the same indices, we denote by $t[\vec\sigma \leftarrow \vec t']$ or $t[\sigma_{i} \leftarrow t'_{i}\; \forall i]$ the simultaneous substitution of all nodes of $t$ with label $\sigma_{i}$ of $t$ by $t'_{i}$ (recall that this replaces the entire subtrees rooted at these nodes). In the case of a singleton family we write $t[\sigma \leftarrow t']$.

For $\sigma \in \Sigma$ of rank $k$ and $t_{1},\dots,t_{k} \in \Tree{\Sigma}$, we write $\sigma(t_{1},\dots,t_{k})$ for the unique tree whose root has label $\sigma$ and whose subtree rooted at the node $i$ is $t_{i}$ for all $i \in \db{k}$. Every tree can be uniquely decomposed this way; since $|t_{i}| < |\sigma(t_{1},\dots,t_{k})|$ this allows reasoning by \emph{structural induction}.

\subsection{A running example}\label{sec:running-example}

\newcommand\pbt{\mathrm{pbt}}

Let $\Sigma^\exa = \{a,b\}$ with $\rank(a) = 2$ and $\rank(b) = 0$. We define $\pbt \colon \N \to \Tree{\Sigma^\exa}$ (for \enquote{perfect binary tree}) by:
\[ \pbt(0) = b \qquad \pbt(n+1) = a(\pbt(n), \pbt(n))\ \text{for}\ n \in \N\]
and our running example $\exa \colon \Tree{\Sigma^\exa} \to \Tree{\Sigma^\exa}$ by
\[ \exa(t) = \pbt(|\Dom(t) \cap 1^* 2^*|) \]
In other words, $\height(\exa(t))$ counts the number of nodes of $t$ that are reached from the root by paths that cannot take a \enquote{turn} from going right to going left. This function has:
\begin{description}
  \item[exponential size increase:] on inputs $a(a(\ldots a(b,b) \cdots), b), b)$;
  \item[linear size-to-height increase:] $\height(\exa(t)) \leqslant |t|$;  
  \item[quadratic height increase:] the number of words in $1^* 2^*$ of length at most $n$ is $(n+1)(n+2)/2$.
\end{description}

\subsection{Classical tree automata}

  A \emph{deterministic bottom-up tree automaton} over $\Sigma$ consists of:
  \begin{itemize}
    \item a finite set $Q$ of states;
    \item for each $\sigma\in\Sigma$, a transition function $\delta[\sigma] \colon Q^{\rank(\sigma)} \to Q$.
  \end{itemize}
  We may identify $\delta[\sigma]$ with an element of $Q$ when $\rank(\sigma)=0$. As usual, this can be extended to $\delta[t] \in Q$ for every tree $t \in \Tree\Sigma$ by structural induction: $\delta[\sigma(t_{1},\dots,t_{k})] = \delta[\sigma](\delta[t_{1}],\dots,\delta[t_{k}])$.

We shall be concerned with two extensions of deterministic bottom-up tree automata with additional data:
\begin{itemize}
  \item A \emph{bottom-up recognizer} comes with a codomain $F$ and a map
      $\varphi\colon Q\shortrightarrow F$. It defines a \emph{recognizable function} $t \mapsto \varphi(\delta[t])$ from trees to $F$ (whose range is always finite even when $F$ is infinite). The languages whose indicator functions are recognizable (with $F=\{0,1\}$) form the well-known class of \emph{regular tree languages}.
  \item A \emph{bottom-up relabeler} comes with a rank-preserving map $\rho\colon \Sigma \times Q \to \Gamma$. It defines a shape-preserving tree-to-tree function, called a \emph{bottom-up relabeling}: for $k = \rank(\sigma)$,
  \begin{align*}
    \sem{\cR}\colon \Tree{\Sigma} &\to \Tree{\Gamma} \qquad\text{(where \(\cR\) is the relabeler)}\\
    \sigma(t_1,\dots,t_{k}) &\mapsto \rho(\sigma,\delta[\sigma(t_1,\dots,t_{k})])(\sem{\cR}(t_1),\dots,\sem{\cR}(t_{k}))
  \end{align*}
\end{itemize}

\subsection{Tree-generating machines}\label{sec:tree-generating}

We recall from~\cite[\S3]{nguyen_et_al:LIPIcs.STACS.2025.68} a lightweight formalism for describing various tree transducer models. A \emph{tree-generating machine} with output alphabet $\Gamma$ is made of:
\begin{itemize}
  \item a set $\Conf$ of \emph{configurations};
  \item an \emph{initial configuration} $C_{init}\in\Conf$;
  \item a (partial) \emph{computation-step function}
        $\Conf \rightharpoonup \Tree{\Gamma}[\Conf]$.
\end{itemize}
$\Tree{\Gamma}[\Conf]$ should be seen as a set of global states of a computation:
\begin{itemize}
  \item the nodes with labels in $\Gamma$ are part of the eventual output tree, which is being produced in a top-down fashion;
  \item the leaves with labels in $\Conf$ correspond to processes running in parallel, each in charge of producing one of the remaining subtrees of the output.
\end{itemize}
By taking rewrite rules that substitute a leaf by the image of its label by the computation-step function, we get a confluent rewriting system on $\Tree{\Gamma}[\Conf]$.
A \emph{run} is a rewriting sequence starting from $C_{init}$. By confluence, there is at most one tree in $\Tree\Gamma$ that can be reached by the runs (indeed, such trees are normal forms). If it exists, we call it the \emph{output} of the machine.

To define the \emph{semantics $\sem{\cT} \colon \Tree\Sigma \rightharpoonup \Tree\Gamma$ of a tree transducer $\cT$}, our methodology is to specify how to build a tree-generating machine from $\cT$ and a given input tree $t\in\Tree\Sigma$, depending on the transducer model. The output of this tree-generating machine is then used as the definition of $\sem{\cT}(t)$.

\subsection{A branchwise semantics}

It is often useful to reason on the trajectory of a single process producing one of the output branches, rather than the entire family of processes producing the output tree in parallel. This is why we propose an alternative presentation of the semantics of a tree-generating machine, which does not appear in~\cite{nguyen_et_al:LIPIcs.STACS.2025.68}.

Let $\Gamma$ be the ranked output alphabet.
We define the set of \emph{branch words} of a tree over $\Gamma$ by structural induction:
\begin{align*}
  \brawo(\gamma) &= \{\gamma\}\quad \text{when}\ \rank(\gamma) = 0\\
  \brawo(\gamma(t_{1},\dots,t_{k})) &= \bigcup_{i} (\gamma,i) \cdot \brawo(t_{i})\  \text{when}\ k \geqslant 1
\end{align*}
\begin{claim}\label{clm:brawo}
  For $t \in \Tree\Gamma$, $b = (\gamma_{1},i_{1}) \dots (\gamma_{n},i_{n}) \gamma_{n+1} \in \brawo(t)$:
  \begin{itemize}
    \item $\lab_{t}(i_{1} \dots i_{\ell}) = \gamma_{\ell+1}$ for all $\ell\in\{0,\dots,n\}$;
    \item mapping $b$ to $\{\varepsilon, i_{1}, \dots, (i_{1} \dots i_{n})\}$ defines a size-preserving bijection from $\brawo(t)$ to the total branches of $t$.
  \end{itemize}
  Consequently, $\brawo \colon \Tree{\Gamma} \to \mathcal{P}(\widetilde{\Gamma}^{*}\rankK{\Gamma}{0})$ is injective, where
\begin{align*}
  \widetilde\Gamma &= \{(\gamma,i) \mid \gamma \in \Gamma,\; i \in \{1,\dots,\rank(\gamma)\}\}
\end{align*}
\end{claim}
We therefore aim to characterize the output of a tree-generating machine (with the notations of the previous subsection) through its branch words.
To do so, we consider a labeled transition system over $\Conf \cup \rankK{\Gamma}{0}$ with labels in $\widetilde{\Gamma}^{*}$.
\begin{definition}
  Let $C \in \Conf$, $b \in \widetilde{\Gamma}^{*}$ and $x \in \Conf \cup \rankK{\Gamma}{0}$.
  \[ C \xrightarrow{b} x \overset{\text{def}}{\iff} bx \in \brawo(\text{computation-step}(C)) \]
\end{definition}
As expected, we define $x \tto[b] y$ by the existence of some decomposition $b = b_{1} \dots b_{n}$ that labels some \emph{path}
\[ x = x_{0} \xrightarrow{b_{1}} \dots \xrightarrow{b_{n}} x_{n} = y \]
This is an instance of labeled transition system \emph{with monoidal structure on labels}, see for instance~\cite[Example~2.3]{RelationalPresheaves}.
\begin{claim}\label{clm:branchwise-sem}
  If the output of the tree-generating machine exists, then its set of branch
  words is $\{b\gamma \mid b \in \widetilde{\Gamma}^*,
  \gamma\in\rankK{\Gamma}{0}, C_{init} \tto[b] \gamma \}$.
\end{claim}
In general, we refer to a path starting from $C_{init}$ (not necessarily ending on an element of $\rankK{\Gamma}{0}$) as a \emph{branch-outputting run}.

\section{Tree-to-tree Hennie machines}\label{sec:hennie}

% Drawing on the usual definition of tree-walking automata, and the notion of Hennie
% machines, we define tree-to-tree Hennie machines as tree-walking tree transducers that
% access to finite memory over each node of the input, and are restricted to a only
% a bounded number of visits of each input node.
% Tito: déjà dit dans l'intro ça

\subsection{Unrestricted THMs: definition \& semantics}\label{sec:uthm}

\begin{definition}\label{def:uthm}
    An \emph{unrestricted tree-to-tree Hennie machine} ($\uTHM$) is a tuple $ \cH = (Q, M, \top, \Sigma, \Gamma, q_{init}, \delta)$ where:
\begin{itemize}
    \item $Q$ is the finite set of \emph{states} of $\cH$;
    \item $M$ is the finite set of \emph{memory symbols} of $\cH$;
    \item $\top \in M$ is the initial memory value of $\cH$;
    \item $\Sigma$ (resp.\ $\Gamma$) is the ranked \emph{input} (resp.\ \emph{output}) alphabet;
    \item $q_{init} \in Q$ is the initial state of $\cH$;
  \item $\delta \colon Q\times\Sigma\times M \rightharpoonup \Tree{\Gamma}[Q\times M\times \dirs{\maxrank{\Sigma}}]$ is the (partial) \emph{transition function} of $\cH$, such that for all $\sigma \in \Sigma$,
        \[ \delta(Q \times \{\sigma\} \times M) \subseteq \Tree{\Gamma}[Q\times M\times \dirs{\rank(\sigma)}]  \]
\end{itemize}
\end{definition}
The semantics of this $\uTHM$ has been informally sketched in the introduction. To describe it rigorously, we use the formalism of tree-generating machines. We start with its \emph{configurations} over an input $t\in \Tree\Sigma$: they are tuples $(u,q,\mu)$ where
\begin{itemize}
  \item $u$ is a node of $t$, the current position;
  \item $q \in Q$ is the current state;
  \item $\mu \colon \Dom(t) \to M$ describes the current memory values.
\end{itemize}
For such a configuration and $(q',m,d) \in Q\times M\times \dirs{\maxrank{\Sigma}}$, let
\[ (u,q,\mu) \ogreaterthan (q',m,d) = (ud, q', \mu')\ \text{where}\
  \begin{cases}
    \mu'(u) = m\\
    \mu'(v) = \mu(v)\ \text{for}\ v \neq u
  \end{cases}
\]
 --- the machine \enquote{writes the symbol $m$ in memory at the current position $u$ before moving to the new position $ud$}.
The result of this operation is \emph{undefined} if $ud \notin \Dom(t)$. 

We extend the definition of $\ogreaterthan$ to $\Tree{\Gamma}[Q\times M\times \dirs{\maxrank{\Sigma}}]$ by application to the leaves; inductively: $C \ogreaterthan \gamma(t_{1},\dots,t_{k}) = \gamma(C \ogreaterthan t_{1},\; \dots,\; C \ogreaterthan t_{k})$.

Let $\Conf(\cH,t)$ be the set of configurations of the $\uTHM$ $\cH$ on the input tree $t$. The \emph{computation-step function} is
\begin{align*}
  \Conf(\cH,t) &\rightharpoonup \Tree{\Gamma}[\Conf(\cH,t)] \\
  (u,q,\mu) &\mapsto (u,q,\mu) \ogreaterthan \delta(q,\lab_{t}(u), \mu(u))
\end{align*}
Finally, the \emph{initial configuration} is $C_{init}(t) = (\varepsilon, q_{init}, (v \mapsto \top))$.

\begin{example}\label{ex:thm}
  We give a $\uTHM$ that computes $\exa$ from~\S\ref{sec:running-example}.
  \begin{itemize}
    \item There are 2 states, $q_0$ (initial) and $q_1$.
    \item There are 4 memory symbols $\top,\ttL,\ttR,\maltese$.
    \item The transitions are (undefined in the 2 unspecified cases): 
          \begin{align*}
            (q_0,a,\top) &\mapsto a((q_1,\ttL,2),(q_1,\ttL,2)) & (q_0,b,\top) &\mapsto a(b,b) \\ 
            (q_1,a,\top) &\mapsto a((q_1,\ttR,2),(q_1,\ttR,2)) & (q_1,a,\ttL) &\mapsto (q_0,\maltese,1) \\ 
            (q_1,b,\top) &\mapsto a((q_1,\maltese,\uparrow),(q_1,\maltese,\uparrow)) & (q_1,a,\ttR) &\mapsto (q_1,\maltese,\uparrow) 
          \end{align*}
  \end{itemize}
  On an input tree $t\in\Tree{\Sigma^\exa}$, all branch-outputting runs ending in $b$ traverse the same configurations in the same order; the only difference is in the labels ($(a,1)$ vs.\ $(a,2)$). Such a run starts by writing $\ttL$ to the root, and going down to the right child in state $q_1$. It then goes right, writing $\ttR$ in memory on all encountered nodes, until it reaches a leaf. It then starts going up while writing $\maltese$ on the visited nodes (including the leaf) until it reaches a node where the memory is $\ttL$. After having thus visited all nodes in $\Dom(t) \cap 2^*$, it then moves to the left child of the root, which has not yet been visited, and switches back to the initial state $q_0$. The computation then proceeds as it did before on the subtree rooted at the left child. 

  Note that the transitions only output a node $a$ when the current memory is $\top$, which ensures that each input node is counted at most once when determining the output height.
\end{example}

\subsection{The bounded-visit restriction}\label{sec:bounded-visit}

\begin{remark}
  A $\uTHM$ computing a string-to-boolean function (i.e.\ $\max(\rank(\Sigma)) = 1$ and $\Gamma = \{\mathtt{true} : 0,\; \mathtt{false} : 0\}$) is the same thing as a deterministic \emph{linear bounded automaton}. It is well known that this automaton model can recognize at least all context-free languages~\cite[Theorem~3]{Kuroda64} while its nondeterministic variant recognizes precisely the context-sensitive languages~\cite[Theorem~1]{Kuroda64}.
\end{remark}

In order to stay within the realm of \enquote{finite-state computability} we shall therefore impose the bounded-visit restriction. The latter can be conveniently defined by using the branchwise semantics.
\begin{definition}\label{def:bounded-visit}
  Let $\cH$ be a $\uTHM$ with input alphabet $\Sigma$.
  
  The \emph{number of visits} of a branch-outputting run of $\cH$ at a node $u$ (resp.\ a set of nodes $S$) of the input tree is the number of configurations in the path whose current position is $u$ (resp.\ is in $S$).

  $\cH$ is \emph{bounded-visit} on an input language $L \subseteq \Tree{\Sigma}$ whenever there is a uniform bound $N$ such that, for any $t \in L$, any node $u$ of $t$ and any branch-outputting run on the input $t$, the number of visits at $u$ is at most $N$.
  A \emph{tree-to-tree Hennie machine} ($\THM$) is a $\uTHM$ that is bounded-visit on all inputs ($L = \Tree\Sigma$).
\end{definition}
\Cref{ex:thm} is a $\THM$: it visits each node at most twice.

We can now establish linear size-to-height increase rigorously.
\begin{proof}[{Proof of Claim~\ref{clm:lshi}}]
  By definition of height \& Claims~\ref{clm:brawo} / \ref{clm:branchwise-sem},
  \begin{align*}
    \height(\sem{\cH}(t))
    &= \max\{|b| \mid \exists \gamma : b\gamma \in \brawo(\sem{\cH}(t))\}\\
    &= \max\{|b_{1} \dots b_{\ell}| \mid \exists \gamma : C_{init}(t)  \xrightarrow{b_{1}} \dots \xrightarrow{b_{\ell}} \gamma \}
  \end{align*}
  The labeled transition system $\to$ for $\cH$ has labels derived from the branch words of the range $\delta(Q\times\Sigma\times M)$ of the transition function of $\cH$ (to see this, unfold the definition of the computation-step function). Since this range is finite, $|b_{i}| = O(1)$ in the above description, so $|b_{1} \dots b_{\ell}| = O(\ell)$. Finally, $\ell = O(|t|)$ because any branch-outputting run of a $\THM$ visits each of the $|t|$ nodes $O(1)$ times --- this is the bounded-visit restriction.
\end{proof}

\begin{remark}
Unfortunately, it is undecidable whether a $\uTHM$ is a $\THM$, already in the string-to-boolean case~\cite[Theorem~1 + Lemma~1]{GuillonPPP23}. However, given an explicit bound, we have a decidability result (\Cref{cor:reg-domain}). 
\end{remark}

We also define the \emph{number of downwards visits} of a branch-outputting run of a $\uTHM$ at an input node $u$ by counting the occurrences of transitions of the form $(\parent{u},\dots) \to (u,\dots)$ (thus, if $u$ is the root, this number is zero). The following property guarantees that \enquote{bounded-downwards-visit} is the same thing as bounded-visit.
\begin{claim}
  The number of visits at a node is at most the sum of the numbers of downwards visits at this node and at its children.
\end{claim}
\begin{proof}
  Along a branch-outputting run, the position of a $\uTHM$ moves locally, therefore every upwards visit of $u$ from $ui$ must have been preceded by a downwards visit of $ui$ from $u$. 
\end{proof}

\subsection{Weight-reducing $\THM$s \& totality}

The bounded-visit restriction is defined semantically. We present here a syntactic counterpart, taken from \cite{GuillonPPP23} for Hennie machine as acceptors of string languages.

\begin{definition}\label{def:weight-reducing}
  A $\uTHM$ with the set $M$ of memory symbols and the transition function $\delta$ is \emph{weight-reducing} if there exists a partial order on $M$ such that
  \begin{itemize}
    \item the initial symbol $\top$ is maximal;
    \item for every $m \in M$, only memory symbols that are strictly smaller than $m$ appear in $\delta(\dots,\dots, m)$. 
  \end{itemize}
\end{definition}
For instance, \Cref{ex:thm} is weight-reducing ($\top > {\ttL,\ttR} > \maltese$).
The following result may be compared with~\cite[Lemma~2]{GuillonPPP23}.
\begin{proposition}\label{prop:weight-reducing}
  Every weight-reducing $\uTHM$ is a $\THM$.

  Conversely, given a $\uTHM$ $\cH$ and $N \in \N$, one can compute a weight-reducing $\THM$ $\cH'$ such that $\sem{\cH'}$ is a \emph{total function} and:
  \begin{itemize}
    \item for $t$ in the domain of $\sem{\cH}$ s.t.\ all branch-outputting runs of $\cH$ visit at most $N$ times each node, $\sem{\cH'}(t) =\sem{\cH}(t)$;
    \item on other inputs, $\cH'$ returns an output that contains some chosen default leaf symbol. 
  \end{itemize}
  In particular, if $\cH$ is a $\THM$ whose number of visits is uniformly bounded by $N$, then $\sem{\cH} \subseteq \sem{\cH'}$.
\end{proposition}
Speaking of totality, the semantics of a $\THM$ is naturally a \emph{partial} function, but we focus on the case when the function happens to be total. This restriction is not very significant: as we explain below (\Cref{cor:reg-domain}), the domain of a partial $\THM$ is effectively regular.

%\begin{remark}\label{rem:total-root}
%We could also syntactically enforce totality as follows: take the data of a weight-reducing $\THM$ with total transition function, and add a total \enquote{transition function at the root} such that
%\[ \delta_{\mathrm{root}}(Q \times \{\sigma\} \times M) \subseteq \Tree{\Gamma}[Q\times M\times \db{\rank(\sigma)}] \]
%(note that $\uparrow \notin \db{\rank(\sigma)}$) in order to handle the special case
%\[ (\varepsilon,q,\mu) \mapsto (u,q,\mu) \ogreaterthan \delta_{\mathrm{root}}(q,\lab_{t}(\varepsilon), \mu(\varepsilon)) \]
%\end{remark}

\subsection{The narrow-visit restriction}

Recall that a \emph{chain} in a partially ordered set (poset) is a totally ordered subset, while an \emph{antichain} is a subset of pairwise incomparable elements. When we say that a set of nodes in a tree is an antichain, we mean with respect to the ancestor ordering.
\begin{definition}
  A $\uTHM$ $\cH$ is \emph{narrow-visit} on an input language $L$ whenever there is a uniform bound $N$ such that, for any $t \in L$, \emph{any antichain $S$ of $t$} and any branch-outputting run on the input $t$, the number of visits at $S$ is at most $N$.
\end{definition}
We rephrase the narrow-visit property by combining a basic observation with a classical duality theorem (which is closely related to other dualities in combinatorial optimization, cf.~\cite{Fulkerson1956}).
\begin{theorem}[Dilworth~{\cite{Dilworth1950}}]
  A poset is included in the union of $N<\infty$ chains iff all its antichains have cardinality at most $N$.
\end{theorem}

\begin{claim}
  For a set of tree nodes, the following are equivalent:
  \begin{itemize}
    \item it is a chain (for the ancestor ordering);
    \item it is included in some branch of the tree.
  \end{itemize}
\end{claim}

\begin{corollary}\label{cor:narrow-branches}
  An $\uTHM$ $\cH$ is narrow-visit on an input language $L$ if and only if the two following conditions hold jointly:
  \begin{itemize}
    \item $\cH$ is bounded-visit on $L$;
    \item for every branch-outputting run whose input tree is in $L$, the set of visited nodes (i.e.\ those with a nonzero number of visits) is included in the union of $O(1)$ input branches.
  \end{itemize}
\end{corollary}
\begin{corollary}\label{cor:NarrowLHI}
  If a $\uTHM$ is narrow-visit on all inputs, then it is a $\THM$ and it has linear height increase.
\end{corollary}
\begin{proof}
  As we have seen when proving Claim~\ref{clm:lshi}, the output height is at most linear in the length of branch-outputting runs. On an input $t$, such a run visits $O(1)$ branches, each of which has at most $\height(t)+1$ nodes, and each node is visited $O(1)$ times.
\end{proof}

\subsection{Basic extensions}\label{sec:basic-extensions}

To conclude \Cref{sec:hennie}, we present three additional features that can be added to $\uTHM$s and $\THM$s. We show that these extensions are conservative, i.e.\ do not lead to an increase in expressive power.

\subparagraph{Stay-instructions}

The positions of two successive configurations of an $\uTHM$ are always distinct neighbors. We extend the syntax of transitions to allow staying on the current position; this is often convenient in our constructions of $\THM$s. Stay-instructions are indicated by a new symbol $\circlearrowleft$, and relax the requirement on the transition function $\delta$ to be
\[ \delta(Q \times \{\sigma\} \times M) \subseteq \Tree{\Gamma}[Q\times M\times (\{\circlearrowright\}\cup\dirs{\rank(\sigma)})] \]
The semantics is extended by taking $u{\circlearrowleft} = u$ for any input node $u$.

\begin{lemma}\label{lem:stay}
  Every $\uTHM$ with stay-instructions can be effectively translated to an equivalent $\uTHM$. This translation preserves the bounded-visit and narrow-visit restrictions on any input language.
\end{lemma}

\subparagraph{Bottom-up initialization.} Our second conservativity result can be understood as a \enquote{closure under precomposition} property.

\begin{lemma}\label{lem:bottom-up-rel}
  Given a $\uTHM$ $\cH$ and a bottom-up relabeler $\cR$ (cf.~\S\ref{sec:notations}), whose respective input and output alphabets match, one can effectively construct a $\uTHM$ $\cH'$ such that $\sem{\cH'}=\sem{\cH}\circ\sem{\cR}$.

  If $\cH$ is bounded-visit on $\sem{\cR}(L)$, then $\cH'$ is bounded-visit on $L$.
\end{lemma}
The idea of the proof is to perform a post-order tree traversal that computes and records in memory, at each node, the state reached there by the deterministic bottom-up automaton contained in $\cR$.

\begin{definition}
  A $\uTHM$ \emph{with bottom-up initialization} is:
  \begin{itemize}
    \item the data of a $\uTHM$ --- we reuse the notations of Def.~\ref{def:uthm};
    \item a deterministic bottom-up tree automaton $(Q',\delta')$;
    \item a map $\mathtt{init} \colon Q' \to M$.
  \end{itemize}
  Its semantics is the same as for $\uTHM$ except that for an input tree $t\in \Tree\Sigma$, the initial memory value at node $u$ is $\mathtt{init}(\delta'[\subtree{t}{u}])$. 

  We then define in the same way as \Cref{def:bounded-visit} the bounded-visit condition. A \emph{$\THM$ with bottom-up initialization} is a $\uTHM$ with bottom-up initialization that happens to be bounded-visit.
\end{definition}
\begin{proposition}
  Every $\uTHM$ with bottom-up initialization can be effectively translated to an equivalent $\uTHM$.

  The translation preserves the bounded-visit condition on any input language, but not the narrow-visit condition.
\end{proposition}
\begin{proof}
  For any $\uTHM$ $\cH''$ with bottom-up initialization, let:
  \begin{itemize}
    \item $\cR$ be the bottom-up relabeler using $(Q',\delta')$ from the above definition and $\rho\colon (\sigma,q') \mapsto (\sigma,\mathtt{init}(q'))$; 
    \item $\cH$ be a $\uTHM$ simulating $\cH''$ while reading from the input letters the initial memory values when it encounters a node for the first time --- thus, when $\cH''$ is bounded-visit on all inputs, $\cH$ is bounded-visit on the range of $\sem{\cR}$.
  \end{itemize}
  Then $\sem{\cH''} = \sem{\cH} \circ \sem{\cR}$.
  Therefore, we can apply \Cref{lem:bottom-up-rel}.
\end{proof}

\subparagraph{Regular lookaround.} Finally, we consider the extension of tree-to-tree Hennie machines with the ability to use regular information about the current configuration to decide the next transition.

\begin{definition}[Encodings of configurations]\label{def:enc-conf}
  Consider a $\uTHM$ $\cH$ with the notations of \Cref{def:uthm}.

  For $t\in\Tree\Sigma$, we encode $(u,q,\mu) \in \Conf(\cH,t)$ as the tree over $\Sigma \times M \times (Q \cup \{\mathtt{notHere}\})$, with the ranks inherited from $\Sigma$, whose domain is $\Dom(t)$ and whose label at each node $v$ is
  \[ (\lab_t(v),\; \mu(v),\; [q\ \text{if}\ v=u\text{, otherwise}\ \mathtt{notHere}]) \]
\end{definition}
\begin{definition}\label{def:thm-rla}
  A \emph{$\uTHM$ with regular lookaround} ($\uTHMr$) $\cH$ from an input alphabet $\Sigma$ to an output alphabet $\Gamma$ consists of:
  \begin{itemize}
    \item the data of a $\uTHM$ \emph{without the transition function} as in \Cref{def:uthm}, i.e.\ $(Q, M, \top, \Sigma, \Gamma, q_{init})$;
    \item a \emph{recognizable} transition function $\delta \colon \Tree{\Sigma \times M \times (Q
        \cup \{\mathtt{notHere}\})} \rightarrow T_{\Gamma}(Q \times M \times \dirs{k})$.
  \end{itemize}
  The configurations of the machine $\cH$ are the same as if it were an ordinary
  $\uTHM$ (the transition function is not involved). The semantics is defined
  analogously to ordinary $\uTHM$s, except that the computation-step function
  passes the entire configuration to the transition function --- since it is recognizable, it only depends on a finite amount of information:
  \begin{align*}
    \Conf(\cH,t) &\rightharpoonup \Tree{\Gamma}[\Conf(\cH,t)] \qquad\qquad(\text{for}\ t \in \Tree\Sigma)\\
    C &\mapsto C \ogreaterthan \delta(\text{encoding of}\ C)
  \end{align*}
\end{definition}
We combine bottom-up initialization with dynamic maintenance of local lookaround information to prove:
\begin{theorem}\label{thm:lookaround}
    Every $\uTHMr$ can be effectively translated to an equivalent $\uTHM$. The translation preserves the bounded-visit condition on any input language.
\end{theorem}
\begin{corollary}\label{cor:precomp-relab}
  The classes of functions computed by $\uTHM$s and by $\THM$s are both closed by precomposition by \emph{$\mso$ relabelings}.
\end{corollary}
The aforementioned $\mso$ relabelings are tree transformations that keep the shape of a tree (i.e.\ its domain) while changing its labels in a regular fashion; see for instance~\cite[Definition~7.20]{courcellebook} for a book reference.
\begin{remark}
  For automata models such as e.g.\ tree-walking tree transducers and macro tree transducers, regular lookaround morally amounts to preprocessing by $\mso$ relabelings. But for $\THM$s, this is not the case because the information given by the regular lookaround depends not only on the input tree and current position, but also on the current memory values at each node.
\end{remark}

\section{Macro tree transducers}\label{sec:mtt}

\subparagraph{Notations specific to this section.}

We reserve two special sets of symbols, a set $ X = \{x_i : i \in \N \setminus \{0\}\} $ of \emph{variables} and another set $ Y = \{y_i : i \in \N \setminus \{0\} \} $ of \emph{parameters}. We also write $X_k = \{x_i : i \leqslant k\} $ and $ Y_k = \{y_i : i \leqslant k \}$ for $k \in \N$.

A tree in $\Tree{\Sigma}[Y]$, whose leaves may be parameters, is called a \emph{tree context} over $\Sigma$. A \emph{$k$-ary tree context} is a tree in $\Tree{\Sigma}[Y_{k}]$.

% \tito{pas sûr que le cas general des SOsubst soit utile}
% In addition to the previously defined first-order substitution of a leaf by a tree, there is also a notion of \emph{second-order substitution} of a possibly internal node by a \emph{context}: given a tree $t \in \Tree{\Sigma}$, a node $u$ in $t$ of rank $k$, and a $k$-ary tree context $c \in \Tree{\Sigma}[Y_{k}]$, we define
% \[ t[u \Leftarrow c] = t[u \leftarrow c[y_i \leftarrow \subtree{t}{ui}\; \forall i]] \]
% Informally speaking, the second-order substitution first instantiates the parameters of $c$ with the  subtrees rooted at the children of $u$, then plugs the result at position $u$ in $t$.
% \tito{Je pense que c'est mieux de changer la notation pour la flèche que pour les crochets, car $\sem{-}$ désigne déjà la sémantique. TODO: propager cette notation partout}

For $Q$ a ranked set and $Z$ an arbitrary set, $\qq{Q, Z}$ is the ranked set of pairs $\qq{q, z}$ for $q\in Q$, $z \in Z$, with $\rank(\qq{q, z})  = \rank(q)$.

\subsection{MTTs: definition and semantics}

\begin{definition}\label{def:mtt}
  A (total deterministic) $\mtt$ $\cM$ from an input
  alphabet $\Sigma$ to an output alphabet $\Gamma$ (both ranked) consists of:
  \begin{itemize}
    \item a finite \emph{ranked} set $Q$ of states \& an initial state $q_{init} \in \rankK{Q}{(0)}$;
    \item for each state $q\in Q$ and each input letter $\sigma\in\Sigma$, a context
    \[ \RHS_{\cM}(q,\sigma) \in \Tree{\Gamma \cup \qq{Q,X_{\rank(\sigma)}}}[Y_{\rank(q)}] \]
  \end{itemize}
\end{definition}
The latter is meant to be the right-hand side of a rewrite rule
\[ \qq{q,\sigma(x_{1},\dots,x_{\rank(\sigma)})}(y_{1},\dots,y_{\rank(q)}) \to \RHS_{\cM}(q,\sigma) \]
Let us illustrate what this means before giving a formal definition:
\begin{example}\label{ex:mtt}
  Here is an $\mtt$ for the function $\exa$ from~\S\ref{sec:running-example}.
  \begin{align*}
    \qq{q_0,a(x_1,x_2)} &\to \qq{q_1,x_2}(a(\qq{q_0,x_1},\qq{q_0,x_1}))\\
    \qq{q_0,b} &\to a(b,b)\quad (\text{i.e.}\ \pbt(1)\ \text{as defined in \S\ref{sec:running-example}}) \qquad\phantom{.}\\
    \qq{q_1,a(x_1,x_2)}(y_1) &\to \qq{q_1,x_2}(a(y_1,y_1))\\
    \qq{q_1,b}(y_1) &\to a(y_1,y_1)
  \end{align*}
$\begin{aligned}
  \text{We have}\ \qq{q_{0},a(b,b)} &\to \qq{q_{1},b}(a(\qq{q_{0},b},\qq{q_{0},b}))
                                    &\to \qq{q_{1},b}(a(\pbt(1),\qq{q_{0},b}))\\
                                    &\to \qq{q_{1},b}(\pbt(2)) \to \pbt(3) = \exa(a(b,b)).   
\end{aligned}$
\end{example}
In the literature, the semantics of $\mtt$s allow the rewrite rules to apply to any subtree rooted at some node with label of the form $\qq{q,t}$. Let us call such trees \emph{redexes}:
\[ \Redex(\cM) = \{ t \in \Tree{\Gamma \cup \qq{Q,\Tree{\Sigma}}} \mid \lab_t(\varepsilon) \in \qq{Q,\Tree{\Sigma}} \} \]
There is also a standard notion~\cite[Definition~3.4]{Macro} of \emph{outside-in derivation}\footnote{Or \enquote{OI derivation}, corresponding to call-by-name evaluation in the $\lambda$-calculus, while inside-out (IO) derivations correspond to call-by-value.} that only rewrites \emph{outermost redexes}: those that are not contained in some other redex, i.e.\ whose root does not have a $\qq{Q,\Tree{\Sigma}}$-labeled strict ancestor. The example above is not outside-in: the only steps rewriting outmost redexes are the first and the last.

We now capture outside-in derivations --- which suffice to interpret total deterministic $\mtt$s as total functions~\cite[\S4.1]{Macro} --- in the framework of tree-generating machines. This serves as our reference definition of the semantics of $\mtt$s in this paper, avoiding the need to manipulate \enquote{second-order substitutions} in full generality.

The tree-generating machine for the $\mtt$ $\cM$ on an input $t\in\Tree\Sigma$ has the set of configurations $\Redex(\cM)$ and the initial configuration $\qq{q_{init},t}$. Its computation-step function is
\begin{align*}
  \Redex(\cM) &\to \Tree{\Gamma \cup \qq{Q,\Tree{\Sigma}}} \cong \Tree{\Gamma}[\Redex(\cM)]\\
  \qq{q,\sigma(\vec{t})}(\vec{t}')
             &\mapsto
               {\begin{cases}
                 \text{replace every label \(\qq{q',x_{i}}\) by \(\qq{q',t_{i}}\)}\\
                 \text{in}\ \RHS_{\cM}(q,\sigma)[\vec{y} \leftarrow \vec{t}']
               \end{cases}}
\end{align*}
using the fact that:
\begin{claim}
  The simultaneous substitution
  \[ [r\ \text{seen as a leaf} \leftarrow r\ \text{seen as a tree}\quad \forall r \in \Redex(\cM)] \]
  defines a bijection $\Tree{\Gamma}[\Redex(\cM)] \xrightarrow{\sim} \Tree{\Gamma \cup \qq{Q,\Tree{\Sigma}}}$ whose inverse turns outermost redexes into leaves. 
\end{claim}

\subsection{Syntactic criteria for linear (size-to-)height}

We now recall the results of Gallot, Maneth, Nakano and Peyrat~\cite{GMNP24} on L(S)HI-MTTs. Originally, they are stated in terms of $\mtt$s with \emph{regular lookahead}. As usual, regular lookahead amounts to preprocessing by a bottom-up relabeling; we propose a rephrasing that takes this point of view.

\begin{definition}
  Let $\cM$ be an $\mtt$ with the notations above.

  Let us fix a rank 0 symbol $\square$. We define the $\mtt$ $\cM^\blacksquare$ with:
  \begin{itemize}
    \item input alphabet $\Sigma \cup \{\square\}$ and output alphabet $\Gamma \cup \qq{Q,\{\blacksquare\}}$;
    \item the same ranked set of states $Q$ as $\cM$;
    \item the rules of $\cM$, plus $\qq{q,\square}(\vec{y}) \to \qq{q,\blacksquare}(\vec{y})$ for $q\in Q$.
  \end{itemize}
\end{definition}
For instance, for the $\mtt$ $\cM_\exa$ of \Cref{ex:mtt}, 
\[ \sem{\cM_\exa^\blacksquare}(a(\square,b)) = \qq{q_{1},\blacksquare}(\pbt(2)) \]
$\cM^\blacksquare$ is called the \enquote{extension} of $\cM$ in~\cite{GMNP24}. Its inputs are trees in $\Tree{\Sigma}[\{\square\}]$; they can be seen as inputs to $\cM$ in $\Tree\Sigma$ with some \enquote{missing subtrees} marked by $\square$.
Informally, the output of $\cM^\blacksquare$ represents what happens if we run $\cM$ on such an incomplete input, computing as much as possible while recording the places where $\cM$ gets stuck because it wants to explore a missing subtree.

For an input language $L \subseteq \Tree\Sigma$, we denote by $L^\square \subseteq \Tree{\Sigma}[\{\square\}]$ the set of \enquote{inputs from $L$ with missing subtrees}, that is, the smallest superset of $L$ closed under replacing a subtree by $\square$. The set of \enquote{inputs from $L$ with a single missing subtree} is
\[ L^{1\square} = \{ t \in L^\square \mid t\ \text{has only one \(\square\)-labeled leaf}\} \]

\begin{definition}
  For $t' \in \Tree{\Gamma \cup \qq{Q,\{\blacksquare\}}}$, let $\height_\blacksquare(t')$ be the maximum number of $\qq{Q,\{\blacksquare\}}$-labeled nodes on a branch of $t'$. 
  
  An $\mtt$ $\cM$ is \emph{finite nesting} (resp.\ \emph{finite yield nesting}) on an input language $L$ when $\height_\blacksquare \circ \sem{\cM^\blacksquare}$ is bounded on $L^{1\square}$ (resp.\ $L^\square$).
\end{definition}
One can show that $\height_\blacksquare(\cM^\blacksquare_\exa(t'))$ is the number of $\square$-labeled nodes in $\Dom(t') \cap 1^* 2^*$ for all $t' \in \Tree{\Sigma^\exa}[\{\square\}]$. Therefore, $\cM_\exa$ is finite nesting, but not finite yield nesting --- which is consistent with the (size-to-)height increase properties stated in \S\ref{sec:running-example}, since:
\begin{theorem}[{rephrasing of~\cite[\S4]{GMNP24}}]\label{thm:mtt-rephrased}
  Given an $\mtt$ $\cM$, one can compute an $\mtt$ $\cM'$ and a bottom-up relabeler $\cR$ such that:  
  \begin{itemize}
    \item $\sem{\cM} = \sem{\cM'} \circ \sem{\cR}$;
    \item $\cM$ is LSHI iff $\cM'$ is finite nesting on the range of $\cR$;
    \item $\cM$ is LHI iff $\cM'$ is finite yield nesting on the range of $\cR$.
  \end{itemize}
\end{theorem}

% \begin{remark}
%   \Cref{thm:mtt-lhi} then follows from the decidability of finite (yield) nesting, also established in~\cite[\S4]{GMNP24}.
% \end{remark}

\subsection{From MTTs to (u)THMs}

The main result of this section is:
\begin{theorem}\label{thm:mtt-hennie}
  Any (total deterministic) $\mtt$ $\cM$ can be effectively translated to a $\uTHM$ $\cH$ such that:
  \begin{itemize}
    \item $\sem{\cM} = \sem{\cH}$;
    \item for any input tree $t$, any antichain $S$ of nodes of $t$, and any branch-outputting run of $\cH$, the number of downwards visits at $S$ is at most $\height_\blacksquare \circ \sem{\cM^\blacksquare}(t[s \leftarrow \square\;\; \forall s \in S])$.
  \end{itemize}
  The second item implies that, for any input language $L$,
  \begin{align*}
    \cM\ \text{finite nesting on}\ L &\implies \cH\ \text{bounded-visit on}\ L\\
    \cM\ \text{finite yield nesting on}\ L &\implies \cH\ \text{narrow-visit on}\ L
  \end{align*}
\end{theorem}

To give some intuition for the proof, one may see the rules of an $\mtt$ as mutually recursive definitions of $\sem{q} \colon \Tree\Sigma \times {\Tree\Gamma}^{\rank(q)} \to \Tree\Gamma$ for each state $q$. This recursion can be compiled to a call stack. In fact, $\mtt$s are expressively equivalent to pushdown tree transducers~\cite[Theorem~5.16]{EngelfrietPushdownMacro},\footnote{Other similar equivalence results include HDT0L systems vs.\ a variant of pushdown transducers~\cite{FerteMarinSenizergues} \& copyful streaming string transducers vs.\ marble transducers~\cite{Marble,gaetanPhD}.} though we do not use this result and prefer to give a direct construction from $\mtt$s in order to control the number of visits. To represent the stack in memory, our $\uTHM$ records at each node the currently active recursive call, if any, whose argument in $\Tree\Sigma$ is the subtree rooted at that node. There can be at most one such call per node, so we can use a finite set of memory symbols. As for $\height_\blacksquare \circ \sem{\cM^\blacksquare}(t[s \leftarrow \square\;\; \forall s \in S])$, it overapproximates the total number of calls, in a branch-outputting run, whose arguments in $\Tree\Sigma$ are subtrees rooted at nodes in $S$.

From Theorems~\ref{thm:mtt-rephrased} and~\ref{thm:mtt-hennie}, combined with \Cref{lem:bottom-up-rel} on precomposing $\THM$s by bottom-up relabelings, we get: 
\begin{corollary}\label{cor:lshimtt-hennie}
  $\lshii\mtt \subset \THM$.
\end{corollary}
\begin{corollary}\label{cor:lhimtt-narrow}
  Given an $\lhii\mtt$ $\cM$, one can compute an $\THM$ $\cH$ and a bottom-up relabeler $\cR$ such that $\sem{\cM} = \sem{\cH} \circ \sem{\cR}$ and $\cH$ is \emph{narrow-visit on the range of $\sem{\cR}$}. 
\end{corollary}
This last property is involved in the proof of \Cref{cor:lhimtt-comp-thm}.

\section{Inclusion in MSO set interpretations}\label{sec:msosi}

We assume basic knowledge of monadic second-order logic. For a ranked alphabet $\Sigma$, the formulas of $\mso(\Sigma)$ use the atomic predicates $\sigma(x)$, $x \leqslant_{anc} y$ and $\child_i(x,y)$ to express respectively that $x$ has label $\sigma$, $x$ is an ancestor or $y$, and $y$ is the $i$-th child of $x$.
\begin{definition}[{\cite{ColcombetL07}}]
An \emph{$\mso$ set interpretation} (on ranked trees), or $\msosi$ for short, consists of:
\begin{itemize}
  \item a natural number $c$;
  \item input and output ranked alphabets $\Sigma$ and $\Gamma$;
  \item the following $\mso(\Sigma)$ formulas, where $\vec{X} = X_1,\dots,X_c$ and $\vec{Y} = Y_1,\ldots,Y_c$ are sequences of second-order variables:
        \[\begin{array}{ccc}
          \text{label} & \text{ancestor} & \text{child} \\
          (\varphi_\gamma(\vec{X}))_{\gamma \in \Gamma} & \varphi_{anc}(\vec{X},\vec{Y}) & (\varphi_i(\vec{X},\vec{Y}))_{i\leqslant\maxrank{\Gamma}}
        \end{array}\]
\end{itemize}
\end{definition}
%We give the semantics of an $\msosi$ as a Tree-generating machine:
%\begin{itemize}
%\item its set of configurations is $(t,s)$ where $t$ is an input tree and $s$ assigns a subset of nodes of $t$ to each $X_i$ and there exists $\gamma$ such that $(t,S_1,\ldots,S_c) \models \varphi_\gamma(\vec{X})\}$.
%\item its initial configuration the configuration $(t,i)$ such that for all other configuration $(t,s)$, $(t,s,i)\not\models \varphi_{anc}(\vec{X},\vec{Y})$
%\item its computation-step function:
%$(t,s) \to \gamma((t,s_1),\ldots,(t,s_k))$ if $(t,s)\models \varphi_\gamma(\vec{X})$ and
%for all $i<k$ $(t,s,s_i)\models \varphi_i(\vec{X},\vec{Y})$  
%\end{itemize}
%Additionally, it should hold that the initial configuration is unique, the computation-step function is deterministic, and for all $\nu,\nu'$, $(t,s)$ is an ancestor of $(t,s')$ (equivalently $(\nu,\nu')$ belongs to the transitive closure of the $child_i$ relations) if and only if $(t,\nu,\nu')\models \varphi_{anc}(\vec{X},\vec{Y})$.
%If these do not hold, then the $\msosi$ is undefined on $t$.
An $\msosi$ $\cI$ on an input tree $t$ generates a relational structure $\cI(t)$ whose universe is the set of assignments
\[\bigl\{\nu\colon \vec{X} \to \mathcal{P}(\Dom(t)) \mid \exists \gamma \in \Gamma : (t,\nu)\models \varphi_\gamma(\vec{X}) \bigr\}\]
and whose relations are $\gamma(\nu) \iff (t,\nu)\models \varphi_\gamma(\vec{X})$ and
\begin{align*}
  \nu\leqslant_{anc}\nu' &\iff (t,\nu,\nu')\models \varphi_{anc}(\vec{X},\vec{Y})\\
  \child_i(\nu,\nu') &\iff (t,\nu,\nu')\models\varphi_i(\vec{X},\vec{Y}) \quad\text{for}\ i\leqslant \maxrank{\Gamma} 
\end{align*}
Then $\sem\cI\colon\Tree\Sigma\rightharpoonup\Tree\Gamma$ is defined on $t$, if, and only if, $\cI(t)$ is isomorphic to a tree, which we take as the value of $\sem{\cI}(t)$.

\begin{example}\label{example:msosi}
We give an $\msosi$ $\cI$ defining our running example from \S\ref{sec:running-example}.
We assume that all formulas are restricted to nodes in $1^*2^*$, which is $\mso$ definable by guarding all quantifications $\exists x$ (resp. $\exists X$) with $\psi_{1^*2^*}(x)=\exists z\ \child_1^*(\varepsilon,z)\wedge \child_2^*(z,x)$ (resp $\forall x\in X\ \psi_{1^*2^*}(x)$) where $\child_i^*$ is the transitive closure of the $\child_i$ relation and $\varepsilon$ is the constant marking the root. 
Intuitively, $\cI$ relies on a (right-depth) ordering $\leqslant$ of the nodes of the input in $\Dom(t)\cap 1^*2^*$. 
We define the successor in this order by 
\begin{align*}
S_\leqslant(x,y) \equiv &\child_2(x,y) \vee 
\big((\forall z\ \lnot\child_2(x,z))\wedge  \\
&\exists z\ \child_2^*(z,x)\wedge \forall z'\ \lnot\child_2(z',z)\wedge \child_1(z,y) \big)
\end{align*}
Then the output nodes of depth $i$ are defined by partitions into $X_1,X_2$ of the $i$-smallest elements of $\leqslant$.

Let us consider the formula $\psi_{dom}(X_1,X_2)$ stating that $X_1$ and $X_2$ are disjoint and $X_1\cup X_2$ is a set downward closed by $\leqslant$, and $max_\leqslant(X_1,X_2,x)$ which holds true if $x$ is the $\leqslant$-largest element in $X_1\cup X_2$. 
Then we define $T$ by $c=2$, $\varphi_a(X_1,X_2)$ is $\psi_{dom}$ and $X_1 \cup X_2$ does not cover all $\leqslant$, while $\varphi_b(X_1,X_2)$ is $\psi_{dom}(X_1,X_2)$ and $X_1\cup X_2$ does cover all $\leqslant$.
The ancestor formula is $$\varphi_{anc}(\vec{X},\vec{Y})=\forall z\ \bigwedge_{i=1,2} z\in X_i\Rightarrow z\in Y_i$$ where $x\leqslant y$ is the transitive closure of $S_\leqslant(x,y)$.
And finally $\varphi_i(\vec{X},\vec{Y})$ is the conjunction of $\varphi_{anc}(\vec{X},\vec{Y})$ and $$\exists x,y\ max_\leqslant(\vec{X},x)\wedge max_\leqslant(\vec{Y},y)\wedge S_\leqslant(x,y)\wedge y\in Y_i$$
\end{example}

The main result of this section is the following Theorem:
\begin{theorem}\label{thm:THMtoMSOSI}
Any total $\THM$ $\cH$ can be effectively translated to an $\msosi$ $\cI$ such that $\sem{\cH}=\sem{\cI}$.
\end{theorem}
\begin{proof}[Sketch of proof]
The main ingredients of the proof are the notions of local and (coherent) global profiles.
Intuitively, a branch-outputting run of $\cH$ is entirely characterized by the finite visits it does to each input nodes. 
A local profile on an input node is then a sequence of visits.
A global profile simply associates to each input node a local profile. 
It is coherent if the local profiles collectively describe a (partial) branch-outputting run.
As a $\THM$ moves step by step, the coherence of a global profile can be decided by neighboring properties and hence it is $\mso$-definable.

To define the $\msosi$ $\cI$, we set $c$ as the number of possible local profiles.
By associating each set $X_p$ to a local profile $p$, an evaluation $(t,\nu)$ where the $X_i$ are a partition of $t$ defines a global profile on $\Dom(t)$. 
The domain formulas $\varphi_\gamma(\vec{X})$ state that the coloring $\vec{X}$ is a  coherent global profile and describes a branch-outputting and a distinguished node $\gamma$ in its last transition.
The ancestor formula $\varphi_{anc}(\vec{X},\vec{Y})$ states that the global profile defined by $\vec{Y}$ is an extension of the one defined by $\vec{X}$.
Finally, the successor formulas $\varphi_i(\vec{X},\vec{Y})$ are similar to the ancestor formula, additionally requiring that the extension is only one step, i.e. with a unique new visit that produces an element of $\Gamma$, the $i$-th successor of $\vec{X}$.
\end{proof}

\begin{corollary}\label{cor:string-msot}
  \emph{Tree-to-string} Hennie machines (i.e.\ $\THM$s whose output alphabet have maximum rank 1) can be effectively translated to equivalent \emph{$\mso$ transductions}.
\end{corollary}
\begin{proof}
	Since the size of a string is equal to its height (plus one),	
	by Claim~\ref{clm:lshi} tree-to-string $\THM$s are linear size increase.
	$\msosi$ of linear size increase can be translated to $\mso$ transductions~\cite[Theorem~1.5]{StructPoly}, therefore \Cref{thm:THMtoMSOSI} gives the result.
\end{proof}
% \tito{vérifier si résultats MSOSI cas total ou partiel}
\begin{corollary}\label{cor:reg-domain}
  The subset of inputs of a $\uTHM$ on which it is defined and visits each node at most $N$ times is effectively a regular language. 
\end{corollary}
\begin{proof}
  Let $\cH$ be a $\uTHM$ with $\sem{\cH}\colon\Tree\Sigma\rightharpoonup\Tree\Gamma$. We apply \Cref{prop:weight-reducing}, taking a default leaf symbol $\blacksquare \notin \Gamma$. We then get a total $\THM$ $\cH'$ with output alphabet $\Gamma\cup\{\blacksquare\}$ such that $\sem{\cH'}^{-1}(\Tree\Gamma)$ is the aforementioned subset. The \emph{first-order} formula $\forall x.\, \lnot \blacksquare(x)$ defines $\Tree\Gamma$ within $\Tree{\Gamma}[\{\blacksquare\}]$, so its inverse image by any $\mso$ set interpretation is regular~\cite[\S2.3]{ColcombetL07}. (\Cref{cor:reg-preserving} states a stronger preservation property of $\THM$s, but its proof goes through $\lambda$-calculus.)
\end{proof}

\section{Actors: an interactive model of computation}\label{sec:actors}

This section treats part of the equivalence proof between tree-to-tree Hennie machines and a $\lambda$-calculus characterization --- a part that  \emph{that does not involve any $\lambda$-term}. We introduce a model of computation for tree-to-tree functions, based on a formalism from from Clairambault and Murawski~\cite{MAHORS} of tree-generating, finite-state processes which we call \emph{actors}. The main result in this section (\Cref{lem:actor-to-thm}) is a translation from actor-based tree transducers with a certain boundedness property to Hennie machines.

The idea is that an actor and its environment interact by exchanging messages; during the course of the interaction, both the actor and the environment may also output tree nodes. We describe the allowed messages for an actor using \emph{types}, which form a subset of linear logic formulas (though a more general notion of interface would have been possible):
$A,B \mathrel{::=} o \mid A \multimap B \mid A \with B$.
\begin{definition}[{\cite[\S3.1]{MAHORS}}]
  The sets $\IM(A)$ of incoming messages and $\OM(A)$ of outgoing messages at a type $A$ are  defined by:
\begin{align*}
  \IM(o) &= \{\bullet\} & \OM(o) &= \varnothing \\
  \IM(A \multimap B) &= \OM(A) + \IM(B) & \OM(A \multimap B) &= \IM(A) + \OM(B)\\
  \IM(A \with B) &= \IM(A) + \IM(B) & \OM(A \with B) &= \OM(A) + \OM(B)
\end{align*}
using the disjoint sum operation $X + Y = \{\ttL\} \times X \cup \{\ttR\} \times Y$.
\end{definition}

Intuitively, an actor is in a \emph{passive} state when it is waiting for an incoming message from the environment. Upon receiving this message, it becomes \emph{active}, then performs some computation, and eventually sends an outgoing message, returning the control to the environment by becoming again passive.

\begin{definition}[{\cite[Def.~2]{MAHORS}}]
  An actor $\alpha$ of type $A$ over the ranked alphabet $\Gamma$ (notation: $\Gamma \Vdash \alpha : A$) consists of:
  \begin{itemize}
    \item a finite set $Q^\ominus$ of \emph{passive states} \& an \emph{initial state} $q^\ominus_0 \in Q^\ominus$;
    \item a finite set $Q^\oplus$ of \emph{active states}, nonempty, disjoint from $Q^\ominus$;
    \item two transition functions $\delta^\ominus \colon Q^\ominus \times \IM(A) \rightharpoonup Q^\oplus$ and
          \[ \delta^\oplus \colon Q^\oplus \rightharpoonup Q^\oplus \cup \{\underbrace{\gamma(q^\oplus_1,\dots)}_{\mathclap{\text{this notation implies that \(\gamma\) takes \(\rank(\gamma)\) arguments}}} \mid \gamma \in \Gamma,\, q^\oplus_i \in Q^\oplus\} \cup (Q^\ominus \times \OM(A))\]
  \end{itemize}
  When $A = o$, we associate to this actor the tree-generating machine:
  \begin{itemize}
    \item whose set of configurations is $Q^\oplus$;
    \item whose initial configuration is $\delta^\ominus(q^\ominus_0,\bullet)$;
    \item whose computation-step function is $\delta^\oplus$.
  \end{itemize}
  The output of this machine is the \emph{tree generated by $\Gamma \Vdash \alpha : o$}.
\end{definition}

\begin{example}
  Every tree over $\Gamma$ can be generated by an actor with a single passive state, and whose active states correspond to the internal nodes of the tree.
\end{example}

\begin{example}\label{ex:actor-diagonal}
  We illustrate the use of multiple passive states with an actor of type $A \multimap (A \with A)$. It relays messages back and forth between the $A$ on the left of $\multimap$ and the two $A$s on the right; the incoming messages on the left are forwarded to a recipient determined by the state. Formally: $Q^\ominus = \{\ttL,\ttR\}$ and $Q^\oplus,\delta^\ominus,\delta^\oplus$ are chosen to satisfy
  \begin{align*}
    (\delta^\oplus \circ \delta^\ominus) \colon Q^\ominus \times \overbrace{\IM(A \multimap (A \with A))}^{\mathclap{\OM(A) + (\IM(A) + \IM(A))}} &\rightharpoonup Q^\ominus \times \overbrace{\OM(A \multimap (A \with A))}^{\mathclap{\IM(A) + (\OM(A) + \OM(A))}} \\
    (x, (\ttL,\mathfrak{m})) &\mapsto (x, (\ttR, (x, \mathfrak{m})))\\
    (x, (\ttR,(y,\mathfrak{m}))) &\mapsto (y, (\ttL, \mathfrak{m}))
  \end{align*}
\end{example}

We now explain how to \enquote{plug together} two actors $\Gamma \Vdash \alpha : A$ and $\Gamma \Vdash \beta : A \multimap B$. The \emph{application}  $\beta(\alpha)$ is a sort of product\footnote{It can be seen as a variant with two-way communication of the cascade product of sequential transducers.} automaton, keeping track of the state of both $\alpha$ and $\beta$, and transferring messages between them: for instance, a message $\mathfrak{m}\in\OM(A)$ emitted by $\alpha$ is then received by $\beta$ under the form $(\ttL,\mathfrak{m}) \in \{\ttL\}\times\OM(A) \subset \IM(A\multimap B)$. The exchanges of messages between $\alpha$ and $\beta$ are treated as internal computation steps of $\beta(\alpha)$.
\begin{definition}[{special case of~\cite[Def.~3]{MAHORS}}]\label{def:actor-app}
  Let $\Gamma \Vdash \alpha : A$ and $\Gamma \Vdash \beta : A \multimap B$. For their sets of states and transition functions, we use the notations of the previous definition with indices $\alpha,\beta$. The \emph{application} $\beta(\alpha)$ has the sets of states
  \[ Q^\ominus_{\beta(\alpha)} = Q^\ominus_\alpha \times Q^\ominus_\beta \qquad Q^\oplus_{\beta(\alpha)} = (Q^\oplus_\alpha \times Q^\ominus_\beta) \cup (Q^\ominus_\alpha \times Q^\oplus_\beta) \]
  Indeed, $\alpha$ perceives $\beta$ to be part of its environment, and vice versa. For this reason, at any point during an interaction, if one of them is active, then the other must be passive.

  The initial state is $(q_{0,\alpha},q_{0,\beta})$. The transition functions are
  \[\delta^\ominus_{\beta(\alpha)} \colon ((q^\ominus_\alpha,q^\ominus_\beta),\mathfrak{m}) \mapsto (q^\ominus_\alpha, \delta^\ominus_\beta(q^\ominus_\beta,(\ttR,\mathfrak{m})))\] (redirecting incoming messages from the environment to $\beta$) and
  \begin{align*}
    \delta^\oplus_{\beta(\alpha)} \colon (q^\oplus_\alpha,q^\ominus_\beta)
    &\mapsto \left[{\begin{aligned}
      p^\oplus &\mapsto (p^\oplus,q^\ominus_\beta)\\
      \gamma(p^\oplus_1,\dots) &\mapsto \gamma((p^\oplus_1,q^\ominus_\beta),\dots)\\
      (p^\ominus,\mathfrak{m}) &\mapsto (p^\ominus,\delta^\ominus_\beta(q^\ominus_\beta,(\ttL,\mathfrak{m})))
    \end{aligned}}\right]\!\!\bigl(\delta^\oplus_{\alpha}(q^\oplus_\alpha)\bigr) \\
    (q^\ominus_\alpha,q^\oplus_\beta)
    &\mapsto \left[{\begin{aligned}
      p^\oplus &\mapsto (q^\ominus_\alpha,p^\oplus)\\
      \gamma(p^\oplus_1,\dots) &\mapsto \gamma((q^\ominus_\alpha,p^\oplus_1),\dots)\\
      (p^\ominus,(\ttL,\mathfrak{m})) &\mapsto (\delta^\ominus_\alpha(q^\ominus_\alpha,\mathfrak{m}),p^\ominus)\\
      (p^\ominus,(\ttR,\mathfrak{m})) &\mapsto ((q^\ominus_\alpha,p^\ominus),\mathfrak{m})\\
    \end{aligned}}\right]\!\!\bigl(\delta^\oplus_{\beta}(q^\oplus_\beta)\bigr)
  \end{align*}
\end{definition}

\begin{definition}\label{def:actor-transducer}
  An \emph{actor-based tree transducer} with input alphabet $\Sigma$ and output alphabet $\Gamma$ (both ranked) is specified by a type~$A$ and:
  \begin{itemize}
    \item $\overbrace{\Gamma \Vdash \beta_\sigma : A \multimap (A \multimap \dots (A \multimap A)\dots)}^{\mathclap{\text{\enquote{transition actors}, with \(\rank(\sigma)+1\) times \(A\)}}}$ for each $\sigma \in \Sigma$;
    \item $\Gamma \Vdash \beta_{\mathrm{out}} : A \multimap o$ (\enquote{output actor}).
  \end{itemize}
  On an input $t\in\Tree\Sigma$, it outputs the tree generated by $\Gamma \Vdash \beta_{\mathrm{out}}(\alpha_{t}) : o$, where we define inductively
  \[ \Gamma \Vdash \alpha_{\sigma(t_1,\dots)} = \beta_\sigma(\alpha_{t_1}) \dots (\alpha_{t_{\rank(\sigma)}}) : A\quad\text{for}\ \sigma\in\Sigma \]
\end{definition}

Let us denote the components of the actor $\beta_\sigma$ with indices: $Q^\ominus_\sigma,q^\ominus_{0,\sigma},Q^\oplus_\sigma,\delta^\ominus_\sigma,\delta^\oplus_\sigma$; and similarly for $\beta_{\mathrm{out}}$. By induction on $t \in \Tree\Sigma$, one can deduce from \Cref{def:actor-app} that:
\begin{claim}\label{clm:big-actor}
  Up to bijections that just reindex products, the states of $\beta_{\mathrm{out}}(\alpha_t)$ are:
  \begin{align*}
    Q^\ominus &\cong Q^\ominus_{\Dom(t)} \times Q^\ominus_{\mathrm{out}} \quad \text{where}\ Q^\ominus_{X} = \prod_{u \in X} Q^\ominus_{\lab_t(u)}\ \text{for}\ X \subseteq \Dom(t) \\
    Q^\oplus &\cong Q^\ominus_{\Dom(t)} \times Q^\oplus_{\mathrm{out}} \;\cup\; \bigcup_{\mathclap{u \in \Dom(t)}}  \{u\} \times Q^\oplus_{\lab_t(u)} \times Q^\ominus_{\Dom(t)\setminus\{u\}} \times Q^\ominus_{\mathrm{out}}
  \end{align*}
\end{claim}

The intuitive idea is that if the composite actor $\beta_{\mathrm{out}}(\alpha_t)$ is in an active state $x = (u,\dots)$, this means that its subprocess $\beta_{\sigma}$ corresponding to the node $u$ in $t$ (with $\sigma=\lab_{t}(u)$) is currently active while the others are passive --- and the passive states of these other subprocesses are also recorded in the $Q^\ominus_{\Dom(t)\setminus\{u\}}$ component of~$x$. We simulate $\beta_{\mathrm{out}}(\alpha_t)$ by a $\uTHM$ running on the input $t$, whose position corresponds to the current active subprocess, and which writes in memory the passive states. We then get:
\begin{lemma}\label{lem:actor-to-uthm}
  Every actor-based tree transducer $\cA$ can be translated to a $\uTHM$ $\cH$ such that the partial function $\sem\cH$ extends $\sem\cA$.
\end{lemma}

The effective construction uses the bottom-up initialization and regular lookaround features presented in \Cref{sec:basic-extensions}.

To conclude this section, we give a sufficient condition for the a priori unrestricted Hennie machine built above to be bounded-visit.
\begin{definition}\label{def:wd-actor}
  An actor $\Gamma \Vdash \alpha : A$ with states $Q^\ominus,Q^\oplus$ and transitions $\delta^\ominus,\delta^\oplus$ is \emph{weight-reducing} when one can map each state $q \in Q^\ominus \cup Q^\oplus$ to a weight $\omega(q)$ in such a way that:
  \begin{itemize}
    \item for every $q^\ominus \in Q^\ominus$ and $\mathfrak{m} \in \IM(A)$, if $\delta^\ominus(q^\ominus,\mathfrak{m})$ is defined, then its weight is strictly less than $\omega(q^\ominus)$;
    \item for $q^\oplus \in Q$, every state that appears in $\delta^\oplus(q^\oplus)$ has a weight less than or equal to $\omega(q^\oplus)$.
  \end{itemize}

  An actor-based transducer is \emph{weight-reducing} when all the actors that compose it are weight-reducing.
\end{definition}
\begin{lemma}\label{lem:actor-to-thm}
  Every weight-reducing actor-based tree transducer $\cA$ can be translated to a $\THM$ $\cH$ such that $\sem\cA \subseteq \sem\cH$.
\end{lemma}

\section{Affine $\lambda$-transducers with additive branching}
\label{sec:lambda}

The goal of this section is to show:

\begin{theorem}\label{thm:lambda-actor-hennie}
  The following devices can compute the same total tree-to-tree functions:
  \begin{enumerate}
    \item tree-to-tree affine $\lambda$-transducers with additive branching;
    \item weight-reducing actor-based tree transducers;
    \item tree-to-tree Hennie machines.
  \end{enumerate}
\end{theorem}
\Cref{lem:actor-to-thm} gives us $(2) \subseteq (3)$. We show $(1) \subseteq (2)$ in \Cref{sec:lambda-actor} and $(3) \Rightarrow (1)$ in \Cref{sec:hennie-lambda}. In \S\ref{sec:precomp-lambda}, we deduce the closure of this function class by precomposition by $\lshimtt$.

But first, we state a few definitions. As in the previous section, our grammar of types is $A,B \mathrel{::=} o \mid A \multimap B \mid A \with B$. 
We use the abbreviations (taking into account that our type system does not contain the multiplicative conjunction $\otimes$):
\[ A^{\with k} = \underbrace{A \with \dots \with A}_{\text{\(k\) times}} \qquad A^{\otimes k} \multimap B = \underbrace{A \multimap \dots \multimap A}_{\text{\(k\) times}} \multimap B\]
where $\multimap$ is right-associative. It does not matter whether `$\with$' is left- or right-associative as long as we stay consistent when encoding additive $k$-tuples as nested pairs (notation: $\pp{T_1,\dots,T_k}$).

Our typing judgments are of the form $\Phi \mid \Theta \vdash T : A$ where $\Phi$ is a context of \emph{reusable constants}, $\Theta$ is a context of \emph{affine variables}, $T$ is a $\lambda$-term and $A$ is a type. We always take the context of constants $\Phi$ to be an encoding of a ranked alphabet $\Gamma$; there are two variants: 
\begin{align*}
  \text{\emph{additive branching:}}\quad \Gamma^\with &= \{ \gamma : o^{\with\rank(\gamma)} \multimap o \mid \gamma \in \Gamma \}\\
  \text{\emph{multiplicative branching:}}\quad \Gamma^\otimes &= \{ \gamma : o^{\otimes\rank(\gamma)} \multimap o \mid \gamma \in \Gamma \}
\end{align*}
Our grammar of $\lambda$-terms, typing rules and reduction rules are standard, cf.\ e.g.~\cite[Figure~1]{MAHORS}, and are recalled in the appendix.  
By turning $\multimap$ into $\to$ and $\with$ into $\times$, we can translate our affine $\lambda$-terms to simply typed $\lambda$-terms with products, in a way that preserves reduction steps. Therefore, we inherit the strong normalization and confluence of $\beta$-reduction from the simply typed case. This ensures that every affine $\lambda$-term has a unique normal form.
\begin{claim}[{\cite[\S2.2.3]{LHORS}}]\label{clm:tree-lambda-encoding}
  The encoding that maps, for instance, the tree $a(b(c),c)$ to the $\lambda$-term $a\, \langle (b\, c),\, c \rangle$ (resp.\ $a\, (b\, c) \, c$) defines, for any ranked alphabet $\Gamma$, a bijection between
  $\Tree\Gamma$ and the $\lambda$-terms $T$ in normal form such that $\Gamma^\with \mid \varnothing \vdash T : o$ (resp.\ $\Gamma^\otimes \mid \varnothing \vdash T : o$).
\end{claim}

We now adapt the definition of a $\lambda$-transducer from~\cite{nguyen_et_al:LIPIcs.STACS.2025.68} (which is a simplification of the classical notion of higher-order transducer) to potentially use the encoding with additive branching for the output, but not for the input --- observe the $A^{\otimes\rank(\sigma)}$ below.

\begin{definition}\label{def:lambda-transducer}
  An \emph{affine $\lambda$-transducer}
  $\Tree\Sigma\to \Tree\Gamma$ is specified by a \emph{memory type}~$A$ and the $\lambda$-terms:
  \begin{itemize}
    \item $\Gamma^\odot \mid \varnothing \vdash T_\sigma : A^{\otimes\rank(\sigma)} \multimap
      A$ (\enquote{transition terms}) for $\sigma \in \Sigma$; 
    \item $\Gamma^\odot \mid \varnothing \vdash U : A \multimap o$ (\enquote{output term}).
  \end{itemize}
  where $\odot$ is equal to:
  \begin{itemize}
    \item $\with$ for a $\lambda$-transducer with \emph{additive branching};
    \item $\otimes$ for a $\lambda$-transducer with \emph{multiplicative branching}.
  \end{itemize}
  For an input tree $t \in \Tree\Sigma$, the output of the $\lambda$-transducer is the tree encoded by the normal form of $U\, T_{t}$, where we define inductively
  $T_{\sigma(t_1,\dots,t_{\rank(\sigma)})} = T_\sigma\, T_{t_1}\, \dots\, T_{t_{\rank(\sigma)}}$ for $\sigma\in\Sigma$.
\end{definition}

\begin{example}
  Here is a $\lambda$-transducer that computes $\exa$ (cf.~\S\ref{sec:running-example}).
  \begin{itemize}
    \item Its memory type is $o \with (o \multimap o)$.
    \item Its $b$-transition is $\pp{a\,\pp{b,b},\; \lambda y.\, a\, \pp{y,y}}$ and its $a$-transition is $\lambda x_1.\, \lambda x_2.\, \pp{\pi_2\, x_1\, (a\, \pp{\pi_1\, x_2,\, \pi_1\, x_2}),\; \lambda y.\, \pi_2\, x_1\, (a\, \pp{y,y})}$.
    \item Its output term is $\lambda x.\, \pi_1\, x$.
  \end{itemize}
  Note the resemblance with the $\mtt$ of \Cref{ex:mtt}. Morally, the two components $o$ and $o\multimap o$ of the memory type correspond to the $\mtt$ states $q_0$ (of rank 0) and $q_1$ (of rank 1) respectively.
\end{example}

\subsection{$\lambda$-transducers to actors: game semantics}%
\label{sec:lambda-actor}

The tree-generating game semantics from~\cite[\S3]{MAHORS} maps each affine \mbox{$\lambda$-term} $T$ to its \emph{strategy} $\sem{T}$: the set of all its possible interactions with its environment. These interactions consist of exchanging messages (i.e, playing moves, in the game metaphor) and outputting tree nodes. Furthermore, there is a partial function $\Strat$ from the actors of \S\ref{sec:actors} to strategies defined in~\cite[\S3.2]{MAHORS}. Overloading metaphors, we say that an actor $\alpha$ \emph{plays} a term $T$ when $\Strat(\alpha) = \sem{T}$.  

We recall some properties of this \enquote{playing} relation.
\begin{itemize}
  \item Every $\lambda$-term $\Gamma^\with \mid \varnothing \vdash T : A$ is played by some effectively computable actor $\Gamma \Vdash \alpha : A$~\cite[Theorem~23 in Appendix~B]{MAHORS}.
        \begin{itemize}
          \item For example~{\cite[Lemma~21 in \S{}B]{MAHORS}}, the actor from \Cref{ex:actor-diagonal} plays $\Gamma^\with \mid \varnothing \vdash \lambda x.\, \langle x,x \rangle : B \multimap (B \with B)$.
        \end{itemize}
    \item Furthermore, if $\Gamma \Vdash \beta : A \multimap B$ plays $U$, then $\beta(\alpha)$ plays $U\, T$~~\cite[Proposition~22 in \S{}B]{MAHORS}.
    \item Finally, if $A=o$, then $\alpha$ generates the tree encoded by the normal form of $T$~~\cite[Proposition~4]{MAHORS}.
\end{itemize}
Thanks to these properties, to translate an affine $\lambda$-transducer with additive branching into a actor-based tree transducer, we can take:
\begin{itemize}
  \item for each input letter, a transition actor that plays the transition term for this letter,
  \item an output actor that plays the output term.
\end{itemize}
We still have to make sure that the actors we take are weight-reducing --- which is not the case for \Cref{ex:actor-diagonal}. Fortunately, the following property suffices to do so.
\begin{proposition}
  For every actor $\Gamma \Vdash \alpha : A$ there exists a weight-reducing actor $\Gamma \Vdash \alpha' : A$ such that $\Strat(\alpha) = \Strat(\alpha')$.
\end{proposition}
\begin{proof}[Proof sketch]
  For lack of space, we do not recall the precise definition of strategies and explain the argument informally.

  $\Strat$ maps an actor to the interaction sequences that it can perform and that are valid for the game interpretation of $A$. As noted in~\cite[proof of Theorem~7]{MAHORS}, each message can occur at most once in a valid interaction in this game semantics, reflecting the affineness of the type system.
  By adding to the states of $\alpha$ a counter of the messages already received, and aborting the computation when this counter would otherwise exceed $|\IM(A)|$, we get a weight-reducing actor $\alpha'$ --- the weight of a state of $\alpha'$ is the value of the counter --- with the same valid interactions as $\alpha$. 
\end{proof}
\begin{remark}
  The \enquote{restriction} on tree stack automata considered in~\cite[Theorem~7]{MAHORS} is analogous to our bounded-visit condition.
\end{remark}

\subsection{$\THM$s to $\lambda$-transducers}%
\label{sec:hennie-lambda}

Let $\cH$ be a \emph{weight-reducing} $\THM$ with the notations of \S\ref{sec:hennie}.
We give some guiding intuitions and the definition of an affine $\lambda$-transducer with additive branching meant to be equivalent to $\cH$. The correctness proof is relegated to the appendix.   

The idea is to represent as a $\lambda$-term the \enquote{behavior} of $\cH$ when it enters a subtree $\subtree{t}{u}$ of the input, taking into account the current memory values for $u$ and its descendants. This behavior is a function of the current state. After possibly producing part of the output, $\cH$ might exit the subtree, i.e.\ move upwards from $u$ into the surrounding context $t[u \leftarrow \square]$. Upon exit, the behavior should \enquote{return} the new state and a new behavior, reflecting the new memory values, to the \enquote{caller} i.e.\ the context. This way, the new behavior can be called the next time $\cH$ enters $\subtree{t}{u}$.

One subtlety is that there may be multiple return points, because $\cH$ may produce multiple output branches in parallel. To handle this, we program in continuation-passing style: a behavior takes as argument a continuation that expects a behavior. Naively, this suggests a recursive type of behaviors, but we leverage the weight-reducing condition to unfold the recursion finitely:
\[ A_0 = o^{\with Q} \qquad A_{n+1} = (A_{n} \multimap o^{\with Q}) \multimap o^{\with Q}\quad \text{for}\ n\in\N \]
following~\cite[\S5]{MAHORS}. Tuples $\pp{t_q \mid q \in Q}$ indexed by the set of states $Q$ are encoded as nested pairs; one can think of them as functions from $Q$.
The type $A_n$ represents \enquote{$n$-truncated behaviors}, allowing for at most $n$ exits of $\subtree{t}{u}$, i.e.\ upwards moves from $u$.
Morally, this truncation is lossless when $n$ is greater than the weight of the current memory value at $u$ --- we define the weight of a memory symbol $m\in M$ as $\omega(m)=|\{m' \in M \mid m' < m\}|$.

Let us fix some default leaf symbol $\gamma_0 \in \rankK{\Gamma}{0}$.
We build truncated behaviors compositionally, using $\lambda$-terms
\[ \Gamma^\with \mid \varnothing \vdash T_{\sigma,m}^{\vec{n};k} : A_{n_1} \multimap \dots \multimap A_{n_{\rank(\sigma)}} \multimap A_k \]
defined by strong induction on $n_1 +\dots + n_{\rank(\sigma)} + k$:
\[ T_{\sigma,m}^{\vec{n};k} = \lambda \vec{z}.\, \lambda x.\, \pp{(\text{encoding of}\ \delta(q,\sigma,m))[\dots] \mid q \in Q}\]
where $[\dots]$ is a simultaneous substitution that consists of:
\begin{itemize}
  \item $(q',m',\uparrow) \leftarrow \pi_{q'}\, \bigl(x\, \bigl(T_{\sigma,m'}^{\vec{n};k-1}\, \vec{z}\bigr)\bigr)$ if $k \geqslant 1$, else $\gamma_0$\\
        Intuitively: we \enquote{return} the new state $q'$ and the new $(k-1)$-truncated behavior, where the memory at the current node has changed from $m$ to $m'$, by calling the continuation $x$.
  \item $(q',m',i)\leftarrow \pi_{q'}\, \bigl(\overbrace{z_i\, \bigl(\lambda z_{\mathrm{new}}.\, T_{\sigma,m'}^{\vec{n}';k}\, \vec{z}'\, x\bigr)}^{\mathclap{\text{or just \(z_i\) without argument when \(n_i=0\)}}}\bigr)$ for $i \in \db{\rank(\sigma)}$\\
        $\begin{aligned}
          \text{where}\ \vec{n}' &= n_1, \dots, n_{i-1}, n_i - 1, n_{i+1}, \dots, n_{\rank(\sigma)}\\
          T\, \vec{z}' &= T\, z_1\, \dots\, z_{i-1}\, z_{\mathrm{new}}\, z_{i+1}\, \dots\, z_{\rank(\sigma)}
        \end{aligned}$\\
        Intuitively: we \enquote{call} the behavior of the subtree $\subtree{t}{ui}$ rooted at the $i$-th child of the current node $u$, passing it the new state $q'$. Once it \enquote{returns} $z_{\mathrm{new}}$ (by calling the continuation that we pass), $u$ becomes the current node again. It is as if we had just entered $\subtree{t}{u}$, but the memory at $u$ has become $m'$, and the memory at $ui$ and its descendants may also have changed as reflected in the new behavior $z_{\mathrm{new}}$.
\end{itemize}
\begin{lemma}\label{lem:hennie-lambda}
  $\sem\cH$ is computed by the $\lambda$-transducer with:
  \begin{itemize}
    \item memory type $A_N$ for $N = |M|$;
    \item transition term $T_{\sigma,\top}^{N,\dots,N;N} \colon A_N^{\otimes\rank(\sigma)} \multimap A_N$ for $\sigma\in\Sigma$;
    \item output term $\lambda z.\, \pi_{q_{init}}\, (z\, (\lambda z_{\mathrm{new}}.\, \pp{\gamma_0 \mid q \in Q})) : A_N \multimap o$. 
  \end{itemize}
\end{lemma}

To prove this lemma we show (in the appendix) that $\beta$-reduction simulates the tree-generating machine for $\cH$.

\subsection{Precomposition by $\lshimtt$}%
\label{sec:precomp-lambda}

\begin{lemma}\label{lem:multiplicative}
  The functions computed by affine $\lambda$-transducers with additive (resp.\ multiplicative) branching are closed under precomposition by affine $\lambda$-transducers with multiplicative branching. 
\end{lemma}
\begin{proof}
  Let $\cA,\cB$ be two $\lambda$-transducers given respectively by:
  \begin{align*}
    \Gamma^\otimes \mid \varnothing &\vdash T_\sigma : A^{\rank(\sigma)} \multimap A & \Pi^\odot \mid \varnothing &\vdash T'_\gamma : B^{\rank(\gamma)} \multimap B \\ 
    \Gamma^\otimes \mid \varnothing &\vdash U : A \multimap o & \Pi^\odot \mid \varnothing &\vdash U' : B \multimap o
  \end{align*}
  where $\odot\in\{\with,\otimes\}$. Then $\sem{\cB}\circ\sem{\cA}$ is computed by this $\lambda$-transducer:
  \begin{align*}
    \Pi^\odot \mid \varnothing &\vdash T_\sigma[\gamma\leftarrow T'_\gamma\; \forall \gamma \in \Gamma] : A[o \leftarrow B]^{\rank(\sigma)} \multimap A[o \leftarrow B] \\ 
    \Pi^\odot \mid \varnothing &\vdash \lambda x.\, U'\,(U[\gamma\leftarrow T'_\gamma\; \forall \gamma \in \Gamma]\, x) : A[o \leftarrow B] \multimap o
  \end{align*}
  The multiplicative encoding of $\Gamma$ is crucial: it is the reason why the substitution $[o \leftarrow B]$ sends the type of $\gamma$ to the type of $T_\gamma$.
\end{proof}

This is the same argument as~\cite[Proposition~1.8]{nguyen_et_al:LIPIcs.STACS.2025.68}.   In that paper, \Nguyen and Vanoni study $\lambda$-transducers for an affine $\lambda$-calculus without `$\with$' --- in particular, the branching is multiplicative. Thus, the above proposition applies to precomposition by the devices from~\cite{nguyen_et_al:LIPIcs.STACS.2025.68}. Their main result~\cite[Theorem~1.5]{nguyen_et_al:LIPIcs.STACS.2025.68}\footnote{This statement by \Nguyen and Vanoni should be seen as a convenient rephrasing of a result of Gallot, Lemay and Salvati~\cite[Theorem~3]{LambdaTransducer} (see also~\cite[Chapter~6]{GallotPhD}).} is that tree-to-tree $\mso$ transductions are equivalent to
\[ \text{multiplicative affine \(\lambda\)-transducers} \circ \mso\ \text{relabelings} \]
\begin{theorem}\label{thm:precomp-msot}
  The functions computed by $\THM$s are closed under precomposition by $\mso$ transductions.
\end{theorem}
\begin{proof}
  By the above result of~\cite{nguyen_et_al:LIPIcs.STACS.2025.68}, it suffices to check that $\THM$s are closed under precomposition by:
  \begin{description}
    \item[MSO relabelings:] cf.~\Cref{cor:precomp-relab};
    \item[the affine $\lambda$-transducers from~{\cite{nguyen_et_al:LIPIcs.STACS.2025.68}}:]  cf.~\Cref{lem:multiplicative}. \qedhere
  \end{description}
\end{proof}

\section{Postcomposition properties of THMs}\label{sec:postcomp}

We have already mentioned that our affine $\lambda$-calculus can be directly translated to the simply typed $\lambda$-calculus with products. It is known that simply typed $\lambda$-transducers with products are regularity reflecting --- see the references given in Remark~\ref{rem:lambda-reg}. Therefore: 
\begin{corollary}\label{cor:reg-preserving}
  $\THM$s are regularity reflecting: the inverse image of a regular tree language by a $\THM$ is effectively regular.
\end{corollary}

We can then deduce that:
\begin{corollary}\label{cor:rel-postcomp}
  The class of total functions computed by $\THM$s is closed under postcomposition by bottom-up relabelings.
\end{corollary}
The proof combines the regular lookaround feature of \S\ref{sec:basic-extensions} with:
\begin{claim}\label{clm:StartingConfig}
  Let $\cH$ be a $\THM$ with $\sem\cH \colon \Tree\Sigma \to \Tree\Gamma$.
  
  The function $g_{\cH}$ that maps a configuration of $\cH$, encoded as per \Cref{def:enc-conf}, to the output tree obtained by starting $\cH$ from this configuration is itself definable by a $\THM$.

  Therefore, it preserves regular languages by inverse image.
\end{claim}

Finally, the most technical result of this section is: 

\begin{theorem}\label{thm:narrow-postcomp}
  If $\cH$ is a $\THM$ and $\cN$ is narrow-visit on the range of $\sem\cH$, then $\sem\cN \circ \sem\cH$ can be realized by some $\THM$.
\end{theorem}
\begin{proof}[Proof idea]
  We only give some high-level ideas here. The construction is detailed in the appendix.

  Essentially, the composite $\THM$ virtually evaluates $\cN$ on the tree outputted by $\cH$: when $\cN$ asks to explore this intermediate tree, $\cH$ is lazily evaluated in order to produce the relevant part of its output. This is reminiscent of the composition of \emph{reversible} two-way transducers~\cite[\S3.1]{ReversibleTransducers}. Indeed, our composite $\THM$ records enough information in memory to \enquote{rewind the computation} of $\cH$, so that we may simulate upwards moves of $\cN$ on the intermediate tree.

  However, in order to simulate the read-write memory accesses by $\cN$, we need to remember the nodes of the intermediate tree that we have already visited (and the memory values stored therein), instead of re-generating them on the fly. In other words, the composition needs to be evaluated in \enquote{call-by-need} whereas the composition of reversible two-way transducers is \enquote{call-by-name}.\footnote{The reference semantics is call-by-value: compute the full output of $\cH$, then feed it to $\cN$. It is well known that call-by-name and call-by-value differ in their observable results in the presence of effects such as mutable state. What happens in our proof is that a form of call-by-need, whose laziness allows us to respect computational bounds, turns out to be a valid optimization of call-by-value in this very specific setting.}

  We store a representation of a \enquote{global state} in $\Tree{\Gamma}[\Conf(\cH,t)]$ (cf.~\S\ref{sec:tree-generating}) of the computation of $\cH$ that is distributed among the nodes of the input $t$. The distribution reflects the origins\footnote{In a sense analogous to~\cite{bojanczykTransducersOriginInformation2014}.} of the nodes of the intermediate tree (in $\Tree\Gamma$) already generated in that global state. Since the memory at each node is finite, this means that the size of the global state must be linearly bounded by the input size. This is possible thanks to the narrow-visit condition and \Cref{cor:narrow-branches}: a branch-outputting run of $\cN$ explores $O(1)$ many branches, so it only requires a portion of the intermediate tree whose size is linearly bounded by the intermediate height; the latter is $O(|t|)$ because $\cH$ is linear size-to-height increase.

  Similarly, the composite $\THM$ is bounded-visit because $\cN$ is bounded-visit and explores $O(1)$ branches of the intermediate tree, each of them generated in a bounded-visit fashion by $\cH$.
\end{proof}

\begin{corollary}\label{cor:lhimtt-comp-thm}
  $\lhii\mtt \circ \THM = \THM$.

  As a consequence, $\bigcup_{k\in\N} \lhii\mtt^{k} \subset \THM$.
\end{corollary}
\begin{proof}
  This reduces to postcomposition by bottom-up relabelings (\Cref{cor:rel-postcomp}) and by narrow-visit $\uTHM$s (\Cref{thm:narrow-postcomp}), thanks to the decomposition of \Cref{cor:lhimtt-narrow} for $\lhii\mtt$s.
\end{proof}

\section{On the LHI-MTT composition hierarchy}\label{sec:lhimtt-comp}

As announced in the introduction, we show here:
\begin{corollary}\label{cor:lhimtt-strict}
  $\lhii\mtt^{k-1} \subsetneq \lhii\mtt^{k}$ for every $k\geqslant1$.
\end{corollary}
We actually derive this from the stronger \Cref{thm:lhitwt-vs-mtt} below, stated in terms of \emph{tree-walking tree transducers} ($\twt$s).\footnote{In the key reference~\cite{EngelfrietPebbleMacro}, they are called \enquote{0-pebble transducers}, cf.\ Remark~\ref{rem:pebble-index}.} Their classical definition may be found in e.g.~\cite[Chapter~8]{courcellebook}; equivalently, we may describe them here concisely as $\THM$s:
\begin{itemize}
  \item with a single memory symbol (the idea is that they do not write any information on the nodes: a $\twt$ configuration consists only of a state and a position);
  \item whose transitions can depend on whether the current node is the root, and if not, for which $i\in\maxrank{\Sigma}$ it is an $i$-th child.
  % (as in Remark~\ref{rem:total-root}).
\end{itemize}
In terms of expressive power, $\twt \subset \mtt$~\cite[Lemma~34]{EngelfrietPebbleMacro}, which is why \Cref{cor:lhimtt-strict} follows from the following result:
\begin{theorem}\label{thm:lhitwt-vs-mtt}
  $\lhii\twt^{k} \not\subset \mtt^{k-1}$ for every $k\geqslant1$.

  As a consequence, $\twhm \not\subset \mtt^{k-1}$.
\end{theorem}

Our proof exhibits a concrete separating function: the $k$-fold iteration $\EncPeb^k$ of a function $\EncPeb$ defined by Engelfriet and Maneth in~\cite[Section~4]{EngelfrietPebbleMacro}. More rigorously, $\EncPeb$ is a family of similar tree-to-tree functions, one for each input alphabet --- which never coincides with the output alphabet, so the composition $\EncPeb^k$ actually involves $k$ different members of the family. To avoid cumbersome notations, we ignore these subtleties concerning alphabets. For lack of space, we do not recall the full definition of $\EncPeb$; our only argument that requires examining it is for:

\begin{proposition}\label{prop:encpeb-lhi-twt}
  $\EncPeb \in \lhii\twt$.
\end{proposition}
\begin{proof}
  A tree-walking tree transducer computing $\EncPeb$ is built in~{\cite[proof of Lemma~9]{EngelfrietPebbleMacro}}.

  Ignoring the node labels, the shape of $\EncPeb(t)$ is obtained by grafting, to each node~$u$ in the tree $t$, an additional subtree which is a copy of $t$ rerooted at $u$ by tree rotations. Therefore, a downwards path from the root of $\EncPeb(t)$ to a leaf corresponds to:
  \begin{itemize}
    \item a downwards path from the root of $t$ to a node $u$,
    \item followed by the reverse upwards path from $u$ to the root,
    \item followed by a downwards path from the root of $t$ to a leaf.
  \end{itemize}
  Hence $\height(\EncPeb(t)) \leqslant 3\times\height(t)$.
\end{proof}

Other than that, we mainly use the connection that $\EncPeb$ has to \enquote{pebble tree transducers} --- a keyword that we syntactically manipulate by matching its occurrences in two theorem statements, with no need to access its definition.
\begin{theorem}[{\cite[proof of Theorem~10]{EngelfrietPebbleMacro}}]\label{prop:pebble-facto}
  Every function computed by some $k$-pebble tree transducer is in $\twt \circ \EncPeb^k$.
\end{theorem}
\begin{remark}\label{rem:pebble-index}
  We index the pebble hierarchy according to the convention from~\cite{PebbleString,EngelfrietPebbleMacro}: tree-walking transducers are 0-pebble transducers. The recent literature on polyregular functions~\cite{PolyregSurvey,gaetanPhD,Kiefer24} would call them 1-pebble transducers instead.
\end{remark}
\begin{proof}[Proof of \Cref{thm:lhitwt-vs-mtt}]
  We show that $\EncPeb^{k} \notin \mtt^{k-1}$.

  We use the following result~\cite[proof of Theorem~5]{PebbleString}: there exists a (string-to-string) function\footnote{What Engelfriet and Maneth actually prove in~\cite{PebbleString} is that $f$ is not in $\mathtt{yield}\circ\mtt^{k-1}$, but this coincides with the subclass of $\mtt^{k}$ consisting of string outputs. For more details on this point, see the discussion by Kiefer, \Nguyen and Pradic~\cite[Section~5.1]{KNP}, who prove a variation on~\cite[Theorem~5]{PebbleString} involving polyregular functions of quadratic size increase~\cite[Theorem~5.1]{KNP}.} $f \notin \mtt^{k}$ which is computed by a $k$-pebble transducer.  By \Cref{prop:pebble-facto}, we can write $f = g \circ \EncPeb^k$ for some $g \in \twt \subset \mtt$. Thus, if $\EncPeb^{k}$ were in $\mtt^{k-1}$, we would get $f \in \mtt \circ \mtt^{k-1}$, reaching a contradiction.
\end{proof}

Let us remark that together, \Cref{cor:lhimtt-comp-thm}, \Cref{prop:encpeb-lhi-twt} and \Cref{prop:pebble-facto} yield a positive result concerning the expressivity of tree-to-tree Hennie machines:
\begin{corollary}
  Every function computed by some pebble tree transducer is in $\twt \circ \twhm$.
\end{corollary}

\section{Future work}

We would like to investigate further the relationship of $\THM$s with other automata models. For instance we believe that $\twt \circ \THM$ should be a natural tree-to-tree counterpart of the \enquote{Ariadne-transducers} from~\cite{ExpReg}.
We also conjecture that $\twhm \not\subset \bigcup_{k} \mtt^{k}$ --- this is because our translation to $\lambda$-transducers seems to be intrinsically unsafe in the sense of higher-order recursion schemes.
Also, several decision problems on $\THM$s are still open, such as equivalence of two $\THM$s.

\bibliography{biblio}
\appendix

\section{Proofs for Section~\ref{sec:hennie}}

\subsection{Proof of \Cref{prop:weight-reducing}}

\begin{proof}
	We first show how to make $\cH$ weight-reducing and $N$-bounded, and secondly and to make $\cH'$ a total function with a default output.
	To make $\cH'$ weight-reducing and equal to $\cH$ on its $N$-bounded branch-outputting runs, we enrich the memory of $\cH$ with a counter initialized by $N$ at each node $u$.
	Each subsequent visit to a node $u$ decrements its counter. Finally, when the counter reaches $0$, the transition function outputs the default leaf symbol, ensuring that all branch-outputting runs are weight-reducing and $N$-bounded, while maintaining equivalence with $\cH$ on $N$-bounded branch-outputting runs.
	
	Secondly, the function $\sem{\cH'}$ is made total by first making the transition function total by defaulting to a special leaf-symbol when $\cH$ would be undefined, and secondly handling the case where $\cH$ would try to move up of its root by allowing $\cH'$ to identify the root through special memory symbols. This can be done at the first transition, since the initial configuration is at the root.
\end{proof}

\subsection{Proof of \Cref{lem:stay}}

Let $\cH$ be a $\uTHM$ with stay instructions of transition function $\delta$. If $t_0 = \delta(q_0, \sigma, m_0)$, then let us build $t_1$ by replacing each symbol $\qq{q, m, \circlearrowleft}$ in $t_0$ with $\delta(q, \sigma, m)$. Iterating this procedure to build $t_2, t_3, \ldots t_k$ deterministically either : yields a finite sequence such that $t_k$ contains no staying transition, or  : loops, which we can detect by remembering, along each branch, every pair $(q_i, m_i)$ ever reached. 
Let us then build $\cH'$, a quasi-copy of $\cH$, by removing rule $r$.
If we built a finite sequence, replace it with a new rule $\delta_{\cH'}(q^0, \sigma, m^0) = t_k$. These substitutions are the same that would have happened in a run of $\cH$, so $\sem{\cH} = \sem{\cH'}$. If, on the other end, we looped, then we do not add any transition. We know from the existence of the loop and determinism that any run that uses this rule terminate anyways, so removing the rule preserves semantics. We use this construction to remove each transition using stay instructions one by one, which will yield a $\uTHM$ without stay instructions but equivalent to $\cH$. Finally, it also (non-strictly) decreases the number of visits to any given input node and thus keeps the bounded-visit and narrow-visit restrictions satisfied.

\subsection{Proof of Lemma~\ref{lem:bottom-up-rel}}
  Informally, the $\uTHM$ $\cH'$ first writes in the nodes' memory the state computed by the bottom-up automaton. It then behaves like $\cH$, using the memory symbol in addition to the input symbol to determine the relabeling.

The first phase visits $\rank(\sigma)+1 = O(1)$ times each node with label $\sigma$. The first time we encounter this node (by going down), we just go further down to the first child if possible. At the $(k+1)$-th visit for $0\leqslant k < \rank(\sigma)$, we have already processed the subtrees rooted at the $k$ first children of the current node, recording these subtrees' relabeler states in the current node's memory, and we go on to visit the $(k+1)$-th child. At the $(\rank(\sigma)+1)$-th visit, we apply the transition function to the information we have gathered to determine the relabeling $\rho(\sigma,p)$ of the current subtree, record it in memory --- it will be used as the new label of the node ---, and move up to the parent in state $p$.

More precisely, let $\cH=(Q,M,\top,\Xi,\Gamma,q_0,\delta)$ and $R=(P,\Sigma,\Xi,\{\alpha[\sigma]\mid \sigma\in \Sigma\},\rho)$.
Let $S$ be the set of sequences of elements of $P$ of length at most $\maxrank{\Sigma}$.
We define $\cH'=(Q\cup P\cup \{init\}, M\times \Xi\cup\{\top\}\cup S,\top, \Sigma, \Gamma, init, \beta)$
where $\beta$ is decomposed into two modes, one doing a postfix reading of the input to add the relabeling information, and one simulating $\cH$ using the previously computed information. The first mode is defined as:
\begin{itemize}
\item $\beta(init,\sigma,\top)=(init,\top,1)$ if $\rank(\sigma)\neq 0$.
\item $\beta(init,\sigma,\top)=(p,(\top,\rho(\sigma,p)),\uparrow)$ if $\rank(\sigma)=0$ and $p=\alpha[\sigma]\in P$.
%\item $\beta((p_1,\ldots,p_k),\sigma,\top)=\qq{(p,(\top,\rho(\sigma,p)),\uparrow}$ if $k=\rank(\sigma)$ and $\alpha[\sigma](p_1,\ldots,p_k)=p$.
%\item $\beta((p_1,\ldots,p_k),\sigma,\top)=\qq{(init,(p_1,\ldots,p_k),k+1}$ if $k<\rank(\sigma)$.
\item $\beta(p',\sigma,(p_1,\ldots,p_k))=((p,(\top,\rho(\sigma,p)),\uparrow)$ if $k+1=\rank(\sigma)$ and $\alpha[\sigma](p_1,\ldots,p_k,p')=p$. 
\item $\beta(p',\sigma,(p_1,\ldots,p_k))=((init,(p_1,\ldots,p_k,p'),k+2)$ if $k+1<\rank(\sigma)$.
\end{itemize}
When this mode would try to move up of the root, we switch to the state $q_0$ and start the simulating mode which is defined as follows:
\begin{itemize}
\item $\beta(q,\sigma,(m,\xi))= (q',(m',\xi),d')$ where $(q',m',d') = \delta(q,\xi,m)$.
\end{itemize}

\subsection{Proof of \Cref{thm:lookaround}}
\tito{remanier preuve (harmoniser déf config + déf lookaround)}

Consider a $\uTHMr$ with the notations from \Cref{def:thm-rla}. We translate it to a $\uTHM$ whose execution simulates the original $\THMr$ step-by-step while maintaining some additional information that replaces the regular lookaround. Its set of states (resp.\ memory values) is of the form $Q' = Q\times(\dots)$ (resp.\ $M' = M\times(\dots)$),
so that the correspondence between new and old configurations amounts to applying suitable projections. Its transition function is
\[ \delta' \colon (q',m',\sigma) \mapsto \text{apply}\ \ell_{q',m',\sigma}\ \text{to every state call leaf of}\ \delta(\lambda(q',m',\sigma)) \]
where:
\begin{itemize}
  \item $\ell_{q',m',\sigma} \colon Q \times M \times \dirs{\maxrank{\Sigma}} \to Q' \times M' \times \dirs{\maxrank{\Sigma}}$ extends the calls in the original transitions with instructions to update the additional information in the state and memory;
  \item $\lambda\colon Q' \times M' \times \Sigma \to Q_\LA$ should output the correct lookaround state: for any reachable configuration, the image by $\delta_\LA$ of the encoding of the corresponding configuration of the original $\THMr$ should be equal to the image by $\lambda$ of the argument passed to the transition function.
\end{itemize}
To define $\lambda$ we decompose the computation of the lookaround state. Let $t \in \Tree{\Sigma\times M}$ be an input tree labeled with memory values for the original machine. Let $v$ be a node of $t$, with label $\qq{\sigma,m}$. Let us break down $t$ with the vertex $v$ removed into its connected components: $t = C[y \leftarrow \qq{\sigma,m}(s_{1},\dots,s_{\rank(\sigma)})]$, where $C$ is a context with a single occurrence, at position $v$, of the parameter $y$. Let us also write $s \times \{\mathtt{notHere}\}$ for the result of applying $x \mapsto (x,\mathtt{notHere})$ to every node label in $s$. We define:
\begin{align*}
  \rho_i(t,v) &= \delta_\LA[s_{i}\times\{\mathtt{notHere}\}] \in Q_\LA \quad \text{for}\ i \leqslant \rank(\sigma),\quad \text{or}\ \mathtt{null}\ \text{otherwise} \\
  \varphi(t,v) &= \delta_\LA[C \times\{\mathtt{notHere}\}] \in (Q_\LA\to Q_\LA)
\end{align*}
(The action of a single-parameter context on the states of a bottom-up automaton is defined in the usual way, with the key property that $\delta_\LA(C'[y \leftarrow s]) = \delta_\LA(C')(\delta_\LA(s))$.)

If we know this information, then we can recover the lookaround state of a configuration. Indeed, for $q \in Q$, the configuration $(v,q,t)$ of the original machine (we use here a definition of configuration whose third component is a memory-annotated input tree) is encoded as
\[(C\times\{\mathtt{notHere}\})[y \leftarrow \qq{\sigma,m,q}(s_{1}\times\{\mathtt{notHere}\},\, \dots,\, s_{\rank(\sigma)}\times\{\mathtt{notHere}\})] \]
whose image by $\delta_\LA$ is equal to
\[ \varphi(t,v)(\delta_\LA[(\sigma,m,q)](\rho_1(t,v), \dots, \rho_{\rank(\sigma)}(t,v))) \in Q_\LA \]
We choose a definition of $\lambda$ that mirrors this equation:
\[ \lambda \colon (q',m',\sigma) \mapsto \varphi'(q',m')(\delta_\LA[(\sigma,m,q)](\rho_1'(q',m'), \dots, \rho_{\rank(\sigma)}'(q',m'))) \]
where $m$ (resp.\ $q$) is the first component of $m'$ (resp.\ $q'$), for some maps
\[ \rho'_i \colon Q' \times M' \to Q_\LA \qquad \varphi' \colon Q' \times M' \to (Q_\LA\to Q_\LA) \]
For this $\lambda$ to satisfy its desired specification, it suffices to make sure that
\[ \rho'_i(q',m') = \rho(t,v) \qquad \varphi'(q',m') = \varphi(t,v) \]
whenever $q'$ and $m'$ are respectively the state and the memory at the current head position in a configuration of the new THM that projects to $(v,q,t)$.

To achieve this, we take the new sets of states and memory values to be respectively
\begin{align*}
  Q' &= Q \times (Q_\LA \cup (Q_\LA\to Q_\LA)) \\
  M' &= M \times \dirs{\maxrank{\Sigma}} \times (Q_\LA \cup \{\mathtt{null}\})^{\maxrank{\Sigma}} \times (Q_\LA\to Q_\LA)
\end{align*}
For $q' = (q,p) \in Q'$ and $m' = (m, d, r_{1}, \ldots, r_{\maxrank{\Sigma}}, f) \in M'$, we define
\[\rho'_i(q',m') = \begin{cases}
    p & \text{if}\ d = \downarrow_i\\
    r_i & \text{otherwise}
  \end{cases}
  \qquad\qquad
  \varphi'(q',m') = \begin{cases}
    p & \text{if}\ d = \uparrow\\
    f & \text{otherwise}
  \end{cases}
\]
In other words, all the lookahead information can be read from memory, except for one value (and $d$ indicates which one): it is the responsibility of the previous transition to compute this missing value and record it in $p$ as part of the state.

Let us now turn to the transitions. If the original $\uTHMr$ moves in a direction $d''$, then, informally, the new $\uTHM$ also does so in the step-by-step simulation. Furthermore:
\begin{itemize}
  \item In the directional component of the memory, it records $d''$. Indeed, if $d''=\downarrow_i$, then the machine might mutate the memory of the $i$-th child of the current node $v$ before the head position comes back to $v$ (if it ever does). Thus, we cannot predict the lookaround information for the $i$-th subtree at the next visit of $v$, so we record the direction $\downarrow_i$ as our one allowed exception. A similar analysis holds when $d=\uparrow$ regarding the context above $v$.
  \item However, since the moves of the head are local, the memory of the rest of the tree (that is, outside of the aforementioned subtree or context) is not going to change before visiting $v$ again. This is why the new machine also records $\rho_1'(q',m'), \dots$ to memory, to be used at the next visit. (Here $(q',m',\sigma)$ is the argument passed to the transition function.)
\end{itemize}
Formally, recall that the transition function involves a relabeling of state call leaves; we define it as:
\[ \ell_{q',m',\sigma}(q'',m'',d'') = ((q'',p''),(m'',d'',r''_1, \dots, r''_{\rank(\sigma)}, f''), d'') \]
where $r''_i = \rho_i'(q',m')$, $f'' = \varphi'(q',m')$ and
\[ p'' =
  \begin{cases}
    \delta_\LA[(\sigma,m'',\mathtt{notHere})](r''_{1},\dots,r''_{\rank(\sigma)}) \in Q_\LA &\text{if}\ d'' = \uparrow\\
    f'' \circ \delta_\LA[(\sigma,m'',\mathtt{notHere})](r''_{1},\dots,r''_{i-1},-,r''_{i+1},\dots,r''_{\rank(\sigma)}) &\text{if}\ d'' = \downarrow_i
  \end{cases}
\]
Let us explain the case $d'' = \uparrow$ of this last equation (the other case is similar). Let $(v,q',t')$ be a reachable configuration, with $(\sigma,m')$ being the label of $v$ in $t'$. By moving in the direction $\uparrow$ we go to the parent $w$ of the node $v$. It must be the case that $w$ has already been visited before (in order to reach $v$ from the root), and the last visit to $w$ has concluded with a move towards the child $v$, in direction $\downarrow_i$ for some $i$. The responsibility of the transition from $(v,q',t')$ to the successor configuration $(w,(q'',p''),s')$ is therefore to make sure that $p'' = \rho_i(s,w)$ where $s$ is the projection of $s'$. By unfolding the definitions, and noting that $s'$ and $t'$ differ only at the node $v$, one can check that the above choice for $p''$ works.

A final detail: our new $\uTHM$ uses the bottom-up initialization feature from \Cref{lem:bottom-up-rel}. If the input subtree rooted at a node is $\sigma(t_1,\dots,t_k)$ then its initial memory is
\[ (\top, \delta_\LA[t_1 \times \{(\top,\mathtt{notHere})\}],\dots,\delta_\LA[t_k \times \{(\top,\mathtt{notHere})\}], \underbrace{\mathtt{null}, \dots, \mathtt{null}}_{\mathclap{\text{\(\maxrank{\Sigma}-k\) times}}}, \mathrm{id}) \]
This guarantees that the correspondence between $\rho,\varphi$ and $\rho',\varphi'$ holds at the first visit of each node.

\section{Proofs for Section~\ref{sec:mtt}}

We prove first the first statement of \Cref{thm:mtt-hennie}, namely that any $\mtt$ can be translated into an equivalent $\uTHM$.
\begin{proof}[Proof of \Cref{thm:mtt-hennie}]
%%%Sketch of proof

Before formally defining the equivalent $\uTHM$ $\cH$, we explain how it simulates $\cM$.
%Let $\cM=(Q,\Sigma,\Gamma,q_0,h,R)$.
The main difference between $\mtt$ and $\uTHM$ is that the $\mtt$ processes its input once while having pointers to different parts of its output, while the $\uTHM$ can process its input several times, but only produces its output in a top-down fashion. On the other hand, a $\uTHM$ has access to a finite memory for each node and for each branch of the output.
In order to simulate $\cM$, the $\uTHM$ $\cH$ applies, for each branch, an outside-in evaluation process, meaning that it always evaluates the topmost call state of a branch. It also stores in each node the sequence of remaining calls to be processed.

$\cH$ starts simulating $\cM$ at the root, applying its initial transition.
When $\cH$ reaches a node from its parent, it does so in a given state $q$ of $\cM$. 
Together with the node label $\sigma$, it gives a transition $(q,\sigma)$ of $\cM$, where $r\in \Tree{\Gamma\cup\qq{Q,X_j}}{(Y_k)}$ with $j$ being the rank of $\sigma$ and $k$ being the rank of $q$.
The corresponding transition of $\cH$ produces a tree where each branch $b$ is truncated at the first $\qq{Q,X_j}$ symbol $\qq{q',x_i}$.
It launches a head in state $q'$ on direction $i$, and writes on its node the truncated part of the branch $b$.
If there is no $\qq{Q,X_j}$ symbol, then either the branch is done and nothing else needs to be done, or there is a $Y_i$ symbol, in which case the reading head moves up, back to its parent, in a state indicating the branch $i$ being produced.

When $\cH$ reaches a nodes from one of its children, then it contains a remainder of  the branch to produce, and the state indicates which branch to follow. It then acts similarly to above, according to the state at the top of the branch, minus the state call at the root whose simulation has finished.
The computation stops when all branches have reached their respective ends.%%top phrase !

%%%%Formal construction

Formally, let $\cM=(Q,\Sigma,\Gamma,q_0,\delta)$. 
We define a $\uTHM$ $\cH=(P,N,\top,\Sigma,\Gamma,p_0,\alpha)$ as follows.
We set the set of states $P=Q\cup\db{\maxrank{Q}}$ and $p_0=q_0$.
To define the memory $N$, we need to consider the possible productions of $M$.
We first consider $Out$ to be the set of output trees $\delta(q,\sigma)$ for $q,\sigma$ in $Q,\Sigma$. 
Let $Out'$ be the set of subtrees of elements of $Out$ that are rooted at an element of $\qq{Q,X_j}$. The set of memory symbols $N$ is finally defined as $\{\top\}\cup Out'$.

It remains to define the transition function $\alpha$:
\begin{itemize}
\item $\alpha(q,\sigma,\top)$ produces the tree $\delta(q,\sigma)$ where each branch is truncated before the first symbol that does not belong to $\Gamma$, and for each truncated subtree $t$, we set a processing head:
\begin{itemize}
\item $\qq{q,t,j}$ if the top symbol belongs is $\qq{q,x_j}$,
\item $\qq{i,\top,\uparrow}$ if $t$ is reduced to a parameter $Y_i$.
\end{itemize} 

\item $\alpha(i,\sigma,t)$ behaves similarly to the first case, using the $i$-th branch of $t$ (without its root) instead of $\delta(q,\sigma)$.
Note that in this case $\alpha(i,\sigma,t)$ nonetheless produces the possible finite $\Gamma$ branches attached to the root of $t$.
\item the case $\alpha(q,\sigma,t)$ is undefined as it does not appear since $\cH$ only reaches a node in a state in $Q$ if its memory is emptied (i.e. in memory state $\top$).
\item Note that the $\uTHM$ can only reach a node in a state $i$ coming from one of its children, so the case $(i,\sigma,\top)$ does not happen either.
\end{itemize}

%%%%Correctness

We now prove the correctness of the construction, namely that $\sem{\cM}=\sem{\cH}$.
As mentioned earlier, we in fact prove that $\cH$ realizes an outside-in evaluation of calls of $\cM$ on every branches, meaning that it produces its output similarly to $\cM$ does by always applying the topmost derivation possible on each branch, called the outside-in derivation.
This proves the correctness of simulation since all derivations of $\mtt$s lead to the same output.

In order to prove this, we describe what outside-in partial derivations of $\cM$ look like.
The $\mtt$ $\cM$ processes its input in a top-down fashion, stacking state calls on the partial output. 
A state call $\qq{q,s}$ is replaced, by a derivation step, by a tree whose state calls are children of $s$.
If the topmost state call is processed each step (branch-wise), then 
the state calls of children of a node are processed before its own state calls.
The resulting derivation is a tree where for any branch, all state calls of siblings of a given node are grouped together. More, the depths of state calls in the derivation tree are inverted compared to the depths of the node they point to.

Let us consider a partial output $t'$ of $\cM$ over some input $t$, meaning that $\qq{q_0,t}\Rightarrow_M^* t'$, obtained using the outside-in derivation.
We show that $\cH$ reaches a configuration where it has produced the maximal prefix $t''$ of $t'$ that belongs to $T_\Gamma$, 
and for each topmost node $n=\qq{q,s}$ appearing of $t'$ (hence not in $t''$), $\cH$ has a processing head on $s$ in state $q$. 
Moreover for each node $s'$ and such node $n$, the memory of $s'$ for the branch denoted by $n$ is $\top$ if $s'$ is a descendant of $s$, and otherwise contains the maximal subtree of $t''$ below $n$ containing symbols of $\Gamma$ and all state calls on its children, and such that its root is a state call.

At the start of the computation of $\cM$, nothing is produced and the initial configuration of $\cH$ satisfies this property.

Now assume that property holds true at some point of the computation, and let $n=\qq{q,s}$ be the topmost node of some branch $b$ of $t'$ that does not belong to $\Gamma$. 
The outside-in derivation of $M$ on this branch applies to $n$ and replaces $\qq{q,s}$ by $\delta(q,\lab(s))$.
By induction hypothesis, $\cH$ has a processing head of $s$ in state $q$.
Since $n$ is the topmost state call and points to $s$, there is no state call related to any children of $s$, so the memory state of $s$ is $\top$.
The derivation of $\cM$ replaces $\qq{q,s}$ with the context $c=\delta(q,\lab(s))$ and plugs it in $t'$ to form a new partial output $t''$.
Similarly, the corresponding transition $\alpha(q,\lab(s),\top)$ produces the $\Gamma$ part of $c$.
For any branch of $c$ with a state call, let $n'=\qq{q',si}$ be its topmost state call. 
$\cH$ launches a processing head on $si$ in state $q'$, and stores on $s$ the (possibly empty) tree below $n'$, satisfying the invariant for the unfinished branch.
If a branch of $c$ is entirely comprised of $\Gamma$-symbols, the branch  is entirely computed and the invariant is also satisfied.
Finally, if a branch of $c$ ends in some parameter $Y_i$, the possible topmost state call of $t''$ of this branch relates to a node $s'$ above $s$.
The machine $\cH$ moves up, in state $i$, up to the first ancestor node that still contains a state call in its memory state, possibly updating the state $i$. Thanks to the invariant, it corresponds to $s'$ and the memory state of $s'$ contains the remainder of the branches to produce and the state indicates which branch is being produced. The transition applied here then launches a processing head on the child of $s'$ corresponding to the state call, in the corresponding state, and removes from $t$ the produced part, satisfying the invariant.

This proves that $\sem{\cM}\subseteq \sem{\cH}$ when seen as relations.
And since $\cM$ and $\cH$ realizes total functions, this means that $\sem{\cM}=\sem{\cH}$, which concludes the proof.
\end{proof}

The second point of \Cref{thm:mtt-hennie} is a consequence of the more general following property:
each call of the machine $\cM$ to a node $u$ corresponds exactly to one downward visit of the machine $\cH$ to the same node $u$.
Indeed by construction of $\cH$, each downward visit of $\cH$ to a node $u$ consumes a call stored by its parent. As the memory is constructed from transitions of $\cM$ that are yet to be expanded, we get that calls of $\cM$ are in bijection to downward visits of $\cH$.
Then the number of downward visits at an antichain $S$ of nodes of an input $t$, for a branch-outputting run of $\cH$, is the number of calls of $\blacksquare$ that can be stored in nodes of $S$ by the branch-outputting run, which is bounded by the maximal number of $\blacksquare$ that can be stacked during a computation of $\cM^\blacksquare$.

\section{Proofs for Section~\ref{sec:msosi}}

The goal of this section is to construct, given a finite visits $\twhm$, an $\msosi$ 
that realizes the same function.
The main ingredient of the proof is the use of local profiles, which are abstractions of the run of a $\twhm$ at a given input position regarding a given output position.
Thanks to the finite visits property, local profiles are bounded. 
We use them for the coloring $X_i$ of the $\mso$ set interpretation.
A global profile of an input then associates to each input position a local profile. The domain formula of the $\msosi$ states that a global profile is coherent, meaning that all local profiles concern a same output position. 

To define the successor and ancestor formulas,
we rely on the fact that the local profiles for a given output position are extensions of the local profiles for its ancestors, as $\twhm$ produces their output in a top-down fashion. In particular, this means that a partial run up to producing a given output position has to produce all its ancestors.

For the remainder of this section, we fix a (finite visits) $ \twhm $ $\cH = (Q, M, \top, \Sigma, \Gamma, q_{init}, \delta) $. We denote by $ B$ (resp. $PB$) the set of (resp. prefix) branches in the right hand side of transition rules of $\cH$, i.e.\ the set $\bigcup_{q,\sigma,m} B(\delta(q, \sigma, m))$ (resp. $\bigcup_{q,\sigma,m} PB(\delta(q, \sigma, m)))$).

\subsection{Hennie Profiles}

A node of the output tree $\cH(t)$ for some input tree $t$ is uniquely characterized by the branch-outputting run leading to it.
Said branch is determined by a sequence of transitions of $\cH$ on $t$, 
i.e. by an input node, a state and a memory state.
Should we call a \emph{visit pair} a pair composed by a state of $Q$ and a branch of $B_P$,
the trace left on an input node of producing an output node is a finite sequence of visit pairs.
Formally, we set:
\begin{definition}[Local profiles]
For $\sigma \in \Sigma$,
 a \emph{local $ \cH $ Hennie profile over $\sigma$} 
is a finite sequence $(q_1,b_1),\ldots (q_k,b_k)$ of visit pairs that
satisfies the following conditions: 
\begin{enumerate}
\item $k$ is smaller of equal to maximal number of visits of $\cH$.
\item for all $0< i < k-1 $, $b_i$ is a branch of $B$ and its leaf is a reading head,
\item $ b_0 \in B(\delta(q_0, \pp{\sigma, \top})) $
, and 
\item $ \forall i<k-1, b_{i+1} \in B(\delta(q_{i+1}, \pp{\sigma, m_i}))$, where $m_i$ is defined inductively as the memory set by the previous transition. 
\end{enumerate}
We call $\Lambda_\cH(\sigma)$ the set of local $\cH$ Hennie profiles over $\sigma$, or simply $\Lambda(\sigma)$ when $\cH$ is clear from context, and $\Lambda_\cH=\cup_{\sigma\in \Sigma} \Lambda_\cH(\sigma)$ the set of all local $\cH$ Hennie profiles.
\end{definition}

Intuitively, a local profile requires that at each visit of the input position, the memory state corresponds to the one written by the last visit (or $\top$ for the first visit).
We also require that each visit but potentially the last selects a complete and non terminating branch of the production of the transition.

Producing a an output node, and the branch leading to it, might require to visit several if not all nodes of the input, generating local profiles for all input nodes. This gives rise to the notion of global profiles:
\begin{definition}[Global profiles]
A \emph{global $\cH$ Hennie profile $\alpha$ over t} is a mapping from nodes of $t$ to local $\cH$ Hennie profiles respecting the condition: $ \forall u \in \Dom(t), \alpha(u) \in \Lambda_\cH( \lab_t(u) ) $. We call $ N_\cH(t) $ the set of global $ \cH $ Hennie profiles over $t$.
\end{definition}

Global profiles assigns to each node a local profile corresponding to its label.
The next notion allows us to describe global profiles that effectively characterize partial runs.

\begin{definition}[Execution order]
Consider a global Hennie profile $ \alpha $.
For $u$ a node of $t$ and $i$ an integer, we denote by $(u,i)$ the $i$-th visiting pair $(q_i,b_i)$ of $\alpha(u)$ if it exists.
An \emph{execution order} of $mnu$ is a total ordering $\leqslant_e$ of all visiting pairs $(u,i)$ such that:
\begin{itemize}
\item For all $u$, $(u,i)\leqslant_e (u,j)$ iif $i\leqslant j$,
\item For all visiting pair $(q,b)$ of some node $u$, with the leaf of $b$ being some head $\qq{q',m,d}$ if, and only if, there exists some visiting pair $(q',b')$ in $ud$ and $(q',b')$ is the $\leqslant_e$-successor of $(q,b)$.
\end{itemize}
\end{definition}

We finally define coherent full profiles that we later prove to be the colorings we need to define the $\msosi$ transformation.
\begin{definition}[Coherence]
A global profile is said to be coherent if it admits an execution order whose smallest element is the first visiting pair of the root, with state $q_{init}$.
\end{definition}
\begin{proposition}\label{prop:coherent_endpoint}
In a coherent global profile $\nu$ over $t$, there exists a unique $u \in \Dom(t)$ that contains a branch $b$ which endpoint is not labeled in $\qq{Q, M, D}$.
\end{proposition}
\begin{proof}
A pair $(q,b)$ has a successor if and only if the endpoint of $b$ is labeled in $ \qq{Q, M, D} $, but the order $\leqslant_e$ needs to be total, so there must be exactly one visiting pair whose endpoint is not in $ \qq{Q, M, D} $ : the maximum $(u, i)$ of $\leqslant_e$.
\end{proof}

The previous definitions amount to this key lemma that links coherent global profiles with output nodes.

\begin{lemma}\label{lem-GlobalProfPBranches}
Given an input $t$ of $\cH$, there is a bijection between the output nodes of $\sem{\cH}(t)$ (or equivalently branch-outputting runs with a distinguished node in the last branch) and
the set of coherent $\cH$-profiles over $t$. 
\end{lemma}
\begin{proof}
First, let $\alpha :\Dom(t)\to \Lambda_\cH$ be a coherent global profile with execution order $\leqslant_e$.
Using \Cref{prop:coherent_endpoint}, there exists exactly one node $u$ with a visiting pair $(q,b)\in \alpha(u)$ such that the endpoint of $b$ is an element of $\Gamma$.
We prove by induction on the execution order that the concatenation of the branches of the visiting pairs along the execution order is a branch-outputting run of $\sem{\cH}(t)$ with the last transition potentially truncated.
Then by \Cref{prop:coherent_endpoint} the endpoint of this truncated branch-outputting run is an element of $\Gamma$, thus an output node of $\sem{\cH}(t)$.
Initially, the smallest element of $\leqslant_e$ is, by definition of coherence, the first visiting pair $(q_{init},b)$ of the root of $t$, where $b$ is a (potentially truncated) branch of $\delta(q_{init},\lab(\varepsilon),\top)$.
Hence at the first step $b$ is a possibly truncated branch-outputting run of the initial transition of $\cH$ on $t$.

Now suppose that we have constructed a branch-outputting run $b_1\ldots b_n$ up to some visiting pair $(q,b_n)$ of some node $u$ of $t$, and let $\qq{q',m,d}$ be the endpoint of $b_n$.
Note that if the endpoint of $b$ is an element $\gamma$ of $\Gamma$, by \cref{prop:coherent_endpoint} we have exhausted the execution order and we can conclude.
By definition of the execution order there is a visiting pair $(q',b')$ in $ud$ that is the successor of the visiting pair $(q,b_n)$. 
Then $b'$ is a branch (or a prefix of a branch, in which case we can also conclude) of $\delta(q',\lab(ud),m')$ where $m'$ is the last memory state written at $ud$ by the previous visiting pair (or $\top$ if it is the first visiting pair of the node).
We concatenate $b'$ to the already produced branch-outputting run by replacing $\qq{q',m,d}$ by $b'$. As we followed the transition of $\cH$, we successfully extended our induction.

Conversely, 
given an output node $o$ of $\sem{\cH}(t)$, we construct a coherent global profile as follows. 
Using the branchwise semantics of $\twhm$, consider the unique sequence $x_0 \xrightarrow{b_0} \ldots \xrightarrow{b_{n-1}} x_n$ where $x_0=C_{init}$ and $\pi_2(b_0\ldots b_{n-1} \gamma)$ is the step by step branch computation from the initial configuration up to the computation step that output node $o$. Let $x_i=(u_i,q_i,t_i)$ for all $i<n$
Then we construct the local profiles as follows: 
local profiles are initialized as the empty sequence.
Then for $i$ ranging from $0$ to $n-2$, we add to the node $u_i$ the visiting pair $(q_i,b_i)$, and we add to $u_{n-1}$ the pair $(q_{n-1},b')$ where $b'$ is $b_{n-1}$ truncated at $o$.
We prove that the set of sequences we defined is a coherent global profile.
First, they are local profiles since :
\begin{enumerate}
\item each profile has a length corresponding to the number of visits of the node, it is hence bounded by the maximal number of visits.
\item Each $b_i$ is one computation step of $\cH$, so it is a full, branch of a right-hand side of $\delta$, and ends with a reading head except for the last configuration $x_n$,
\item The first visit of a given input node has the memory set to $\top$,
\item The branch $b_i$ is the result of applying the transition of the reading head in state $q_i$, with the memory set by the last visit of the position.
\end{enumerate}

It is then a global profile since all visiting pairs of the local profile apply transitions of the label of the corresponding node.
The execution order is naturally defined by the sequence of $b_i$ ordering their corresponding pairs. By construction it satisfies both conditions of the definitions.
And finally it is coherent since we start at the initial configuration.
\end{proof}

We are now able to define the labeling formulas $\phi_\gamma$.
\begin{lemma}\label{lem:MSOSIdomainFormula}
We can construct $\mso$ formulas $\phi_\gamma(\{X_\lambda\mid \lambda\in \Lambda_\cH\})$, for $\gamma\in\Gamma$, 
such that for any input $t$ and assignment 
$\alpha : \Dom(t)\to \Lambda_\cH$, $(t,\alpha)\models \phi_\gamma(\{X_\lambda\mid \lambda\in \Lambda_\cH\})$ if, and only if, $\alpha$ is a coherent global profile in bijection with an output node labeled by $\gamma$.
\end{lemma}
\begin{proof}
We first define a formula $\phi(\{X_\lambda\mid \lambda\in \Lambda_\cH\})$ such that
$(t,\alpha)\models \phi$ if, and only if, $\alpha$ is a coherent global profile.
First, the formula states that the sets $X_i$ form a partition of the domain of $t$.
The predicate $x\in^ X_\lambda$ is interpreted as the node quantified by $x$ is labeled by local profile $\lambda$.
The global profile condition is simply the subformula $\psi_{gp}=\forall x, \wedge_{\sigma\in\Sigma} \big(\sigma(x) \to \vee_{\lambda\in\Lambda_\cH(\sigma)} x\in X_\lambda\big)$.
Next, we define the execution order via a formula $\psi_{i,j}(x,y)$  that is satisfied if the successor in the execution order of the $i$-th visiting pair of the profile of $x$ is the $j$-th visiting pair of $y$.
As the successor in the execution order is always a neighbor (self, child or parent) and only depends on the local profile of the two which is finite information, the formula $\psi_{i,j}(x,y)$ can be written as a disjunction of types on $x$ and $y$ in its direct neighborhood.
Then the existence of the execution order is the formula stating that the transitive closure the successor formula is a total order, i.e. transitive, reflexive and antisymmetric, which are all $\mso$ definable properties.
Finally, the coherence is defined by stating that the minimal element of the execution order is the first visiting pair of the root.
The label $\gamma$ correspond to the label of the maximal element of the execution order.
\end{proof}

%\begin{theorem}\label{thm:TWHMtoMSOSI}
%Given a $\twhm$ $\cH$, we can effectively construct an $\msosi$ $T$ such that
%$\sem\cH=\sem{T}$.
%\end{theorem}
We conclude this Section by giving the reduction from $\THM$ to $\msosi$.
\begin{proof}[Proof of \Cref{thm:THMtoMSOSI}]
Similarly to the proof of \Cref{lem:MSOSIdomainFormula}, we set $c=|\Lambda_\cH|$, the number of second order variables of $T$ is equal to the number of local profiles.
The labeling formulas $\phi_\gamma$ are taken from \Cref{lem:MSOSIdomainFormula}.

Finally, we need to define the $\phi_i$ formula interpreting the $i$-th child relation as well as the $\phi_{anc}$ formula interpreting the ancestor relation.
A global profile $\alpha$ describes the ancestor of an other global profile $\alpha'$ if, and only if, for all node $u$, $\alpha'(u)$ is obtained from $\alpha(u)$ by 
 adding visiting pairs to it, and by enriching the maximal element of the execution order. This is a local property that can be tested separately on each node, verifying that $\alpha'(u)$ is indeed an extension of $\alpha(u)$, and both are coherent global profiles.
Similarly, a coherent global profile $\alpha'$ describes the $i$-th child of an other coherent global profile $\alpha$ if $\alpha'$ is obtained from $\alpha$ by replacing the maximal element of the execution order $(q,b)$ with $b$ ending in some $\gamma$ by $(q,b')$ where $b'$ is $b$ enriched with the $i$-th child in the transition $\delta(q,\pp{\sigma,m})$.
Then either the leaf of $b'$ is an element of $\Gamma$, or
the local profiles of $\alpha'$ are local profile of $\alpha$ enriched with visiting pairs whose branches are reduced to a single node.
The fact that $\alpha'$ is a coherent global profile assures then that all these branches added are reduced to a single reading head, except for the maximal element.
As these properties are $\mso$ definable by a disjunction over the finite set of local profiles, we are able to define formulas $\phi_i$ for the successor.
\end{proof}

\section{Proofs for Section~\ref{sec:actors}}

\subsection{Proof and further details for Claim~\ref{clm:big-actor}}

We first look at the $k$-ary application of actors in general. Recall that $\OM(A \multimap B) = \{\ttL\}\times\IM(A) \cup \{\ttR\}\times\OM(B)$ and that $\IM(A \multimap B)$ is defined similarly. We use the abbreviation $\mathtt{arg}_i\mathfrak{m} = (\ttR, \dots (\ttR, (\ttL, \mathfrak{m}))\dots)$ with $i-1$ times $\ttR$.

\begin{claim}
  Let $k\in\N$. Let $\Gamma \Vdash \alpha : B_1 \multimap \dots \multimap B_k \multimap A$ and $\Gamma \Vdash \beta_i : B_i$ for $i \in \db{k}$. Then:
  \begin{footnotesize}
  \[ Q^\ominus_{\beta(\alpha_1)\dots(\alpha_k)} = Q^\ominus_\beta \times \prod_{i=1}^k Q^\ominus_{\alpha_i} \qquad Q^\oplus_{\beta(\alpha_1)\dots(\alpha_k)} = Q^\oplus_{\beta} \times \prod_{i=1}^k Q^\ominus_{\alpha_i} \cup \bigcup_{i=1}^k Q^\ominus_\beta \times Q^\ominus_{\alpha_1} \times \dots \times Q^\ominus_{\alpha_{i-1}} \times Q^\oplus_{\alpha_i} \times Q^\ominus_{\alpha_{i+1}} \dots \times Q^\ominus_{\alpha_k} \]
  where $\cong$ denotes bijections that only reindex products. Intuitively, at any time, at most one of the subprocesses is active.
  \end{footnotesize}

  The initial state is $(q_{0,\beta}^\ominus,q_{0,\alpha_1}^\ominus,\dots,q_{0,\alpha_k}^\ominus)$ and the transitions are:
  \[\delta^\ominus_{\beta(\alpha_1)\dots(\alpha_k)} \colon ((q^\ominus_\beta,q^\ominus_{\alpha_1},\dots,q^\ominus_{\alpha_k}),\mathfrak{m}) \mapsto (\delta^\ominus_\beta(q^\ominus_\beta,(\ttR,\dots(\ttR,\mathfrak{m})\dots)),q^\ominus_{\alpha_1},\dots,q^\ominus_{\alpha_k})\ \text{with}\ k\ \text{times}\ \ttR \]
  \begin{footnotesize}
    \begin{align*}
      \delta^\oplus_{\beta(\alpha_1)\dots(\alpha_k)} \colon (q^\ominus_\beta, q^\ominus_{\alpha_1}, \dots, q^\oplus_{\alpha_i}, \dots q^\ominus_{\alpha_k})
      &\mapsto \left[{\begin{aligned}
        p^\oplus &\mapsto (q^\ominus_\beta, q^\ominus_{\alpha_1}, \dots, p^\oplus, \dots q^\ominus_{\alpha_k})\\
        \gamma(p^\oplus_1,\dots) &\mapsto \gamma((q^\ominus_\beta, q^\ominus_{\alpha_1}, \dots, p_1^\oplus, \dots q^\ominus_{\alpha_k}),\dots)\\
        (p^\ominus,\mathfrak{m}) &\mapsto (\delta^\ominus_\beta(q^\ominus_\beta,\mathtt{arg}_i\mathfrak{m}), q^\ominus_{\alpha_1}, \dots, p^\ominus, \dots q^\ominus_{\alpha_k})
      \end{aligned}}\right]\!\!\bigl(\delta^\oplus_{\alpha_i}(q^\oplus_{\alpha_i})\bigr) \\
      (q^\oplus_\beta, q^\ominus_{\alpha_1}, \dots, q^\ominus_{\alpha_k})
      &\mapsto \left[{\begin{aligned}
        p^\oplus &\mapsto (p^\oplus,q^\ominus_{\alpha_1}, \dots, q^\ominus_{\alpha_k})\\
        \gamma(p^\oplus_1,\dots) &\mapsto \gamma((p^\oplus,q^\ominus_{\alpha_1}, \dots, q^\ominus_{\alpha_k}),\dots)\\
        (p^\ominus,\mathtt{arg}_i\mathfrak{m}) &\mapsto (p^\ominus,q^\ominus_{\alpha_1}, \dots,\delta^\ominus_\alpha(q^\ominus_{\alpha_i},\mathfrak{m}),\dots,q^\ominus_{\alpha_k})\\
        (p^\ominus,(\ttR,\dots(\ttR,\mathfrak{m})\dots)) &\mapsto ((p^\ominus,q^\ominus_{\alpha_1}, \dots, q^\ominus_{\alpha_k}),\mathfrak{m})
      \end{aligned}}\right]\!\!\bigl(\delta^\oplus_{\beta}(q^\oplus_\beta)\bigr)
    \end{align*}
  \end{footnotesize}
\end{claim}
\begin{proof}
  The case $k=0$ is tautological, while $k=1$ is the definition of application, up to the bijection that swaps pairs. Inductive case:
  \[ Q^\ominus_{\beta(\alpha_1)\dots(\alpha_{k+1})} = Q^\ominus_{\alpha_{k+1}} \times Q^\ominus_{\beta(\alpha_1)\dots(\alpha_{k})} = Q^\ominus_{\alpha_{k+1}} \times \left( Q^\ominus_\beta \times \prod_{i=1}^k Q^\ominus_{\alpha_i} \right) \cong Q^\ominus_\beta \times \prod_{i=1}^{k+1} Q^\ominus_{\alpha_i} \]
  \begin{footnotesize}
    \begin{align*}
      Q^\ominus_{\beta(\alpha_1)\dots(\alpha_{k+1})}
      &= Q^\oplus_{\alpha_{k+1}} \times Q^\ominus_{\beta(\alpha_1)\dots(\alpha_k)} \cup Q^\ominus_{\alpha_{k+1}} \times Q^\oplus_{\beta(\alpha_1)\dots(\alpha_k)}\\
      &= Q^\oplus_{\alpha_{k+1}} \times Q^\ominus_\beta \times \prod_{i=1}^k Q^\ominus_{\alpha_i} \cup Q^\ominus_{\alpha_{k+1}} \times \left( Q^\oplus_{\beta} \times \prod_{i=1}^k Q^\ominus_{\alpha_i} \cup \bigcup_{i=1}^k Q^\ominus_\beta \times Q^\ominus_{\alpha_1} \times \dots \times Q^\ominus_{\alpha_{i-1}} \times Q^\oplus_{\alpha_i} \times Q^\ominus_{\alpha_{i+1}} \dots \times Q^\ominus_{\alpha_k} \right)\\
      &\cong Q^\ominus_\beta \times \prod_{i=1}^k Q^\ominus_{\alpha_i} \times Q^\oplus_{\alpha_{k+1}} \cup Q^\oplus_{\beta} \times \prod_{i=1}^{k+1} Q^\ominus_{\alpha_i} \cup \bigcup_{i=1}^k Q^\ominus_\beta \times Q^\ominus_{\alpha_1} \times \dots \times Q^\ominus_{\alpha_{i-1}} \times Q^\oplus_{\alpha_i} \times Q^\ominus_{\alpha_{i+1}} \dots \times Q^\ominus_{\alpha_{k+1}} \\
      &\cong Q^\oplus_{\beta} \times \prod_{i=1}^{k+1} Q^\ominus_{\alpha_i} \cup \bigcup_{i=1}^{k+1} Q^\ominus_\beta \times Q^\ominus_{\alpha_1} \times \dots \times Q^\ominus_{\alpha_{i-1}} \times Q^\oplus_{\alpha_i} \times Q^\ominus_{\alpha_{i+1}} \dots \times Q^\ominus_{\alpha_{k+1}}
    \end{align*}
  \end{footnotesize}
  The expressions for the transitions can also be derived inductively from the definitions. 
\end{proof}

Let us now fix an actor-based tree transducer with the notations of \Cref{def:actor-transducer}.
By structural induction on $t\in\Tree\Sigma$, and using the above claim, we can deduce expressions (up to product reindexing) of the states and transitions of $\alpha_t$ involving $\Dom(t)$, using the fact that
\[ \Dom(\sigma(t_1,\dots,t_{\rank(\sigma)})) \cong \{\varepsilon\} + \Dom(t_1) + \dots + \Dom(t_{\rank(\sigma)})\]
Finally using the definition of application we can describe the states and transitions of $\beta_{\mathrm{out}}(\alpha_t)$. The states are given in the statement of Claim~\ref{clm:big-actor}. The initial state is $q^\ominus_0 = ((q^\ominus_{0,\lab_t(u)})_{u\in\Dom(t)}, q^\ominus_{0,\mathrm{out}})$, and we have $\delta^\ominus((\vec{q}, q^\ominus_{\mathrm{out}}),\bullet) = (\vec{q}, \delta^\ominus_{\mathrm{out}}(q^\ominus_{\mathrm{out}},\bullet))$.

For an active state of the form $x = {(u,q^\oplus,(q^\ominus_v)_{v\neq u},q^\ominus_{\mathrm{out}})}$, we have $\delta^\oplus(x) = f_x(\delta^\oplus_{\lab_t(u)}(q^\oplus))$ where
  \begin{align*}
    f_x \colon \gamma(p^\oplus_1,\dots,p^\oplus_{\rank(c)}) &\mapsto \gamma((u,p^\oplus_1,(q^\ominus_v)_{v\neq u},q^\ominus_{\mathrm{out}}),\dots)
        \\
    p^\oplus &\mapsto (u,p^\oplus,(q^\ominus_v)_{v\neq u},q^\ominus_{\mathrm{out}})\\
    (q^\ominus_u,\mathtt{arg}_i\mathfrak{m}) &\mapsto (ui,\delta^\ominus_{\lab_t(ui)}(q^\ominus_{ui},(\ttR, \dots (\ttR, \mathfrak{m}))),(q^\ominus_v)_{v\neq ui},q^\ominus_{\mathrm{out}}) \\
    (q^\ominus_u,\underbrace{(\ttR, \dots (\ttR, \mathfrak{m})\dots)}_{\mathclap{\rank(\lab_t(u))\ \text{times}\ \ttR\qquad}}) &\mapsto \begin{cases}
      (w,\delta^\ominus_{\lab_t(w)}(q^\ominus_{w},\mathtt{arg}_i\mathfrak{m}),(q^\ominus_v)_{v\neq w},q^\ominus_{\mathrm{out}}) &\text{if}\ u = wi\ \text{for some}\ i\\
      ((q^\ominus_v)_{v\in\Dom(t)},\delta^\ominus_{\mathrm{out}}(q^\ominus_{\mathrm{out}},(\ttL,\mathfrak{m}))) &\text{otherwise, i.e.}\ u = \varepsilon
    \end{cases}
  \end{align*}
Finally, $\delta^\oplus(((q^\ominus_u)_{u\in\Dom(t)},q^\oplus))$ is given by the case $(q^\ominus_\alpha,q_\beta^\oplus)$ of \Cref{def:actor-app}.

\subsection{Proof of \Cref{lem:actor-to-uthm}}

We design a machine $\cH$ whose runs on the input $t$ simulate $\beta_{\mathrm{out}}(\alpha_t)$ step by step. Its states and memory values are:
\[ Q_H = Q^\oplus_{\mathrm{out}} \cup \bigcup_{\sigma \in \Sigma} Q^\oplus_{\sigma} \times Q^\ominus_{\mathrm{out}} \qquad M_H = \bigcup_{\sigma \in \Sigma} Q^\ominus_{\sigma} \]
A configuration $(u,q,\mu)$ of $\cH$ represents an active state $q^\oplus$ of $\beta_{\mathrm{out}}(\alpha_t)$ when:
\begin{itemize}
  \item either $q \in Q^\oplus_{\mathrm{out}}$ and $q^\oplus = ((\mu(v))_{v\in\Dom(t)}, q)$,
  \item or $q = (q^\oplus_u,q^\ominus_{\mathrm{out}})$ and  $q^\oplus = (u,q^\oplus_u,(\mu(v))_{v\neq u},q^\ominus_{\mathrm{out}})$.
\end{itemize}
The initial state of $\cH$ is $\delta^\ominus_{\mathrm{out}}(q^\ominus_{0,\mathrm{out}},\bullet)$. Using bottom-up initialization, for each node with label $\sigma$, we take the initial memory value to be $q^\ominus_{0,\sigma}$. This way, the initial configuration of $\cH$ is the unique representation of $\delta^\ominus(q^\ominus_0,\bullet)$ where $q^\ominus_0$ is the initial state of $\beta_{\mathrm{out}}(\alpha_t)$. (However, the representations of active states in $Q^\oplus \setminus (Q^\ominus_{\Dom(t)} \times Q^\oplus_{\mathrm{out}})$ are not unique: the memory value at the current position is undetermined.)

The argument of the transition function is some regular lookaround state, containing enough information to determine:
\begin{itemize}
  \item the label $\sigma$ of the current node;
  \item the current state --- let us treat here the case where it has the form $(q^\oplus,q^\ominus_{\mathrm{out}})$;
  \item the labels $\sigma_1,\dots,\sigma_{\rank(\sigma)}$ and memory values $q^\ominus_1,\dots$ of the current node's children;
  \item either the label $\sigma_{\mathrm{up}}$ and the memory value $q^\ominus_{\mathrm{up}}$ of the parent or the knowledge that the current node has no parent i.e.\ is the root.
\end{itemize}
On such an input, the transition function returns the result of applying to $\delta^\oplus_\sigma(q^\oplus)$ the map
\begin{align*}
  p^\oplus &\mapsto ((p^\oplus,q^\ominus_{\mathrm{out}}),\, [\text{some arbitrary value}],\, \circlearrowleft) \quad(\text{note the stay-instruction})\\
  \gamma(p^\oplus_1,\dots,p^\oplus_{\rank(c)}) &\mapsto \gamma(((p_1^\oplus,q^\ominus_{\mathrm{out}}),\, [\text{some arbitrary value}],\, \circlearrowleft),\dots)\\
  (p^\ominus,\mathtt{arg}_i\mathfrak{m}) &\mapsto ((\delta^\ominus_{\sigma_i}(q^\ominus_{i},(\ttR, \dots (\ttR, \mathfrak{m}))),q^\ominus_{\mathrm{out}}),\, p^\ominus,\, i) \\
  (p^\ominus,(\ttR, \dots (\ttR, \mathfrak{m}))) &\mapsto \begin{cases}
    (\delta^\ominus_{\mathrm{out}}(q^\ominus_{\mathrm{out}},(\ttL,\mathfrak{m})),\, p^\ominus,\, \circlearrowleft) &\text{if the current node is the root}\\
    ((\delta^\ominus_{\sigma_{\mathrm{up}}}(q^\ominus_{\mathrm{up}},\mathtt{arg}_j\mathfrak{m}),q^\ominus_{\mathrm{out}}),\, p^\ominus,\, \uparrow) &\text{if the current node is an \(j\)-th child}
        \end{cases}
\end{align*}
(for a non-root node, this number $j$ can also be determined by regular lookaround).
This definition closely reflects the transitions of $\beta_{\mathrm{out}}(\alpha_t)$ described in the previous subsection of the appendix, in order to get a step-by-step simulation.

\section{Details and proofs for Section~\ref{sec:lambda}}

\subsection{Specification of our affine $\lambda$-calculus}

Grammar of $\lambda$-terms:
$t,u \mathrel{::=} x \mid a \mid \lambda x.\, t \mid t\,u \mid \langle t,u\rangle \mid \pi_1\, t \mid \pi_2\, t$.
Reduction rules: $(\lambda x.\, T)\, U \to_\beta T[x \leftarrow U]$ and $\pi_i\, \pp{T_1,T_2} \to_\beta T_i$.
Typing:
\[ \frac{}{\Phi \mid \Theta, x : A \vdash x : A} \qquad
  \frac{\Phi \mid \Theta, x : A \vdash T : B}{\Phi \mid \Theta \vdash \lambda x.\, T : A \multimap B} \qquad
  \frac{\Phi \mid \Theta \vdash T : A \multimap B \quad \Phi \mid \Theta' \vdash U : A}{\Phi \mid \Theta, \Theta' \vdash T\, U : B}
\]
\[ \frac{}{\Phi,\; c : A \mid \Theta \vdash c : A} \qquad
  \frac{\Phi \mid \Theta \vdash T : A \quad \Phi \mid \Theta \vdash U : B}{\Phi \mid \Theta \vdash \pp{T,U} : A \with B} \qquad
  \frac{\Phi \mid \Theta \vdash T : A_1 \with A_2}{\Phi \mid \Theta \vdash \pi_i\, T : A_i}(i \in \{1,2\})
\]
 
\subsection{Proof of \Cref{lem:hennie-lambda}}

In order to show that $\beta$-reduction, starting from a $\lambda$-term built from an input $t\in\Tree\Sigma$ by our $\lambda$-transducer, simulates the tree-generating machine for $\cH$ on the same input $t$, we explain how relate configurations of $\cH$ to $\lambda$-terms.

First we note that $(u,q,\mu) \in \Conf(\cH,t)$ contains the same information as the data of:
\begin{itemize}
  \item the state $q$;
  \item the subtree $\subtree{t}{u}$ annotated with the memory values from $\mu$ --- this gives us a tree $\overline{t}^{u,\mu} \in \Tree{\qq{\Sigma,M}}$;
  \item the context $t[u \leftarrow \square]$ above $u$, plus memory annotations, yielding a one-hole context $\underline{t}_{u,\mu} \Tree{\qq{\Sigma,M}}[\{\square\}]$ with a single leaf with label $\square$.
        \tito{les chevrons servent pas à la même chose c'est confusing (par ex double chevron comme dans l'article des allemands)}
\end{itemize}
We explain how to represent (memory-annotated) subtrees $\overline{t}^{u,\mu}$, then contexts $\underline{t}_{u,\mu}$, then configurations.
\begin{definition}[set of $\lambda$-representations of a subtree]
  $\bet{-}_n \colon \Tree{\qq{\Sigma,M}} \to \{\text{sets of terms of type}\ A_n\}$ is defined by:
  \[\bet{\qq{\sigma,m}(t_1,\dots,t_{\rank(\sigma)})}_k = 
    \begin{cases}
      \varnothing &\text{if}\ k < \omega(m)\\
      \left\{ T_{\sigma,m}^{\vec{n};k}\; T_1 \; \dots \; T_{\rank(\sigma)} \mid \vec{n} \in \N^{\rank(\sigma)},\; T_i \in \bet{t_i}_{n_i} \right\} &\text{otherwise}
    \end{cases}
  \]
\end{definition}
For contexts note that the continuation passed to a subtree behavior can also be thought of as the behavior of the surrounding context. Let $B_n = A_{n-1} \multimap o^{\with Q}$ for $n\geqslant1$, so that $A_{n} = B_n \multimap o^{\with Q}$. We define $\bec{C}_n$ as a set of terms of type $B_n$ for a one-hole context $C$. 
\begin{definition}[set of $\lambda$-representations of a context]
  $\bec{-}_k$ maps one-hole contexts to sets of terms of type $B_{k}$ for $k\geqslant1$. It is defined by induction on the depth of the hole $\square$: the base case is $\bec{\square}_k = \{ \lambda x.\; \pp{\gamma_0 \mid q \in Q} \}$ and the inductive case is
  \begin{align*}
    \bec{C[\square \leftarrow \qq{\sigma,m}(t_1,\dots,t_i,\square,t_{i+1},\dots)]} &= \big\{ \lambda x.\; T_{\sigma,m}^{\vec{n};\ell}\; T_1\; \dots \; x \; \dots \; T_{\rank(\sigma)}\; U \; \big| \\
                                                                                     &\qquad \vec{n} \in \N^{\rank(\sigma)},\; n_i = k-1,\; T_j \in \bet{t_j}_{n_j},\; \ell \geq \omega(m),\; U \in \bec{C}_\ell \big\}
  \end{align*}
\end{definition}
\begin{definition}[set of $\lambda$-representations of a configuration] $\bek{-}$ maps configurations to sets of terms of type $o$:
  \[\bek{(u,q,\mu)} = \{ \pi_q\,(T\,U) \mid \exists k \geqslant 1 : T \in \bet{\overline{t}^{u,\mu}}_k,\; U \in \bec{\underline{t}_{u,\mu}}_k \}\]
  We define $(\simulambda) \subset \{\lambda\text{-terms}\ T\ \text{s.t.}\ \Gamma^\with \mid \varnothing \vdash T : o\} \times \Tree{\Gamma}[\Conf(\cH,t)]$ as the relation generated from $T \simulambda (u,q,\mu) \iff T \in \bek{(u,q,\mu)}$ by the context closure $(\forall i,\; T_i \simulambda g_i) \implies \gamma\,\pp{T_1,\dots,T_k} \simulambda \gamma(g_1,\dots,g_k)$ ($g$ stands for \enquote{global state}).
\end{definition}

\begin{lemma}[Simulation]
  For all $(u,q,\mu) \in \Conf(\cH,t)$ and $V \in \bek{(u,q,\mu)}$, there is $V \to_\beta^* V'$ s.t.\ $V' \simulambda \text{computation-step}((u,q,\mu))$. 
\end{lemma}
\begin{proof}
  Let $\sigma = \lab_t(u)$ and $m = \mu(u)$. By definition of $\bek{-}$, for some $T_1\in \bet{\overline{t}^{u1,\mu}}_{n_1},\dots,T_{\rank{\sigma}}\in \bet{\overline{t}^{u\rank(\sigma),\mu}}_{n_{\rank(\sigma)}}$ and $U \in \bec{\underline{t}_{u,\mu}}_k$,
  \[ T' = \pi_q\; (T_{\sigma,m}^{\vec{n};k}\; T_1 \; \dots \; T_{\rank(\sigma)}\; U) \to_\beta^* (\text{encoding of}\ \delta(q,\sigma,m))\underbrace{[\text{the substitution of \S\ref{sec:hennie-lambda}}]}_{\mathclap{\text{call this}\ \mathfrak{s}_{\text{move}}}}\underbrace{[\vec{z} \leftarrow \vec{T},\; x \leftarrow U]}_{\mathclap{\text{call this}\ \mathfrak{s}_\beta}} \]
  Meanwhile, the inductive definition of $\ogreaterthan$ on $\Tree{\Gamma}[\Conf(\cH,t)]$ entails that the image of $(u,q,\mu)$ by the computation-step function can be similarly described as $\delta(q,\sigma,m)[(q',m',d) \leftarrow (u,q,\mu) \ogreaterthan (q',m',d) \; \forall q',m',d]$. Therefore, it suffices to check that:
  \[ \forall q',m',d,\; \exists V'' :\quad (q',m',d)\mathfrak{s}_{\text{move}}\mathfrak{s}_\beta \to_\beta^* V'' \simulambda \bigl( (u,q,\mu) \ogreaterthan (q',m',d) \bigr) = (ud, q', \mu')\]
  where $\mu'(u) = m'$ and $\mu'(v) = \mu(v)$ for $v\neq u$, and $(q',m',d')$ ranges over the labels of leaves of $\delta(q,\sigma,m)$ of this form. If there is no such leaf, then the result is vacuously true. If there is at least one, then we have: $k \geqslant \omega(m) > \omega(m') \geqslant 0$. 
  Note that the definition of $\mathfrak{s}_{\text{move}}$ depends on whether $k=0$ or $k\geqslant1$; thus, the weight-reducing property guarantees that we are in the latter case.

  We then proceed by case disjunction on $d$.
  \begin{itemize}
    \item First we consider the case $d \in \db{\rank{\sigma}}$, i.e.\ $d \neq \uparrow$.
          \begin{align*}
            (q',m',d)\mathfrak{s}_{\text{move}}\mathfrak{s}_{\beta}
            &= \pi_{q'}\, \bigl(z_d\, \bigl(\lambda z_{\mathrm{new}}.\, T_{\sigma,m'}^{\vec{n}';k}\, z_1\, \dots\, z_{d-1}\, z_{\mathrm{new}}\, z_{d+1}\, \dots\, z_{\rank(\sigma)}\, x\bigr)\bigr)\mathfrak{s}_\beta  & n'_d &= n_d - 1\\
            &= \pi_{q'}\, \bigl(T_d\, \bigl(\underbrace{\lambda z_{\mathrm{new}}.\, T_{\sigma,m'}^{\vec{n}';k}\, T_1\, \dots\, T_{d-1}\, z_{\mathrm{new}}\, T_{d+1}\, \dots\, T_{\rank(\sigma)}\, U}_{\text{call this}\ U'}\bigr)\bigr) & n'_i &= n_i\ \text{for}\ i \neq d
          \end{align*} 
          One can recognize the shape of a $\lambda$-representation of a one-hole context: $U' \in \bec{t'}_{n_d}$ where
          \begin{align*}
            t' &= \underline{t}_{u,\mu}[\square \leftarrow \qq{\sigma,m'}(\overline{t}^{u1,\mu},\dots,\overline{t}^{u\cdot(d-1),\mu},\square,\overline{t}^{u\cdot(d+1),\mu},\dots)] \\
               &= \underline{t}_{u,\mu'}[\square \leftarrow \qq{\sigma,\mu'(u)}(\overline{t}^{u1,\mu'},\dots,\overline{t}^{u\cdot(d-1),\mu'},\square,\overline{t}^{u\cdot(d+1),\mu'},\dots)]
               = \underline{t}_{ud,\mu'}
          \end{align*}
          Note also that $T_d \in \bet{\overline{t}_{ud,\mu}}_{n_d} = \bet{\overline{t}_{ud,\mu'}}_{n_d}$. Therefore
          $(q',m',d)\mathfrak{s}_{\text{move}}\mathfrak{s}_{\beta} = \pi_{q'}\, \bigl(T_d\, U'\bigr) \in \bek{(ud,q',\mu')}$.
    \item In the remaining case $d={\uparrow}$, we can assume that $u$ is not the root, since we are translating a $\THM$ that realizes a total function. Therefore, writing $\sigma' = \lab_t(\parent{u})$: 
          \begin{footnotesize}
            \[ U \in \bec{\underline{t}_{u,\mu}}_k = \bec{\underline{t}_{\parent{u},\mu}[\square \leftarrow \qq{\sigma',\mu(\parent{u})}(\overline{t}^{\parent{u}1,\mu},\dots,\square,\dots,\overline{t}^{\parent{u}\rank(\sigma'),\mu})]}_k = \left\{ T_{\sigma',\mu(\parent{u})}^{\vec{n}';\ell}\; T'_1\; \dots \; x \; \dots \; T'_{\rank(\sigma)}\; U'  \mid \dots\right\} \]
            \[ (q',m',\uparrow)\mathfrak{s}_{\text{move}}\mathfrak{s}_{\beta} = \pi_{q'}\, \bigl(x\, \bigl(T_{\sigma,m'}^{\vec{n};k-1}\, \vec{z}\bigr)\bigr)\mathfrak{s}_{\beta} = \pi_{q'}\, \bigl(U\, \bigl(T_{\sigma,m'}^{\vec{n};k-1}\, \vec{T}\bigr)\bigr) \to_\beta \underbrace{\pi_{q'}\, \bigl(T_{\sigma',\mu(\parent{u})}^{\vec{n}';\ell}\; T'_1\; \dots \; \bigl(T_{\sigma,m'}^{\vec{n};k-1}\, \vec{T}\bigr) \; \dots \; T'_{\rank(\sigma')}\; U' \bigr)}_{\text{call this}\ V''} \]
          \end{footnotesize}
          First, from the shape we can see that $T_{\sigma,m'}^{\vec{n};k-1}\, \vec{T} \in \bet{\overline{t}^{u,\mu'}}_{k-1}$ modulo one sanity check: since $\cH$ is weight-reducing,
          \[ \omega(m') < \omega(m) \leqslant k \qquad \text{therefore} \qquad \omega(m') \leqslant k-1 \]
          Also, if $u$ is an $i$-th child, then $T'_j \in \bet{\underline{t}_{\parent{u}j,\mu}}_{n'_j}$ for $j\neq i$ and $U' \in \bec{\underline{t}_{\parent{u},\mu}}_\ell$ and $n'_i = k-1$ and $\ell \geqslant \mu(\parent{u}) = \mu'(\parent{u})$ by definition of $\bec{-}$. Therefore, $V'' \in \bek{(u,q',\mu')}$. \qedhere
  \end{itemize}
\end{proof}

\begin{remark}
  We only use the fact that the memory writes that occur during upwards exits are weight-reducing. Thus $\omega$ could be replaced by a function on memory symbols that bounds the upwards exists instead of the number of visits. This is consistent with the intuitions for \enquote{$n$-truncated behaviors} given in the main text.
\end{remark}

To conclude the proof of \Cref{lem:hennie-lambda}, we observe that the $\lambda$-transducer that we build maps $t\in\Tree\Sigma$ to a $\lambda$-term
\[ (\lambda z.\, \pi_{q_{init}}\, (z\, (\lambda x.\, \pp{\gamma_0 \mid q \in Q}))) \underbrace{T_t}_{\mathclap{\text{in \(\bet{\overline{t}^{\varepsilon,\mu_{init}}}_N\) by structural induction}}} \to_\beta \pi_{q_{init}}\, (T_t\, (\underbrace{\lambda x.\, \pp{\gamma_0 \mid q \in Q}}_{\mathclap{\text{in \(\bec{\square}_N\) by definition}}})) \in \bek{(\varepsilon,q_{init},\mu_{init})} = \bek{C_{init}(t)} \]

\section{Proofs of Section~\ref{sec:postcomp}}

\subsection{Proof of Claim~\ref{clm:StartingConfig} and \Cref{cor:rel-postcomp}}

  Let $\cH$ be a $\THM$ with $\sem{\cH} \colon \Tree\Gamma \to \Tree\Sigma$ and $\cR$ a bottom-up relabeler with $\sem\cR \colon \Tree\Sigma \to \Tree\Xi$.
	
  We first give a justification of Claim~\ref{clm:StartingConfig}: the new machine for $g_\cH$ uses bottom-up initialization to copy in memory the values provided by the input, then performs a tree traversal to position its head at the position indicated by the configuration, and finally, it simulates $\cH$.

  For any position $u \in \db{\maxrank{\Xi}}^{*}$ and label $\xi \in \Xi$, the language $L_{u,\xi} = \{t \in \Tree\Xi \mid u \in \Dom(t) \land \lab_{t}(u) = \xi\}$ is regular. A classical property of bottom-up relabelings is that they preserve regular tree languages by inverse image; and so does $g_\cH$. Therefore, $g_\cH^{-1}(\sem\cR^{-1}(L_{u,\xi}))$ is regular.

  Let $h$ be the maximum height of a tree in the range of the transition function of $\cH$. There are finitely many languages $g_\cH^{-1}(\sem\cR^{-1}(L_{u,\xi}))$ for \emph{$u$ of length at most $h$} and $\xi \in \Xi$. Thus, without loss of generality, we may assume that these languages are all recognized by the \emph{same} bottom-up automaton $(\delta,Q)$, with only the set of accepting states changing. This means for every $u \in \db{\maxrank{\Xi}}^{\leqslant h}$ there exists a partial function $\psi_{u}\colon Q \rightharpoonup \Xi$ such that $\psi_{u}(\delta[t]) = \lab_{\sem\cR(g_\cH(t))}(u)$ for any tree $t$. Note that if $t$ encodes a configuration of $\cH$ reachable on the input tree $s$, then $g_\cH(t)$ appears as a subtree of $\sem\cH(s)$ at some position $w$. Then the bottom-up nature of $\cR$ entails that $\lab_{\sem\cR(\sem\cH(s))}(wu) = \lab_{\sem\cR(g_\cH(t))}(u)$.

  \tito{new def of regular lookaround}
  We can now sketch the construction of a new $\THM$ \emph{with regular lookaround} (cf.\ \Cref{thm:lookaround})  that computes $\sem\cR \circ \sem{\cH}$. It simulates $\cH$ step-by-step --- more than that: its states and memory values (and therefore its configurations) are exactly the same --- while applying the relabeling on the fly. The lookaround automaton is $(\delta,Q)$; we assume, again w.l.o.g., that if $t$ encodes a configuration of $\cH$, then the lookaround state $\delta[t]$ suffices to determine the current state of this configuration, as well as the label and memory value of its current node. The new transition function maps a lookaround state to a relabeling of the image of the corresponding (state,label,memory) tuple by the transition function of $\cH$. To determine the new label at position $u$, we just apply $\psi_{u}$ to the lookaround state --- this is justified by the reasoning in the previous paragraph.

\subsection{Proof of \Cref{thm:narrow-postcomp}}

Given $\cH$ and $\cN$ two Hennie machines, we aim to construct a Hennie machine $ \cG$ that realizes the composition $ \cN \circ \cH$. Over an input tree $t$, $\cG$ simulates the execution of $\cN$ on the intermediate tree $\sem{\cH}(t)$ by simulating $\cH$ to build it implicitly in memory. During the execution of $\cG$, each node of $t$ holds in memory all the right hand sides of rules of $ \cH $ that it has produced (we call those \emph{(memory) fragments}). We compute that lazily, by moving forward the simulation of $\cH$ only when $\cN$ asks for a yet unknown part of $\sem{\cH}(t)$.
In addition, each node of $\sem{\cH}(t)$ in memory of $\cG$ is decorated with the memory of $\cN$ used in our simulated run of $\cN$. We remember where the reading head of $\cN$ is in the intermediate tree $\sem{\cH}(t)$ by adding pointers to our memory symbols. 

Finally, to allow movement of the simulated reading head of $\cN$ between different $\sem{\cH}(t)$ fragments, the root and some leaves of the fragments are enriched with $\Link$ symbols, which take the place of those $\sem{\cH}(t)$ fragments that are stored at a different input node in $t$. They contain the information of what is the node itself, and a precise tree in the memory forest of said node, indicated by a natural number.

Note that there are only finitely many transitions rules in $\cH$. Additionally, as long as forests are of bounded size (in the number of trees they contain), then there also are only finitely many link symbols to decorate trees with. Thus, if our forests have bounded size, then we need only finitely many memory symbols. 

Lets us describe formally the tools we have introduced here

\begin{definition}[Ordered pointed forest of pointed trees]
A pointed tree on $\Sigma$ is a tree in $\Sigma$ together with a node $\bullet$ in its domain. An ordered pointed forest $m$ of pointed trees of size $k$, on alphabet $\Sigma$, is a finite sequence $t_1, \ldots, t_k$ of pointed trees on $\Sigma$ together with a natural number $1 \leq \diamond \leq k$. For $ 1 \leq i \leq k$, we use $\bullet_i$ to denote the distinguished node in tree $t_i$.
\end{definition}

The following notations are useful when discussing pointed forests : if $t$ is a pointed tree on alphabet $\Sigma$ with memory in $N$, and $s \in \Dom(t)$, then $t[\bullet = s]$ is the tree obtained by setting the pointed node to $s$. Letting $\sigma \in \Sigma$, we also define $t[\bullet \leftarrow \sigma]$ the tree obtained from $t$ by setting the label of $\bullet$ to $\sigma$. In a similar fashion, if $m$ is an ordered pointed forest of pointed trees on $\Sigma$ with $k$ trees, and $ 1 \leq i \leq k$, then $m[\diamond = i]$ changes the pointed tree of $m$, and $m[i \leftarrow t]$ sets the $i$-th tree to be $t$. The size $|m|$ of $m$ is the number of trees it contains. 

To have an easier description of the transition function of $\cN \circ \cH$, we allow it to take stationary transitions, that do not move the reading head, as discussed in \Cref{sec:basic-extensions}. 

\begin{theorem}
  Let  $ \cH = (Q_\cH, M_\cH, \Sigma, \Gamma, q^0_\cH, \delta_\cH) $ be a $\twhm$ which visits input nodes at most $C_{fv}$ times, and $ \cN = (Q_\cN, M_\cN, \Gamma, \Xi, q^0_\cN, \delta_\cN) $ a narrow-visit Hennie machine, visiting on each branch outputing run on a tree of $\Dom(\sem{\cN}) \cap \Img(\sem{\cH}) $ the union of at most $C_{fnv} $ branches.
  Then there exists a $\twhm $ $ \cG $ such that $ \sem{\cG} = \sem{\cN} \circ \sem{\cH} $.
\end{theorem}

\begin{definition}[ $ \cN \circ \cH $ ]
Letting the maximum forest size $C_{fs} = C_{fv} \times C_{fnv}$, we define $ \cG $ as follows :
Its set of states $Q_\cG$ is separated into three subsets we call "regimes".
\[ Q_\cG = \underbrace{Q_\cN}_{\cN\text{-simulating regime}} \cup \underbrace{(Q_\cN\times [C_{fs}])}_{\cG\text{-machinery regime}} \cup \underbrace{(Q_\cN \times Q_\cH \times [C_{fs}] \times [C_{fs}])}_{\cH \text{-simulating regime}} \cup \{\qzero\}\]
Its set of memory symbols is $M_\cG$, a finite subset of the set of pointed ordered forests of pointed trees over the alphabet $ \pp{\Gamma, M_\cN} \cup \Link $, with size $ \leq C_{fs} $. To obtain a finite subset, we restrict ourselves to trees which have the same shape as the right-hand side of some rule of $\cH$.
Symbols in $\Link$ are special symbols used to hold both : the necessary information about both the inter-fragments structure of the intermediate tree $\sem{\cH}(t)$, and the memory of our simulation of $\cH$. We explain the intuition behind their definition later while discussing the transitions $\delta_\cG$ of $ \cG $, and here, we only define them as
\[\Link = \underbrace{(Q_\cH \times D_{R_\Sigma} \times [C_{fs}] )}_{\downLink\text{, of arity } 0}
\cup \underbrace{(Q_\cH \times D_{R_\Sigma} \times [C_{fs}])}_{\upLink\text{, of arity } 1} 
\cup \underbrace{(Q_\cH \times D_{R_\Sigma} \times Q_\cH)}_{\downNewLink\text{, of arity } 0} \]
The input and output alphabets are $ \Sigma $ and $ \Xi $.
The initial state is $ \qzero $, and the initial memory the empty forest.
\end{definition}

Finally, we describe the transitions of $ \cG $, grouped by how they move between different regimes of the states of $ \cG $. The following figure is a representation of these 5 types of transitions and their relation to regimes. 

% https://q.uiver.app/#q=WzAsMyxbMCwxLCJHXFx0ZXh0ey1zaW11bG\cG0aW5nIHJlZ2ltZX0iXSxbMCwwLCJGXFx0ZXh0ey1tYWNoaW5lcnkgcmVnaW1lfSJdLFswLDIsIkhcXHRleHR7LXNpbXVsYXRpbmcgcmVnaW1lfSJdLFswLDEsIiIsMCx7Im9mZnNldCI6LTUsIm\cN1cnZlIjotNX1dLFsyLDAsIiIsMix7Im9mZnNldCI6NSwiY3VydmUiOjV9XSxbMCwyLCIiLDIseyJvZmZzZXQiOjUsIm\cN1cnZlIjo1fV0sWzAsMCwiIiwwLHsicmFkaXVzIjoxfV0sWzEsMCwiIiwwLHsib2Zmc2V0IjotNSwiY3VydmUiOi01fV1d
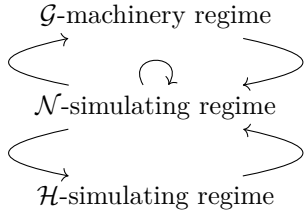
\begin{center}
\[\begin{tikzcd}
    {\cG\text{-machinery regime}} \\
    {\cN\text{-simulating regime}} \\
    {\cH\text{-simulating regime}}
    \arrow[shift left=5, curve={height=-30pt}, from=1-1, to=2-1]
    \arrow[shift left=5, curve={height=-30pt}, from=2-1, to=1-1]
    \arrow[from=2-1, to=2-1, loop, in=60, out=120, distance=5mm]
    \arrow[shift right=5, curve={height=30pt}, from=2-1, to=3-1]
    \arrow[shift right=5, curve={height=30pt}, from=3-1, to=2-1]
\end{tikzcd}\]
\end{center}

The "transitions of regime X" are those transitions which go in or out of regime X.

A first transition is used to initialize the memory content at the root of the input tree : we let $\delta_\cG(\qzero, \pp{\sigma, \top_\cG}) = \qq{q_0, t^*, \circlearrowleft } $ where $t^*$ is built by taking $ \delta_\cH(q^0_\cH, \pp{\sigma, \top_\cH}) $, and replacing each occurrence of $ \qq{q_\cH, m_\cH, d}$ with $ (q_\cH, m_\cH, d)$.

Now let $\sigma \in \Sigma$, $q_\cG \in Q_\cG $, and $m_\cG \in M_\cG$.

The driving transition type is the type of $\cN\shortrightarrow \cN$ transitions, named after the second machine in our composition. They are driving in the sense that the other types of transitions are taken lazily, not any earlier than when $\cN \shortrightarrow \cN$ transitions come to need them. In the $\cN$-simulating regime, where $\cN\shortrightarrow \cN$ transitions happen, there is a fragment of the intermediate tree $\sem{\cH}(t)$ loaded in memory over the node which we read, and we use it to simulate executing our second machine $\cN$, updating the decorating memory as it instructs us. Here, the forest-level pointer $\diamond$ points to the particular fragment in our forest that our simulation of $\cN$ is currently visiting, while the tree-level pointer $\bullet_{\diamond}$ is used to represent the precise position of the reading head of $\cN$ in this fragment.

\begin{definition}[$\cN \shortrightarrow \cN$ transitions]
If $q_\cG \in Q_\cN$ and $\lab_{m_\cG}(\bullet_\diamond) = \pp{\gamma, m} \in \pp{\Gamma, M}$, then we define the right-hand side as 
\[ \delta(q_\cG, \pp{\gamma, m})[\qq{q', m', d'} \leftarrow \qq{q', m_\cG[\bullet_\diamond \leftarrow \pp{\gamma, m'}][\bullet_\diamond = \bullet_\diamond d'], \circlearrowleft}]
\]
\end{definition}

The second and third types of transitions we define are the $\cN\shortrightarrow \cG$ and $\cG \shortrightarrow \cN$ transitions. They are named after the composed machine itself, because they do not correspond to any real step in an execution of either $\cN$ or $\cH$, they are only related to the internal workings of $\cG$.
They are invoked when our simulated $\cN$ asks for a node of the intermediate tree that is not in the current fragment. When this happens, our construction guarantees they instead find an $\upLink$ or $\downLink$ symbol, say $(n,d,i)$. For these transitions, we'll only need the $(d,i)$ information. This tuple represents both : a direction in the input tree, and the position of a tree in the forest stored over in this direction. This tells our machine $\cG$ where in memory is stored the fragment containing the node we were looking for. By taking a $\cN\shortrightarrow\cG$ transition, the machine $\cG$ commits this information to its states, and moves on the input tree to try and find the missing fragment.

\begin{definition}[$\cN\shortrightarrow \cG$ transitions]
If $q_\cG = q_\cN \in Q_\cN$ and $\lab_{m_\cG}(\bullet_\diamond) = (n, d, i) $, with $i \geq 1$, then we define the right-hand side as the root only tree $ \qq{(q_\cN, i), m_\cG, d} $. 
\end{definition}

Because of the particular shape of the resulting state of $\cG$, these transitions are always followed by the other type $\cG\shortrightarrow\cN$. These transitions consume the information $i$ that was stored in state to place the pointers $\bullet$ and $\diamond$ at the right place, preparing to resume the simulation of the second machine $\cN$.

\begin{definition}[$\cG\shortrightarrow \cN$ transitions]
Arriving in state $(q_\cN,i)$ and reading a new memory symbol $m_\cG'$, we take a transition of right-hand side $\qq{q_\cN, m_\cG'[\diamond = i][\bullet_i = \bullet_i d_!], \circlearrowleft}$. Direction $d_!$ is the unique direction that points to a neighbor of $\bullet_i$. It is unique because we only leave a fragment when pointing to a symbol in $ \Link $. The definition of $\cH\shortrightarrow \cN$ transitions guarantees that all symbols in $\Link$ are either, by construction of the memory, a root of arity 1, or a leaf of arity 0, so they have a unique neighbor.
\end{definition}

Finally, we describe $\cN\shortrightarrow \cH$ and $\cH\shortrightarrow \cN$ transitions. They can be seen as a particular case of the previous type : since the memory is initially empty, we'll, at the first visit, actually have to compute the fragments of the intermediate tree that we want to work on before they can be processed. This is done "on the fly", when our simulation of $\cN$ asks for them. That is to say, like with the previous $\cN\shortrightarrow\cG$ transitions, it is done only when we encounter a $\Link$ symbol. The difference this time is that these transitions are triggered only by reading $\downNewLink$ symbols. They are the links that represent yet uncomputed fragments of $\sem\cH(t)$. Because the fragment does not yet exist, we do not have any information about its position in the forest, but what we have instead is the necessary state $q_\cH$ that is used to compute it, explaining the $q_\cH$ and $d$ components of their $(q_\cH, m_\cH, d)$ shape. Once read and committed to state, this information is discarded by overwriting with a $\downLink$ symbol, pointing to a fragment that will be generated in the next transition.
To ensure that the future fragment will be able to link back to the current one through its root, we also commit to states the information $\diamond$ of the position of the current fragment, for use in a future $\upLink$ symbol.

\begin{definition}[$\cN\shortrightarrow\cH$ transitions]
If $q_\cG = q_\cN \in Q_\cN$ and $\lab_{m_\cG}(\bullet_\diamond) = (q_\cH, m_\cH, d) $, then we define the right-hand side as the root-only tree $ \qq{(q_\cN, q_\cH, i, \diamond), m_\cG[\bullet_\diamond \leftarrow (m_\cH, d, i)], d} $. $i$ is defined here as the one plus the length of the forest in memory at direction $d$, and is computed using regular lookaround.
\end{definition}

Like the $\cG$-machinery regime, the $\cH$-simulating regime is temporary, and we always go back to $\cN$-simulating with a $\cH\shortrightarrow \cN$ transition. Here we finally make use of the $m_\cH$ information stored in $\Link$ symbols : since we are guaranteed to leave an input node only with a $\cN\shortrightarrow\cG$ or $\cN\shortrightarrow\cH$ transition, we know that when we move to a new node in the input, the pointer $\bullet_\diamond$ in memory (if it exists, i.e. the forest is non-empty), is pointing to a $\Link$ symbol. By definition of $\cN\shortrightarrow\cH$ transitions, this $\Link$ symbol contains the last written memory of $\cH$ at this point in its current simulated branch-outputting run. We use that information together with the $g_\cH$ stored in our current $\cG$ state to compute the a new fragment to be added to memory. This new fragment corresponds to the right hand-side of a rule of $\cH$, but every reading head in the right-hand side of our rule is replaced with a $\downNewLink$ symbol, we add an $\upLink$ symbol at the root to link back to the fragment we moved down from, and finally we resume simulation of $\cN$ by setting the current visited fragment to be this newly created one.

\begin{definition}[$\cH\shortrightarrow \cN$ transitions]
    Arriving on node $sd$ from node $s$ in state $(q_\cN, q_\cH, i, j)$ and reading a new memory symbol $m_\cG$, we take a transition of right-hand side $\qq{q_\cH, m_\cG[i\leftarrow\delta_\cH(q_\cH, \pp{\sigma, m_\cH})^*][\diamond = i], \circlearrowleft}$, where $m_\cH$ is either $\top_\cH$ if $m_\cG$ is empty, or $\pi_1(\lab_{m_\cG}(\bullet_{\diamond}))$ if it is not. The operation $*$ replaces every symbol $\qq{q_\cH', m_\cH', d'}$ with $(q_\cH', m_\cH', d')$, decorates every node with label in $\Gamma$ with $\top_\cN$, adds a new root $(m_\cH, d^{-1}, j)$ (where $d^{-1}$ is the direction such that $sdd^{-1} = s$, and can be computed using lookaround), and finally sets $\bullet_i$ to the old root.
\end{definition}

This concludes the definition of our constructed machine $\cG$.

\begin{theorem}
$ \sem{\cG} = \sem{\cN}\circ\sem{\cH} $
\end{theorem}

The core of the proof relies on the notion of stitched tree. The stitched tree $\Stitch(c)$ is built from the memory $\mu$ in a configuration $c$ of $\cG$ on an input tree $t$, and is defined as the tree obtained by plugging together all the trees fragments in memory, along pairs of $\downLink$, $\upLink$ symbols produced during the same $\cG\shortrightarrow\cH, \cH\shortrightarrow\cG$ transition pair. Using the stitched tree, we decode from a configuration of $\cG$ both a configuration of $\cN$ on $\sem{\cH}(t)$, and the smallest global state in an execution of $\cH$ on $t$ that contains all the node visited by that configuration of $\cN$.

\begin{definition}[Matching nodes]
    If $ c = (v, q, \mu)$ is a configuration of $\cG$ over some input $t$, $u$
    and $u'$ two nodes in $t$, $i < |\mu(u)| $, and $j < |\mu(u')| $, we say
    that a node $s$ in $\mu(u)_i$ of label $(m_\cH, d, k) $ in $\downLink$ matches the root node in $\mu(u')_j$ of label in $\upLink$ if and only if $u'=ud$ and $j=k$. 
\end{definition}
\begin{definition}[Stitched tree]
    The stitched tree $\Stitch(c)$ is the tree defined by the following
    relational structure : 
    \begin{itemize}
        \item the universe is $ \uplus_{u \in \Dom(t), i <
    |\mu(u)|} \Dom(\mu(u)_i) $ the disjoint union of all nodes in all forests
    over all input nodes. 
\item except for the "first child" relation, the child
    relations are the union of the child relations of the family $(\mu(u)_i)_{u
    \in \Dom(t), i<|\mu(u)|}$. 
\item the "first child" relation is the union of the
    first child relations of the family $(\mu(u)_i)_{u
    \in \Dom(t), i<|\mu(u)|}$, union the "$s$ matches $s'$" relation. 
\end{itemize}
    For any node $s \in \Stitch(c)$, we call "the origin $\orgn(s)$ of $s$" the node $u$ in the input such that $s$ is a node in $\Dom(\mu(u)_i)$ for some $i$. For any $u\in \Dom(t)$, any $i < |\mu(u)|$, any node $s \in \Dom(\mu(u)_i)$, $\Stitch_c(u,i,s)$ is the corresponding node in $\Dom(\Stitch(c))$.
\end{definition}
It is undefined if $C$ has its state in the $\cH$-simulating regime.

    We verify that this relational structure indeed represents a tree.
    $\cG\shortrightarrow\cH$ rules are the only type of rules that add
    $\downLink$ symbols to memory. Since the $n, d,i$ are chosen such that there
    is no $i$-th tree in direction $d$, a node can only match another node that
    is created later, and thus there are no looping sequences of matching nodes.
    Additionally, since each root node is matched by the $\downLink$ that was
    written at the last step, the relational structure describes a connex graph.
    The relational structure defines a connex, directed graph without loop, and
    we can order the children of any node, thus $\Stitch(c)$ is indeed a tree.

    We will now decode, from a stitch tree $\Stitch(c)$ on input tree $t$, the global state of $\cH$ it represents. Intuitively, $\downNewLink$ symbols correspond one-to-one to ouput heads of $\cH$ (see how they are added to memory in transitions of type $\cH\shortrightarrow\cG$). We just need to build back the memories of those different output heads. If $u$ is a node in $\Dom(t)$, we define $\nu^c_s(u)$ as the memory symbol $m_\cH$ in the last $\upLink$ symbol that has origin $u$ along the branch leading to $s$.
    \begin{definition}
        $\DecodeH(c)$ is defined by the procedure that builds it from the
        stitched tree $\Stitch(c)$. Replace all symbols $(q_\cH, m_\cH, d) \in
        \downNewLink$ at some leaf $s$ with $\qq{q_\cH, \orgn(s)d, \nu^c_s}$,
        then remove all remaining internal $\downLink$ and $\upLink$ symbols by
        replacing them with edges from their unique parent to their unique child
        (remember both types contain only symbols of arity 1).
    \end{definition}

    Assuming for now that this is indeed a global state of $\cH$ on some input,
    this encoding describes a partial mapping $\Re$ from nodes in fragments in
    $\mu(\Dom(t))$ to nodes in $\sem{\cH}(t)$ : $\Re$ first send any node that
    is not a $\Link$ to its place in the stitched tree, then to the
    corresponding node in the encoded global state of $\cH$, and then to
    $\sem{\cH}(t)$ by inclusion. 
   
    We now describe how to decode the configuration $\DecodeN(C)$ of $\cN$ from
    a configuration $C$ of $\cG$. Letting $ C = (u, q, \mu) $, we define $\Point(C)$ as
the node $ \bullet_\diamond $ of $\mu(u)$. $\lab(\Point(C))$ is then a shorthand for
$\lab_{\mu(u)}(\bullet_\diamond)$.
    \begin{definition}
        $\DecodeN(C)$ is defined by the procedure that builds it from the
        stitched tree $\Stitch(C)$. Letting $C = (q_\cG, u, \mu)$, we take the
        state of $\cN$ that is already explicit in $q$, set the reading head to
        $\Re(\Point(C))$, and define the memory over a node $s \in
        \Dom(\sem{\cH}(t)) $ as the memory stored at its unique preimage by
        $\Re$, or $\top_\cN$ if it does not have one.
    \end{definition}
    It is undefined if $\Point(C)$ has a label in $\Link$.
If a configuration $C$ is such that both $\DecodeH(C)$ and $\DecodeN(C)$ are well-defined, we call $C$ decodable.

\begin{lemma}
    Given a branch-outputting run of $\cG$, the sequence of configurations of
    $\cH$ decoded from the stitched trees is a run of $\cH$.
\end{lemma}

\begin{proof}
Let $t \in T_\Sigma$ be some input tree, $C^\cG_0 \xrightarrow{b_0^\cN} C^\cN_1
\xrightarrow{b_1^\cN} \ldots \xrightarrow{b_{n-1}^\cG} C^\cG_n$ be a
branch-outputting run of $\cG$ on $t$, and for $i \leq n$, let $h_i$ be the
global state of $\cH$ encoded in $C_i$.
For any $i\leq n$, we let $C_i = (q_i, u_i, \mu_i)$.

\subparagraph{Irrelevant transitions}
We first verify that only transitions of the $\cH$-simulating regime have an impact on the sequence $(h_i)$. For any $i < n$ such that neither $q_i$ nor $q_{i+1}$ are in the $\cH$-simulating regime, we know that the transition that goes from $C_i$ to $C_{i+1}$ is not a transition of the $\cH$-simulating regime, so the changes in memory between $\mu_i$ and $\mu_{i+1}$ are only made on $\cN$ memory symbols or pointing relations. The encoded global states $h_i$ and $h_{i+1}$ don't depend on this information, and thus $h_i = h_{i+1}$.

\subparagraph{Simulation of a $\cH$ step}
Each step of $\cH$ is simulated by a pair of $\cN\shortrightarrow\cH,
\cH\shortrightarrow\cN$ transitions, hence we will compare pairs
of configurations $h_{k-1}$ and $h_{k+1}$.
We look at such pairs of $\cH$ global states, and verify that they are
transformed in a way that is coherent with the transition function of $\cH$.
Let $k < n$ such that $q_k = (q_\cN, q_\cH, i, j)$ is in the $\cH$-simulating regime. 

We will show that there exists a node $s$ in $h_{k-1}$ labelled with $(u, q,
\mu)$ such that $h_{k+1} = h_{k-1}[s \leftarrow (u, q, \mu) \ogreaterthan
\delta_\cH(q, \lab_t(u), \mu(u)] $, i.e. $h_{k+1}$ is obtained from $h_{k-1}$ by
applying $\delta_\cH$ at some node $s$.

We first look at the effect of the transitions on the stitched trees.
We know that $k>0$, because $q_0$ is not in the $\cH$-simulating regime, and
again, that the transitions from $C_{k-1}$ to $C_k$ and from $C_k$ to $C_{k+1}$
are respectively of type $\cN\shortrightarrow\cH$ and $\cH\shortrightarrow\cN$.
Let us look at the differences between $\mu_{k-1}$ and $\mu_{k+1}$. Because the
first transition is of type $\cN\shortrightarrow\cH$, we know that
$\lab_{\mu_{k-1}(u_{k-1})}(\bullet_\diamond) \in \downNewLink$, and let it equal
$(q_\cH, m_\cH, d)$. Then by definition of $\cN\shortrightarrow\cH$ transitions,
this first transition sets $\mu_k(u_{k-1}) =
\mu_{k-1}(u_{k-1})[\bullet_\diamond\leftarrow (m_\cH, d, i)]$, where $i$ is one
plus the length of the forest $\mu(u_{k})$, and $\mu_k = \mu_{k-1}$ everywhere
else.
The second transition defines $\mu_{k+1}(u_k) =
\mu_k(u_k)[i\leftarrow\delta_\cH(q_\cH, \lab_t(u_k), m_\cH)^*]$ with $m_\cH =
\top_\cH$ if $\mu_k(u_k)$ empty and $\pi_1(\lab_{\mu_k(u_k)}(\bullet_\diamond))$
otherwise, and defines $\mu_{k+1} = \mu_k$ everywhere else. 
Then $\mu_{k+1}$ differs from $\mu_{k-1}$ only over $u_{k-1}$ and $u_k$. Additionally, because $C_0, \ldots C_n$ is a branch-outputting run, $u_k = u_{k-1} d$, and thus node $\bullet_\diamond$ in $\mu_{k+1}(u_{k-1})$ matches the root of $\mu_{k+1}(u_k)$. The relational structure defining $\Stitch(C_{k+1})$ is the same as that of $\Stitch(C_{k-1})$, except we took its disjoint union with the relational structure of tree $\delta_\cH(q_\cH, \lab_t(u_k), m_\cH)^*$, and the root $\memToStitch_{C_{k+1}}(u_k, i, \varepsilon)$ of this $\delta_\cH(q_\cH, \lab_t(u_k), m_\cH)^*$ is the first child of node $\memToStitch_{C_{k+1}} (u_{k-1}, \diamond, \bullet)$, which now has a label in $\downLink$ instead of $\downNewLink$. Because decoding a global state of $\cH$ from a stitched tree forgets the decorating $\cN$ memory symbols and removes the $\downLink$ and $\upLink$ nodes, then we have $\Stitch(C_{k+1}) = \Stitch(C_{k-1})[\memToStitch_{C_{k-1}}(u_{k-1}, \diamond, \bullet) \leftarrow \delta_\cH(q_\cH, \lab_t(u_{k-1}), m_\cH)]$.

We identified a rule of $\cH$ that describes how $\Stitch(C_{k+1})$ is built
from $\Stitch(C_{k-1})$. We now verify that the new configurations in $h_{k+1}$
have correctly defined memories. Let $s'$ be a configuration-labeled node in
$h_{k+1}$, and call $s$ the node $\memToStitch_{C_{k-1}}(u_{k-1}, \diamond,
\bullet)$. Recall that the memory in $\lab_{h_{k+1}}(s')$ is $\nu_{s'}$, and
that for a node $u\in\Dom(t), \nu_{s'}(u)$ is defined as the $\cH$ memory symbol
in the last $\upLink$ symbol of origin $u$ in the branch leading to $s$.
The only new $\downLink$ symbol in $\Stitch(C_{k+1})$ is the symbol $(m_\cH, d,
i)$ from the first transition. And thus, by definition of $\nu^{C_{k+1}}$, for
any $s' \in \Dom(h_{k+1})$, $\nu_{s'} = \nu_{s}$ everywhere, except that
$\nu_{s'}(u_{k-1}) = m_\cH$ if $s' \leq s$.

\end{proof}

\begin{lemma}
     $ \sem{\cN}\circ\sem{\cH} \subseteq \sem{\cG} $
\end{lemma}

    Given a branch-outputting run of $\cN$ on $\sem{\cH}(t)$, we inductively build a branch-outputting run of $\cG$ that has the same output. We verify a two-part induction hypothesis to help us build this run of $\cG$, composed of both a "fidelity" and a "frugality" condition. The fidelity condition states that the run of $\cG$ we're building decodes to exactly the branch-outputting run of $\cN$ we were given, and that outputs are equal, while the frugality condition states that we do this without storing any unnecessary fragment in the memory of $\cG$. This is an important condition to verify, because $\cG$ supports only a bounded amount of fragments in memory, and we need to stay within these bounds. Together with the narrow-visit property of $\cN$, the frugality condition ensures we do. 

    We'll use induction on $n \in \N$ on the following induction hypothesis : for all branch-outputting runs $ C^\cN_0 \xrightarrow{b^\cN_0} C^\cN_1 \xrightarrow{b^\cN_1} \ldots \xrightarrow{b^\cN_{n-1}} C^\cN_n $ of $\cN$ on $\sem{\cH}(t)$, there exists $m \in \N$ and a branch-outputting run $ C^\cG_0 \xrightarrow{b^\cG_0} C^\cG_1 \xrightarrow{b^\cG_1} \ldots  \xrightarrow{b^\cG_{m-1}} C^\cG_m  $of $\cG$ on $t$, such that, noting $C^\cN_i = (u^\cN_i, q^\cN_i, \mu^\cN_i) $, and $C^\cG_i = (u^\cG_i, q^\cG_i, \mu^\cG_i) $ : 
\begin{itemize}
    \item \tikz[remember picture] \node (a1) {Both runs have the same output};
    \item \tikz[remember picture] \node (a2) {$ C^\cG_m $ is decodable};
    \item \tikz[remember picture] \node (a3) {$ \DecodeN(C^\cG_m) = C^\cN_k $};
    \item \tikz[remember picture] \node (b1) {$ \forall u \in \Dom(t), \forall f \in \mu_m^\cG(u), \exists s \in \Dom(f), \exists i \in [n], \Re(s) = u^\cN_i $};
\end{itemize}
\begin{footnotesize}
  \begin{tikzpicture}[overlay, remember picture]
    % First brace (items 1–3)
    \draw[decorate, decoration={brace,amplitude=6pt}, line width=1pt]
    ($(a1.east |- a1.north)+(10ex, 0ex)$) -- ($(a1.east |- a3.south)+(10ex, 0ex)$)
    node[midway, right=6pt] {Fidelity conditions};
    % Second brace (item 4)
    \draw[decorate, decoration={brace,amplitude=6pt}, line width=1pt]
    (b1.east |- b1.north) -- (b1.east |- b1.south)
    node[midway, right=6pt] {Frugality condition};
  \end{tikzpicture}
\end{footnotesize}
Remember we proved in the preceding lemma that all runs of $\cG$ decode to some run of $\cH$, so we are indeed allowed to evaluate $\Re$ on any arbitrary node in memory.
Assuming the hypothesis is true for $n \in \N$. Let $ C^\cN_0
\xrightarrow{b^\cN_0} C^\cN_1 \xrightarrow{b^\cN_1} \ldots
\xrightarrow{b^\cN_{n}} C^\cN_{n+1} $ be a branch-outputting run of $\cN$, and $
C^\cG_0 \xrightarrow{b^\cG_0} C^\cG_1 \xrightarrow{b^\cG_1} \ldots
\xrightarrow{b^\cG_{m-1}} C^\cG_m  $ the branch-outputting run of $\cG$ built by
our induction hypothesis, applied to $ C^\cN_0
\xrightarrow{b^\cN_0} C^\cN_1 \xrightarrow{b^\cN_1} \ldots
\xrightarrow{b^\cN_{n-1}} C^\cN_{n} $.
Let us consider the transition $ C^\cN_n \xrightarrow{b^\cN_{n}}
    C^\cN_{n+1} $. We know by induction hypothesis that $ \DecodeN(C^\cG_m) =
    (u^\cN_n, q^\cN_n, \mu^\cN_n) $. Because of the definition of $\DecodeN$,
    this means that $ \pi_1(\lab_{\mu(u^\cG_m)}(\bullet_\diamond)) = \lab_t(u^\cN_n) $,
    $ q^\cG_m = q^\cN_n $, and $ \pi_2(\lab_{\mu(u^\cG_m)}(\bullet_\diamond)) =
    \mu^\cN_n(u^\cN_n) $. By definition of $\cN\shortrightarrow\cN$ transitions
    in $\cG$, this means that $\delta_\cG( q^\cG_m, \lab_t(u^\cG_m),
    \mu^\cG_m(u^\cG_m)) = \delta_\cN( q^\cN_n, \lab_{\sem{H}(t)}(u^\cN_n),
    \mu^\cN_n(u^\cN_n)$. And thus, letting $b^\cG_m = b^\cN_n$ and $C^\cG_{m+1}
    = C^\cG_m \ogreaterthan \delta_\cG( q^\cG_m, \lab_t(u^\cG_m),
    \mu^\cG_m(u^\cG_m))[b_m^\cG]$, we verify the first fidelity condition. 
    
    Now, there are three main possibilities, depending on the label $
    \pi_1(\lab_{\mu^\cG_{m+1}(u^\cG_{m+1})}(\bullet_\diamond))$ of the
    pointed node. 

    First case, if that label is not a $\Link$, we immediately know that
    $C^\cG_{m+1}$ is decodable.
    We prove the third fidelity condition for this first case. Because of the
    definition of $\cN\shortrightarrow\cN$ transitions, we know that
    $q_{i+1}^\cG = q_{k+1}^\cN$. Again by definition, the memory symbol had its
    pointed second label replaced with $\mu_{k+1}^\cN(u_k^\cN)$, and then its
    point was moved along the direction indicated in the transition of $\cN$.
    Thus we conclude $\DecodeN(C_{m+1}^\cG) = C_{k+1}^\cN$
    Because we have not produced any new fragments in memory, the frugality
    condition at step $m+1$ directly follows from that of time $m$.
    our induction hypothesis.
    This concludes the first case 
    $ \pi_1(\lab_{\mu^\cG_{m+1}(u^\cG_{m+1})}(\bullet_\diamond)) \notin \Link$.

    Second case, if $ \pi_1(\lab_{\mu^\cG_{m+1}(u^\cG_{m+1})}(\bullet_\diamond)) \in \downLink
    \cup \upLink$, we prove that $\cG$ takes a serie of transitions through the
    $\cG$-machinery regime, but eventually lands in a decodable configuration. 
    This corresponds, in the memory of $\cG$, to moving through fragments
    composed only of $\upLink$ and $\downLink$ symbols. Those "trivial"
    fragments correspond to transitions of $\cH$ that did not produce any
    output.
    From $C^\cG_{m+1}$, $\cG$ takes a $\cN\shortrightarrow\cG$,
    then a $\cG\shortrightarrow\cN$ transition. These transitions modify
    neither the memory nor the output, so we can define $\mu^\cG_{m+2} =
    \mu^\cG_{m+3} = \mu^\cG_{m+1}$ and $b^\cG_{m+1} = b^\cG_{m+2} =
    \varepsilon $ to preserve the first fidelity condition. If after these
    transitions, $C^\cG_{m+3}$ still is not decodable, then we repeat the
    process until we reach a decodable configuration. We can repeat,
    because being undecodable means we must still be pointing to a $\Link$. It
    cannot be a $\downNewLink$, otherwise the current fragment would be trivial (contain
    only links). Having a $\downNewLink$ would then imply that it is a leaf in
    the stitched tree, violating the frugality hypothesis.
    We now prove that repeatedly taking $\cN\shortrightarrow\cG,
    \cG\shortrightarrow\cN$ transiitons eventually terminates. All links pointed after taking a
    $\cG\shortrightarrow\cN$ transition must go in the same
    direction. If any link is in $\downLink$ (resp. $\upLink$),
    then we enter the next fragment from an $\upLink$ (resp. $\downLink$),
    and the following $ \cG\shortrightarrow\cN$ transition 
    moves us away, to a node with a parent (resp. child), which
    thus has to, like the first one, be a $\downLink$ (resp. $\upLink$) (if it
    is still a $\Link$).
    Finally, each transition is a step in the stitched tree, which is finite and constant
    along these transitions. The up case is easy, and we eventually reach, in
    the worst case, the root of the stitched tree which is not a $\Link$. The
    down case is covered by the frugality condition, which states that we know
    we will eventually reach a node that is not a link, because it is in the
    domain of $\Re$.
    We prove the third fidelity condition similarly to the first case for
    equality of memory and state in the decoded configuration of $\cN$. 
    Equality of reading heads is subtler. We did begin by moving the pointed
    node along some direction $j$ indicated in the transition of $\cN$. Because we saw only $\Link$
    symbols along the way, the pointed node of $C_{m+k}^\cG$ and of
    $C_{m+k}^\cG$ may not map (with respect to $\memToStitch$) to a parent and its $j$-th child in
    $\Stitch(C_{m+k}^\cG)$, but they do in the final
    decoded configuration of $\cH$, which can be built by removing $\Link$
    symbols from $\Stitch(C_{m+k}^\cG)$, and thus they do (with respect to
    $\Re$) in $\sem{H}(t)$ by inclusion.
    We still have not produced any new fragment in memory, so the frugality
    condition again follows directly from induction.
    This concludes the second case.

    The $\cH$-simulating case happens when $
    \pi_1(\Point(C^\cG_{m+1})) $ is in
    $\downNewLink$. We will take a $\cN\shortrightarrow\cH,
    \cH\shortrightarrow\cN$ transition pair, and reach a configuration $C\cG_{m+3}$,
    following branches $b^\cG_{m+1} = b^\cG_{m+2} = \varepsilon $. In the same
    way, we repeat the process until it stops. We reach a decodable
    configuration because each pair of transitions corresponds to a step in the
    simulated run of $\cH$, which cannot loop indefinitely. Additionally, there
    is no empty tree, so we eventually produce a fragment with a non-$\Link$
    node. As soon as we do, we are back in a decodable configuration
    $C^\cG_{m+k}$. 
    The proof of the third fidelity condition is similar to the proof for the
    previous case.
    We did produce new fragments in memory, so we need to verify
    the frugality condition. But we just proved that $ \DecodeN(C^\cG_{m+k}) =
    C^\cN_{n+1}$. In particular, $ \Re(\Point(C^\cG_{m+k})) = u^\cN_{n+1} $,
    which means that $ \Point(C^\cG_{m+k}) $ satisfies the frugality condition
    for all the fragments we just added.
    This concludes the last case.

\begin{lemma}
    $ \sem{\cG} \subseteq \sem{\cN}\circ\sem{\cH} $
\end{lemma}

\begin{proof}
    We reuse the two previous lemmas.
    Let $ R_\cG = C^\cG_0 \xrightarrow{b^\cG_0} C^\cG_1 \xrightarrow{b^\cG_1} \ldots
    \xrightarrow{b^\cG_{n-1}} C^\cG_n $ be a total branch-outputting run of $\cG$ on
    some input tree $t$, along some branch $b$.
    We will prove by induction on the length of a prefix of this run that there
    exists an equivalent run of $\cN$ on $\sem{H}(t)$.

    The run $ R_\cG $ up to its first decodable state is equivalent to the
    initial configuration of $\cN$ on $\sem{H}(t)$, by the previous lemma.

    Assume there exists $k \in [n-1]$ such that $ C^\cG_k = (q_\cG, u_\cG,
    \nu_\cG)$ is decodable, and $ C^\cG_0 \xrightarrow{b^\cG_0} C^\cG_1
    \xrightarrow{b^\cG_1} \ldots \xrightarrow{b^\cG_{k-1}} C^\cG_k $ is
    equivalent to a branch-outputting run $ R_\cN = C^\cN_0
    \xrightarrow{b^\cG_0} C^\cN_1 \xrightarrow{b^\cN_1} \ldots
    \xrightarrow{b^\cN_{m-1}} C^\cN_m $ of $\cN$ on $\sem{H}(t)$. Because
    $C^\cG_k$ is decodable, it must be the case that $\lab(\Point(C^\cG_k))$ is
    some pair $\pp{\gamma, m_\cN}$, and $q_\cG \in Q_\cN$. Then the transition
    from $C^\cG_k$ to $C^\cG_{k+1}$ must be a $\cN\shortrightarrow\cN$
    transition. By definition, there exists a corresponding rule
    $\delta_\cN(q_\cG, \gamma, m_\cN)$ of $\cN$, and thus we can define a
    configuration successing $C^\cN_{m+1}$ by following branch $b$.  By the
    previous proof, this new run $ C^\cN_0 \xrightarrow{b^\cG_0} C^\cN_1
    \xrightarrow{b^\cN_1} \ldots \xrightarrow{b^\cN_{m}} C^\cN_{m+1} $ accepts
    an equivalent run of $\cG$ on $t$, and by determinism of $\cG$ and totalness
    of $R_\cG$, this equivalent run is a prefix of $ R_\cG $.
    We conclude by induction on $k$.
\end{proof}

\end{document}